%% file: main.tex
\documentclass[11pt]{extarticle}
\input{macros}

\begin{document}

\input{title}
\input{intro}

\input{protocols}

\input{upper}
\input{spectral}

\input{lower_bound_general_discrepancy}

\input{onewayparities}

\input{related}
\bibliographystyle{alpha}
\bibliography{refs}

\appendix

\input{appendix}

\newpage

\end{document}

%% file: macros.tex
\usepackage[utf8]{inputenc}
\usepackage{amsmath, amssymb,amsfonts, amsthm, bbm, graphicx, url, hyperref, setspace, mathtools}
\usepackage[capitalize]{cleveref}

\hypersetup{colorlinks}
\usepackage{float}
\usepackage{graphicx}
\usepackage{bm}
\usepackage{amssymb}
\usepackage{bbm}
\usepackage{xcolor}
\usepackage{algorithm}
\usepackage{dsfont}
\usepackage{xspace}
\usepackage{comment}
\usepackage{tikz}
\usepackage{dsfont}
\usepackage{fullpage}
\newcommand*\samethanks[1][\value{footnote}]{\footnotemark[#1]}

\newcommand{\eat}[1]{}

\newcommand{\ip}{\mathsf{IP}}
\newcommand{\rip}{\mathsf{IP}_{\mathbb{R}}}

\newcommand{\iprod}[2]{{\langle #1, #2 \rangle}}

\newcommand{\mR}{\mathbb{R}}
\newcommand{\tdO}{\tilde{O}}

\newcommand{\Jee}{Distributed estimation\xspace}
\newcommand{\jee}{distributed estimation\xspace}

\newcommand{\varred}{debiasing\xspace}

\newcommand{\X}{\mathcal{X}}
\newcommand{\Y}{\mathcal{Y}}

\newcommand{\mD}{\mathcal{D}}

\newcommand{\mL}{\mathcal{L}}

\DeclareMathOperator{\Var}{Var}

\newcommand{\F}{{\mathbb{F}}}
\newcommand{\RR}{{\mathbb{R}}}

\newcommand{\lt}{\left}
\newcommand{\rt}{\right}

\newcommand{\zo}{\ensuremath{\{0,1\}}}

\newcommand{\Id}{\mathrm{I}}

\newcommand{\norm}[1]{\left\lVert#1\right\rVert}

\newcommand{\D}{\mathrm{D}}
\newcommand{\R}{\mathrm{R}}
\newcommand{\Dow}{\mathrm{D}^{\mathrm{ow}}}
\newcommand{\Row}{\mathrm{R}^{\mathrm{ow}}}
\newcommand{\Rhow}{\bar{\mathrm{R}}^{\mathrm{ow}}}
\newcommand{\Dhow}{\bar{\mathrm{D}}^{\mathrm{ow}}}
\newcommand{\Dhat}{\bar{\mathrm{D}}}
\newcommand{\Rhat}{\bar{\mathrm{R}}}

\newcommand{\Unif}{\mathrm{Unif}}

\newcommand{\eps}{\varepsilon}

\newcommand{\nbx}[0]{t_X}
\newcommand{\nby}[0]{t_Y}

\newcommand{\pmo}{\ensuremath{ \{\pm 1\} }}
\newcommand{\fr}[1]{\ensuremath{\frac{1}{#1}}}

\newcommand{\ind}[1]{\ensuremath{\mathbbm{1}(#1)}}

\newcommand{\gh}{\mathrm{gapH}}

\newcommand{\ab}{\mathrm{Abs}}

\newcommand{\f}{\frac}
\newcommand{\p}[1]{\left( #1 \right)}
\newcommand{\e}{\eps}
\newcommand{\Ot}{\tilde{O}}
\newcommand{\inv}[1]{{#1}^{-1}}
\newcommand{\al}{\alpha}

\DeclareMathOperator{\poly}{poly}

\DeclareMathOperator{\rank}{rank}

\DeclareMathOperator*{\E}{\mathbf{E}}

\newcommand{\EQ}{\mathsf{EQ}}
\newcommand{\GT}{\mathsf{GT}}
\newcommand{\ABS}{\mathsf{Abs}}
\newcommand{\EE}{\mathrm{EE}}
\newcommand{\DI}{\mathsf{DI}}

\newcommand{\tp}{\tilde{p}}
\newcommand{\tq}{\tilde{q}}

 \newcommand{\PG}[1]{\pgcomment}

\newcommand{\sfT}{\mathsf{T}}

\newcommand{\fh}{\hat f}

\newcommand{\ignore}[1]{}

\newtheorem{theorem}{Theorem}[section]
\newtheorem{lemma}[theorem]{Lemma}
\newtheorem{definition}[theorem]{Definition}

\newtheorem{corollary}[theorem]{Corollary}

\newtheorem{fact}[theorem]{Fact}
\newtheorem{claim}[theorem]{Claim}
\theoremstyle{definition}
\newtheorem{remark}[theorem]{Remark}

\Crefname{claim}{Claim}{Claims}

\newcommand{\cD}{\mathsf{D}}

\newcommand{\del}{\partial}

%% file: title.tex
\title{The communication complexity of \jee}
 \author{Parikshit Gopalan\\
Apple \and 
Raghu Meka\thanks{Work done while consulting at Apple.}\\
UCLA \and
Prasad Raghavendra\samethanks\\
UC Berkeley \and 
Mihir Singhal\\
UC Berkeley \and
Avi Wigderson\\
IAS
}

\date{}

\maketitle
\begin{abstract}
We study an extension of the standard two-party communication model in which Alice and Bob hold probability distributions $p$ and $q$ over domains $\mathcal{X}$ and $\mathcal{Y}$, respectively. Their goal is to estimate
\[
\mathbb{E}_{x \sim p,\, y \sim q}[f(x, y)]
\]
to within additive error $\varepsilon$ for a bounded function $f$, known to both parties. We refer to this as the \jee problem. Special cases of this problem arise in a variety of areas including sketching, databases and learning. Our goal is to understand how the required communication scales with the communication complexity of $f$ and the error parameter $\varepsilon$.

The random sampling approach---estimating the mean by averaging $f$ over $O(1/\varepsilon^2)$ random samples---requires $O(\R(f)/\varepsilon^2)$ total communication, where $R(f)$ is the randomized communication complexity of $f$. We ask if and when one can do better.

\begin{itemize}
    \item \textbf{Debiasing protocol.} We design a new {\em \varred} protocol for arbitrary bounded functions that gives  better variance reduction than random sampling, and requires communication linear in $1/\eps$ rather than quadratic.
    
    \item \textbf{Better protocols for specific functions.} For the equality ($\mathsf{EQ}$) and greater-than ($\mathsf{GT}$) functions, we obtain an upper bound of $\tilde{O}(\R(f)/\varepsilon^{2/3})$.\footnote{$\tilde{O}$ notation hides polylogarithmic factors in $1/\eps$.} When $f:[0,1]^2 \to [0,1]$ is Lipschitz, the communication can be reduced to $\tilde{O}(1/\sqrt{\varepsilon})$, with even stronger bounds under higher smoothness assumptions.

    \item \textbf{Spectral lower bounds.} We develop a novel spectral method to lower bound the $\varepsilon$-dependence for arbitrary $f$, which yields tight bounds for all the functions above. For any full-rank Boolean matrix, we show that the communication complexity is at least $\tilde{\Omega}(1/\varepsilon^{2/3})$. Hence $\EQ$ is essentially the easiest full-rank Boolean function for \jee.

    \item \textbf{Discrepancy-based lower bounds.} For any Boolean function $f$ with sufficiently high discrepancy, we prove a lower bound of $\Omega(\R(f)/\varepsilon)$, building on machinery from query-to-communication lifting.

    \item \textbf{Random sampling versus \varred.}  Our \varred protocol requires communication $\tilde{O}(\Row(f)/\varepsilon)$, where $\Row(f)$ bounds the one-way communication cost of $f$ in either direction, compared to $\tilde{O}(\R(f)/\eps^2)$ for random sampling. We construct a function where the upper bound $\min(\R(f)/\varepsilon^2,\, \Row(f)/\varepsilon)$ is tight, ruling out $O(\R(f)/\varepsilon)$ as an upper bound.
\end{itemize}

\end{abstract}

\thispagestyle{empty}
\newpage

{
  \hypersetup{linkcolor=black}
  \tableofcontents
}

\thispagestyle{empty}
\newpage
\setcounter{page}{1}

%% file: intro.tex
\section{Introduction}

Communication complexity measures the number of bits transmitted between two or more players for computing a joint function of their inputs.
In the standard model \cite{Yao79}, Alice holds an input $x$ and Bob holds an input $y$, and their goal is to compute the value of a function $f(x,y)$ over their inputs.
Different settings of the problem are obtained by permitting the players to use deterministic or randomized protocols, communicate in one, two, or an arbitrary number of rounds.  

In this work, we investigate a family of communication problems that we term as {\it \jee} problems, that extend the standard model, and are defined as follows.

\begin{definition}[\Jee for $f$]
Fix a function $f: \X \times \Y \to [-1,1]$ known to both players. Alice's input is a  probability distribution $p$ over the set $\X$, while Bob's input is a probability distribution $q$ over the set $\Y$.  The players' goal in \jee is to estimate $\E_{p,q}[f] = \E_{x \sim p, y \sim q}[f(x,y)]$ within an additive error of $\varepsilon$.
\end{definition}

\eat{
Alice is given access to samples of a random variable $x$ over a set $\X$, and Bob is given access to samples a random variable $y$ over a set $\Y$.  
There is a function $f: \X \times \Y \to [-1,1]$ known to both players.
In the {\it \jee} problem for $f$, the goal of the players is to estimate the expected value $\E[f(x,y)]$ where the random variables $x$ and $y$ are independent of each other.
}
A deterministic communication protocol $\Pi$ for the \jee problem for $f$ should estimate the value of $\E_{p,q}[f]$ within an error $\varepsilon$, for every pair of input distributions $p$ and $q$.  
Analogously, a randomized communication protocol is required to compute $\E[f(x,y)]$ with high probability over random coins of Alice and Bob. We will assume throughout that the players have access to shared randomness.

The primary focus of this work will be communication complexity, as opposed to sample  or computational efficiency. There is a modeling question of how Alice and Bob can access their respective random variables. The most restrictive model might be to assume that they can draw repeated samples from their respective distributions. We have assumed the stronger access model where they each know the probability distributions of their respective inputs. It turns out that all the protocols we devise only require sample access (and have small sample complexity as well); while the lower bounds hold even when Alice and Bob know their complete distributions, which makes them stronger. 
%
%

\paragraph{Distributional communication complexity, direct sum.}

At first glance, \jee may seem reminiscent of distributional communication complexity under rectangular distributions, where Alice and Bob receive inputs $x \in \X$ and $y \in \Y$ drawn independently from respective distributions $p$ and $q$. However, the two settings are fundamentally different. For simplicity, let us take $f$ to be Boolean. In distributional complexity, we study (deterministic) protocols $\Pi : \X \times \Y \to \zo$ that attempt to compute $f$ correctly on random inputs from a distribution $\mD$ on $\X \times \Y$. The goal is to construct $\mD$ such that any protocol $\Pi$ that computes $f$ correctly on most inputs from $\mD$ must use large communication.\footnote{It is known that product distributions can be strictly easier than general ones \cite{KN97}.}

In contrast, in \jee, protocols take distributions themselves as inputs, and the goal is to estimate the expectation of $f$ under the product distribution $\mD = p \times q$.  The standard communication problem of computing $f(x,y)$ given $x$ and $y$ as inputs is a special case of the \jee problem, wherein the distributions $p$ and $q$ are supported on single elements $x$ and $y$. Therefore, protocols for \jee of $f$ are also valid protocols for computing the function $f$. In the other direction, we can start from a (possibly randomized) protocol $\Pi$ that computes $f$, run it independently on $O(1/\eps^2)$ samples from $\mD$, and output the empirical mean. If computing $f$ on each individual copy takes $C$ bits of communication, we get an easy protocol with cost $O(C/\eps^2)$. 

Can one do better than this? This raises a direct-sum type question: does solving \jee\ essentially require correctly computing $f$ on multiple individual inputs drawn from $p \times q$? Answering this question is the starting point of our work.  
Before discussing our results in more detail, let us motivate why we believe the communication complexity of \jee merits deeper understanding. 
%


%
%
%

%
\subsection{Motivation}

Various special cases of \jee, corresponding to specific functions or restricted communication models have arisen and been studied in the literature, in contexts as diverse as sketching, join estimation for databases, computational learning and two player zero-sum games. 

\paragraph{Sketching.}
         Fix $\X = \Y = \{1,\ldots,n\}$ and the function to be the equality function $\EQ(x,y) = 1[x = y]$.
    If Alice holds a prob. distribution $p = (p_1,\ldots,p_n)$ over $\{1,\ldots,n\}$ and Bob holds a prob. distribution $q = (q_1,\ldots,q_n)$ then their goal is to estimate the quantity (\emph{collision probability}),
    \[ \E_{p,q}[\EQ] =  \sum_i p_i q_i\]
    to an additive error of $\varepsilon$. The literature on sketching and streaming considers the problem of computing the inner product of $p, q$, but in the simultaneous or one-way communication models, with various assumptions about the norms of $p, q$. Count-Min Sketch \cite{CormodeMuthukrishnan} gives a simultaneous protocol that uses space $\tilde{O}(1/\eps)$ to get an $\eps\|p\|_1\|q\|_1$ additive approximation. The Johnson-Lindenstrauss sketch \cite{JL84, Achlioptas03}, the AMS sketch \cite{AlonMatiasSzegedy-JCSS99} and CountSketch \cite{CharikarChenFarachColton-ICALP02} all give $\tilde{O}(1/\eps^2)$ simultaneous protocols whose error is $\eps\|p\|_2\|q\|_2$. The \jee problem for Equality asks how the communication complexity changes when we allow multiple rounds. 
    
    From the lower bounds side, the well-studied problems of set-disjointness \cite{Razborov92} and gap-Hamming \cite{ChakrabartiR12} are concerned with the communication complexity of the quantity $\sum_i p_i q_i$, where $p, q$ are bounded in $\ell_\infty$.

    \paragraph{$P$-join estimation.}
    Join-estimation is a fundamental problem in databases, see \cite{van2015communication} and the references therein. Given two relations $\mathcal{A} \in [m] \times [n]$ and $\mathcal{B} \in [n] \times [p]$, 
        the {\it composition} and {natural join} of the two relations is defined as,
        \begin{align*}
        \text{(Composition)\ \ \ }  \mathcal{A} \circ \mathcal{B} & = \{ (i,j) | \exists \ell : (i,\ell) \in \mathcal{A} \text{ and } (\ell,j) \in \mathcal{B}\} \\
        \text{(Natural join)\ \ \ } \mathcal{A} \bowtie \mathcal{B} & = \{ (i,\ell,j) | \exists \ell : (i,\ell) \in \mathcal{A} \text{ and } (\ell,j) \in \mathcal{B}\} 
        \end{align*}
        
        The communication complexity of computing these joins and their the sizes was studied in \cite{van2015communication}.  The problem of estimating $| \mathcal{A} \circ \mathcal{B}|$ to within additive error $\eps |\mathcal{A}||\mathcal{B}|$ is an instance of \jee for $\EQ$. Indeed join estimation is one of the applications of CountMinSketch \cite[Corollary 1]{CormodeMuthukrishnan}. Similarly, estimating the size of the composition join is \jee for the Set-disjointness function. See Section \ref{sec:related} for details. 

        More generally, for any predicate $\mathcal{P}$ on pairs of subsets of $[n]$,  \cite{van2015communication} define the $\mathcal{P}$-joins of two relations as
        \[ \mathcal{A} \bowtie^{\mathcal{P}} \mathcal{B} =  \{ (i,j) | \mathcal{P}(A_i,B_j) \text{ is true for } i \in [m], j \in [p] \}\]
        Estimating the size of $\mathcal{A} \bowtie^{\mathcal{P}} \mathcal{B}$ is equivalent to \jee for the function $\mathcal{P}$. Our parametrization of the error is additive, and differs from \cite{van2015communication} who consider multiplicative guarantees, which can be significantly harder (see Section \ref{sec:related} for a detailed comparison).

\paragraph{Real-valued functions and omnipredictors.} 

The task of \jee is equally natural for real-valued functions and over continuous domains. A particularly illustrative case is when $p$ and $q$ are distributions on a metric space and the goal is to estimate the expected distance between independent draws from $p$ and $q$. For instance, for distributions on $[0,1]$, the $\ell_1$–\jee problem is to estimate $\E_{x\sim p,y\sim q}\big[|x-y|\big]$. Such problems were studied in the work of \cite{GORSS24} on omnipredictors for regression.

Omnipredictors, introduced in \cite{omni}, and subsequently studied in several works including \cite{lossOI, GopalanKR23, GargJRR24, GORSS24} are predictors guaranteed to perform well for every loss in a rich family $\mL$, rather than being tailored to a single loss. While the initial works treated binary classification, \cite{GORSS24} study regression with real-valued labels in $[0,1]$. They analyze the following one-way communication problem whose complexity governs the algorithm's runtime: Alice holds a random variable $x\in[0,1]$, Bob holds a bounded loss function $\ell\in\mL$. Alice must send a message enabling Bob to estimate $\E[\ell(x)]$ to within $\eps$ for all $\ell\in\mL$.\footnote{Their algorithm requires some additional {\em linearity} constraints on the one-way protocol, see \cite{GORSS24} for details.}

This is an instance of \jee with rows indexed by $[0,1]$ and columns by $\ell\in\mL$. A key result of \cite{GORSS24} is a protocol with communication $O(1/\eps^{2/3})$ when $\mL$ is the class of bounded convex losses. This case is equivalent to the $\ell_1$-\jee example above, using the fact that the functions ${|x-t|}_{t\in[0,1]}$ span bounded convex 1-Lipschitz functions on $[0,1]$. Other choices of $\mL$ similarly give rise to \jee for different function classes.

    \paragraph{Zero-sum games.}
    Consider a two player zero-sum game where $\X$ and $\Y$ are the strategies available to two players Alice and Bob and $f(x,y)$ is payoff from Alice to Bob. \Jee for $f$ is the problem of estimating the payoff for a pair of strategies $p, q$. For instance, let $\X = \Y = \{1,\ldots,n\}$ and consider the function $\GT(x,y) = 1[x \geq y]$. The associated \jee problem is to estimate the probability $\Pr[ x \geq y]$. 

    \paragraph{Estimating edge density.}
    Say that Alice holds a subset $A$ and Bob holds a subset $B$, of vertices in a graph $G = (V,E)$.  Their goal is to estimate the density of edges between sets $A$ and $B$. Replacing each player's input with the uniform distribution over their sets $A,B$, and using the function $f(i,j) = 1[(i,j) \in E]$, the  \jee problem is to estimate the edge density $|E(A,B)|/(|A||B|)$.
\medskip

    \paragraph{Summary.}
    Instantiations of \jee\ have been analyzed in numerous settings, often under specific protocol restrictions or for particular function families, but a general theory was missing. We discuss related work further in Section \ref{sec:related}. In this work we initiate a systematic study of the communication complexity of \jee, aiming to relate it to the complexity of $f$, the number of rounds, and the allowed error.



\subsection{Our results}

\subsubsection{Preliminaries and notation}
To state our results, we begin by reviewing some notation. 
For a function $f: \X \times \Y \to [-1,1]$, the minimum number of bits exchanged by a deterministic protocol computing the function $f(x,y)$ is classically denoted $\D(f)$.  The communication complexity of randomized protocols that err with probability at most $\delta$ is denoted by $\R_{\delta}(f)$. We also use $\Dow$ and $\Row$ to denote the corresponding complexities when the communication is restricted to be one-way. For one-way protocols, we consider the case where each of Alice and Bob speaks first, and take the maximum of the two. We assume shared randomness in all our protocols.

\begin{definition}
For a function $f$, let $\EE_{f,\varepsilon}$ denote the task of distributed estimation of $f$ with an additive error $\varepsilon$.
$\Dhat_{\varepsilon}(f)$ is the communication complexity of deterministic protocols for $\EE_{f,\varepsilon}$.  Similarly, we use $\Rhat_{\varepsilon, \delta}$ for communication complexity of randomized protocols that estimate $f$ within an error $\varepsilon$ with probability at least $1-\delta$.
\end{definition}

For simplicity, we often omit $\varepsilon, \delta$ and simply write $\Dhat(f)$ and $\Rhat(f)$ when they are clear from context (when not otherwise defined, $\delta$ will be assumed to be any constant less than $1/2$, such as $1/3$).
Moreover, we will similarly use $\Dhow$ and $\Rhow$ when the communication is restricted to be one-way. 
These measures are functions of two parameters -- the size of the inputs often parametrized as $n = \log_2 (| \X | + |\Y|)$ and the accuracy $\varepsilon$.  The dependence on input size is often borrowed from the traditional communication problem of computing $f$, while the accuracy $\varepsilon$ is the new ingredient in \jee.  Therefore, although we will state and prove upper and lower bounds in terms of both these parameters, the dependence on $\varepsilon$ would be our primary interest.

We also write $A_f$ to denote the matrix of values of $f$. In other words, the $(x, y)$-th entry of $A_f$ is $f(x,y)$ for $x \in \X, y \in \Y$.

We use $\Ot$ notation to hide logarithmic factors in $1/\e$ (but not in any other parameters).

\subsubsection{Protocols for \jee}

\paragraph{Beating random sampling.}

The obvious approach to estimating $\E_{p,q}[f]$  is to use random sampling: Alice samples $t$ inputs  $x_1,x_2,\ldots,x_t$ from distribution $p$, and Bob samples $t$ inputs $y_1,\ldots,y_t$ from his distribution $q$.  Using a randomized protocol, both Alice and Bob compute $f(x_i,y_i)$ for each $i = 1,\ldots, t$, and output the estimate
\[
\hat{\mu}\;=\;\frac{1}{t}\sum_{i=1}^t f(x_i,y_i).
\]
By Chebyshev’s inequality, taking $t=\Theta(1/\varepsilon^2)$ ensures $|\hat{\mu}-\E_{p,q}[f]| \le \varepsilon$ with probability at least $0.99$. Since evaluating $f$ on $t$ independently sampled pairs costs $t \cdot \R(f)$ bits, this yields a \jee\ protocol with communication $\Theta(\R(f)/\varepsilon^2)$.

It is tempting to conjecture that this sampling bound is generically optimal. For instance, for a Boolean $f:\{0,1\}^n \times \{0,1\}^n \to \{0,1\}$ with maximal randomized complexity $\R(f)=n$ (e.g., a random function or the inner-product function), one might expect $\Rhat_\varepsilon(f)=\Omega(n/\varepsilon^2)$.

Surprisingly, this intuition is false. For every bounded function $f$, there is a two-round ``debiasing'' protocol whose communication is linear in $1/\varepsilon$! More precisely we show that if $\Row(f)$ bounds the one-way communication complexity of $f$ in either direction, then
    \[\Rhat(f) = \tilde{O}\left(\frac{\Row(f)}{\varepsilon}\right) =  \tilde{O}\left(\frac{n}{\varepsilon}\right)   \qquad \qquad (\text{see \cref{thm:mainupper}})\]

Direct sum theorems in communication complexity often imply that for computing $t$ instances of a function $f$, the naive approach of running the protocol $t$ times is optimal. The above result stands in contrast to these direct sum theorems, in that the naive approach of computing $f$ independently for $t = \Theta(1/\varepsilon^2)$ separate instances is not optimal, even for a general function $f$. 

\paragraph{Protocols for arbitrary bounded functions.} 
The problem $\EE_{f, \varepsilon}$ has two parameters, the input size $n = \log|\X| + \log|\Y|$ and the error parameter $\eps$ (we consider constant failure probability). 
Focusing on the dependence on $\varepsilon$, for any given function $f$, our work demonstrates a tradeoff between the dependences on $\varepsilon$ and the input size $n$, through the following results:
\begin{itemize}
    \item A random sampling protocol (see \cref{thm:var-red1}) that has communication $\tilde{O}(\R(f)/\eps^2)$.
    \item A \varred protocol (see \cref{thm:mainupper}) that has communication $\tilde{O}(\Row(f)/\eps)$.
    \item A spectral protocol (see \cref{thm:spectral-upper}) that has communication $\tdO(\sigma(f)/\eps^{2/3})$, where $\sigma(f)$ is the spectral norm of (the matrix of) $f$.
    \item A deterministic protocol (see \cref{thm:svd-protocol}) that has communication $\tdO(\rank(f) \cdot n)$.
\end{itemize}

To compare the numerators, we observe that $\R(f) \leq \Row(f) \leq n$ whereas $\sigma(f) \in[0,2^n]$ and $\rank(f) \in \{0, \ldots, 2^n\}$. The last two bounds as stated are incomparable. Of these protocols, (1) and (4) are folklore/immediate, while (2) and (3) are new contributions of this work.

\paragraph{Protocols for specific functions.}

We develop protocols for \jee for specific functions $f$ using a varied set of tools including hashing, sketching, variance reduction via sampling and singular value decompositions. Here is a summary of our upper bounds:

\begin{itemize}

\item A simple randomized protocol for the equality function $\EQ(x, y) = \ind{x = y}$ with communication $\tilde{O}(1/\eps^{2/3})$ (see \cref{thm:eq-upper}). We extend this upper bound to all matrices of bounded spectral norm.

\item A protocol for the greater-than function $\GT(x, y) = \ind{x \ge y}$ which has communication $O((\log n) /\varepsilon^{2/3})$ (see \cref{thm:gt-upper}).

\item Protocols with communication $O(1/\varepsilon^{2/5})$ for the function $\ABS(x,y) = |x-y|$ on $[-1,1]$ (see \cref{thm:abs-upper}). 

\item For functions $f: [0,1] \to [-1,1]$ whose derivatives of $k$-th order are bounded, we show (see \cref{thm:smooth-upper}) that
\[ \Rhat_\varepsilon(f) \leq O( \varepsilon^{-1/(k+1)})\] 

\end{itemize}
We also develop lower bound techniques to show that all of the above protocols are optimal!

\subsubsection{Lower bounds}

In this work, we establish lower bounds on the randomized complexity of \jee for both Boolean and real-valued bounded functions. Our main results are the following:
\begin{itemize}
    \item {\bf Discrepancy based lower bound.} There exist functions $f:\zo^n \times \zo^n \to [-1,1]$, where \jee requires $\Omega(n/\eps)$ bits of communication (see \cref{cor:ipcor}). Thus the \varred protocol cannot be improved in general. This is a consequence of a general lower bound showing that high \emph{discrepancy} makes a function maximally hard for \jee.

    \item {\bf Random sampling versus \varred.} 
    Compared to random sampling, \varred reduces the dependence on $1/\eps$ from quadratic to linear. However the dependence on the randomized complexity of $f$ is worse: \varred (Theorem \ref{thm:mainupper}) gives a bound of $O(\Row(f)/\varepsilon)$ bits, whereas random sampling gives $O(\R(f)/\eps^2)$. In the worst-case, $\Row(f)$ can be higher than $\R(f)$. We show that one cannot get a {\em best-of-both worlds} upper bound of $\R(f)/\eps$ (see \cref{thm:di-lb}).

    \item {\bf Bounds for generic Boolean functions.} Every $n$-bit Boolean function whose communication matrix has large rank requires $\tilde{\Omega}(1/\eps^{2/3})$ communication for \jee (see \cref{thm:boolean-lb}). In particular, among full-rank Boolean functions, none is asymptotically easier than $\EQ$ and $\GT$. 

    \item {\bf Lower bounds for specific functions.} We show that all of the upper bounds for specific functions mentioned in the previous section are tight (see \cref{sec:specific-lower}). 

    \item {\bf One-way lower bound for $\ip$.} For the function $\ip:\{0,1\}^n \times \{0,1\}^n \rightarrow [-1,1]$ defined by $\ip(x,y) = (-1)^{\langle x, y \rangle}$, we give a tight lower bound of $\Rhow_{\eps, \delta}(\ip) = \Omega( n\log(1/\delta)/\eps^2)$ on the one-way communication complexity (see \cref{cor:ipcor}).
\end{itemize}

\subsection{Overview of lower bounds}

In this subsection, we elaborate the above lower bounds and state them precisely.
%

\paragraph{A discrepancy-based lower bound.}

Classical results relate the communication complexity of a function to various structural measures associated with the function such as discrepancy, partition number, smooth rectangle bound and information complexity.
It is natural to ask whether one can use any of these measures to obtain lower bounds for \jee  that also have a dependence on the error $\varepsilon$.
We show a lower bound on the communication complexity of \jee of a function $f$, in terms of its worst case discrepancy $\mathrm{Disc}^{\times}(f)$ under a product distribution. Roughly speaking, the discrepancy of a function measures how unbalanced the function can be on a large rectangle in its matrix (see Definition \ref{def:rect-disc} for formal definition).

\begin{theorem}[Informal version of \cref{thm:main}]\label{th:discintro}
    For a function $f: \X \times \Y \to \{-1,1\}$, let
    \[ k = \Omega\left(\log\left(\frac{1}{\mathrm{Disc}^{\times}(f)}\right)\right).\]
    If $\Pi$ is a protocol for \jee of $f$ where the two players transmit $M_A$ and $M_B$ bits respectively then,
    \[ \left( 1 + \frac{M_A}{k}\right) \left(1 + \frac{M_B}{k} \right) \geq \Omega\left( \frac{1}{\varepsilon^2} \right).\]
    In particular, $\Rhow(f) = \Omega(k/\varepsilon^2)$, and $\Rhat(f) = \Omega(k/\varepsilon)$. 
\end{theorem}

This lower bound reveals a trade-off between the messsage lengths of the two players in the protocol. In a one-round protocol where Alice sends a message and Bob just announces the answer, Bob transmits relatively few bits (in other words, $M_B/k = o(1)$).  Therefore, \cref{th:discintro} implies that $M_A \geq \Omega(k/\varepsilon^2)$. In contrast, for multi-round protocols, it gives a total communication lower bound of $M_A + M_B \geq \Omega(k/\varepsilon)$.

\paragraph{Optimality of variance reduction.}
Theorem \ref{th:discintro} yields the following corollary for the inner-product function which has exponentially small discrepancy under the uniform distribution:
\begin{corollary}[Restatement of \cref{cor:ipcor}]
    Define $\ip:\{0,1\}^n \times \{0,1\}^n \rightarrow [-1,1]$ by $\ip(x,y) = (-1)^{\langle x, y \rangle}$. Then, $\Rhat(\ip) = \Omega(n/\eps)$ and $\Rhow(\ip) = \Omega(n/\eps^2)$. 
\end{corollary}
\ignore{
Concretely, for the $\mathbb{F}_2$ inner product function, 
\[\ip_2(x,y) = \sum_{i = 1}^n x_i \cdot y_i (\mod\ 2)\]
the discrepancy under the uniform distribution is $\Omega(n)$.  Here this tradeoff translates to a lower bound of 
\begin{equation}\label{eq:ip-one-way}
 \Rhat^{ow}(\ip) \geq \Omega\left(\frac{n}{\varepsilon^2}\right)
 \end{equation}
for one-round protocols and 
\begin{equation} \label{eq:ip-multi-way}
 \Rhat(\ip) \geq \Omega\left(\frac{n}{\varepsilon}\right)
\end{equation}}

The above results show that the variance reduction protocol in Theorem \ref{thm:var-red1} is optimal in two important ways.  First, the communication complexity of any protocol (irrespective of number of rounds) is at least $\Omega(n/\varepsilon)$ in general. Second, achieving the improved upper bound in \cref{thm:var-red1} requires at least two rounds: $1$-round protocols have a communication complexity of $\Omega(n/\varepsilon^2)$.

\paragraph{Random sampling vs \varred.}
When compared to random sampling, \varred reduces the dependence on $1/\eps$ from quadratic to linear. However the dependence on the randomized complexity of $f$ in the two protocols is different: \varred (Theorem \ref{thm:var-red1}) gives a bound of $O(\Row(f)/\varepsilon)$ bits, whereas random sampling gives $O(\R(f)/\eps^2)$. For functions where $R(f) \approx \Row(f)$, \varred beats random sampling handily. But there are functions where the two are far apart, so which protocol is better depends on the value of $\eps$.  It is natural to ask whether one could devise a protocol with complexity $O(\R(f)/\varepsilon)$ and get the best of both worlds.

We show this is impossible in general: we define the Double-Indexing function $\mathsf{DI}$ (see Theorem \ref{thm:di-lb}), which is a two-sided variant of the standard Indexing function such that 
\begin{enumerate}
    \item $\Row(\DI) = \Omega(n)$, so variance reduction gives communication $O(n/\eps)$, 
    \item $\R(\mathsf{DI}) \leq O(\log n)$, so random sampling gives communication $O(\log(n)/\eps^2)$, 
    \item For $\varepsilon \geq 2/n$, $\Rhat(\mathsf{DI}) = \Omega(\varepsilon^{-2})$.
\end{enumerate}
This rules out a protocol with cost $O(\R(f)/\varepsilon)$ in general.

\ignore{
for the $\mathbb{F}_2$-inner product $\ip_2: \{0,1\}^n \times \{0,1\}^n \to \{0,1\}$, $\ip(x,y) = \sum_{i=1}^n x_i y_i (\mod 2)$, we show that 
\[\Dhat_{\varepsilon}^{ow} (\ip) \geq \Omega\left(\frac{n^2}{\varepsilon^2}\right)   \ \ \ \ (\text{ see Theorem } \ref{thm:addreference}) \]

Notice that the deterministic $1$-way protocols need to transmit $\frac{n^2}{\varepsilon^2}$ bits, while randomized $1$-way protocols only need $\frac{n}{\varepsilon^2}$ bits of communication (see \cref{eq:ip-one-way}).}

\paragraph{Spectral lower bounds.}

By diagonalizing the matrix associated with the function $f: \X \times \Y \to [-1,1]$ in its singular vector basis, one can reduce the problem of \jee to the case where function $f$ is a diagonal matrix.  Crucially, the diagonalizing transformation can be applied locally by the two players. This suggests that the communication complexity of \jee of a function $f$ is intimately tied to the communication complexity of the identity matrix, with the singular values and singular vectors of $f$ controlling how tight this connection is.

We crystallize this intuition by showing the following general lower bound for \jee.
\begin{theorem}[Informal version of \cref{thm:spectral}]
    For every function $f: \X \times \Y \to [-1,1]$, and $t \leq \mathrm{rank}(f)$,
    \[ \Rhat_\varepsilon(f) \geq \Omega\left(\min\left(t, \frac{t^2}{\varepsilon^2 k^2 \lambda_t(f)^2}\right)\right)\]
    where $\lambda_t(f) = \sum_{i=1}^t \sigma_i^{-1}$ is the sum of the reciprocals of the $t$ largest singular values of (the matrix of) $f$.
\end{theorem}

The above bound is based on a reduction from the Gap-Hamming problem \cite{ChakrabartiR12} and diagonalization. Somewhat surprisingly, this generic result yields tight lower bounds for a variety of specific functions of interest.
For instance, it implies the following tight lower bound for equality.
\begin{equation*}
 \noindent(\text{\cref{lem:eq-lb}}) \qquad \Rhat_\varepsilon(\mathsf{EQ}_n) \geq \Omega\left(\min\left(\frac{1}{\varepsilon^{2/3}},2^n \right)\right)
\end{equation*}

Consider the function $\mathsf{Abs}: [0,1] \times [0,1] \to [0,1]$ given by $\mathsf{Abs}(x,y) = | x - y |$.  The spectral technique implies a lower bound of  
\[ \Rhat_{\varepsilon}(\mathsf{Abs}) \geq \frac{1}{\varepsilon^{2/5}}, \]
tightly matching the upper bound in \cref{thm:abs-upper}.

For the $\ip$ function, it recovers a lower bound of $\Omega(1/\eps)$. While the tight lower bound is $\Omega(n/\eps)$, this is still noteworthy since the intuition behind diagonalization is to reduce to the identity matrix, where an $O(1/\eps^{2/3})$ upper bound holds. 

Finally, the spectral lower bound can be used to show that among all Boolean functions of high rank, the equality function $\mathsf{EQ}$ has the smallest communication complexity, as a function of $\varepsilon$.
\begin{theorem}[Restatement of \cref{thm:boolean-lb}] \label{thm:boolean-lb-restate}
    Let $f:\zo^n\times \zo^n \to \zo$ be a (non-constant) Boolean function. Then 
    \[ \Rhat_\varepsilon(f) = \tilde \Omega\left(\min\left(\frac{1}{\varepsilon^{2/3}}, \frac{\rank(f)}{\log(\rank(f))} \right)\right).\] 
\end{theorem}

\paragraph{Deterministic one-way protocols}

The weakest model of communication is deterministic one-way model wherein Alice creates a deterministic sketch of her distribution and sends it to Bob, who announces the answer.
The multiplicative weights update method can be used to show that every function $f : \zo^n \times \zo^n \to [-1,1]$ admits a canonical deterministic one-way protocol whose communication complexity is $O(n^2/\varepsilon^2)$.

We show that this protocol is optimal in general.  Specifically, we show in \cref{th:iponeway} that 
$$\Dhow_\varepsilon(\ip) = \Omega\lt(\frac{n^2}{\eps^2}\rt).$$

This follows from a tight characterization of the randomized one-way communication for $\ip$, where we show that for any error probability $\delta \geq 2^{-n}$,
\[ \Rhow_{\eps, \delta}(\ip) = \Omega\left(\frac{n \log(1/\delta)}{\eps^2}\right).\]

\subsection{Conclusion and future work}

Our work introduces and motivates the problem of \Jee, and presents a systematic study of its communication complexity. We answer a number of basic problems about it, but several intriguing problems remain open. We highlight some of our favorites:
\begin{itemize}
    \item What is the complexity of \jee for set-disjointness? The spectral method yields an $\Omega(1/\eps^{2/3})$ lower bound. The multiplicative lower bound for set-disjointness as studied in \cite{van2015communication} also implies an additive approximation lower bound of $\Omega(n/\sqrt{\eps})$. The best upper bound is $O(n/\eps)$.

    \item Can we show $\Omega(n/\eps)$ lower bounds for functions under weaker conditions than high discrepancy? 

    \item We have seen that for large rank Boolean functions, $\Rhat(f)$ as a function of $\epsilon$ grows as $1/\eps^c$ for $c \in [2/3,1]$. We know examples where the right dependence is $\Theta(1/\eps^c)$ for $c =\{2/3, 1\}$. What other values can $c$ take?
    
    \item The Boolean functions for which we show $\Rhat(f) = \tilde{O}(1/\eps^{c})$ for $c < 1$ have non-trivial randomized protocols. Is there an upper bound on $\R(f)$ in term of $\Rhat(f)$ that explains this? Or, are there functions where $\R(f) = \Omega(n)$, yet $\Rhat(f)$ grows as $(1/\eps)^c$ for some $c < 1$? 
\end{itemize}

%% file: protocols.tex
\section{Communication protocols for general functions}
\label{sec:protocols}

We present our protocols for \jee for arbitrary bounded functions, starting with the obvious sampling based protocol. 
We state and analyze it for completeness. Throughout the paper, for any finite set $\mathcal{A}$, $\Delta_\mathcal{A}$ denotes the space of probability distributions on $\mathcal{A}$.

\begin{algorithm}
\caption{Random Sampling Protocol for expectation estimation}
\label{alg:samp-ee}
\noindent 

 \textbf{Parameters: } Function $f: \X \times \Y \to [-1,1]$.  Error parameter $\eps \in [0,1]$. 
 
 \textbf{Inputs: } Alice $\gets p \in \Delta_\X$, Bob  $\gets q \in \Delta_{\Y}$.

\textbf{Output:}  An $\eps$ additive estimate of $\E_{p,q}[f]$.

\begin{enumerate}
\item  Alice and Bob sample $k = O(1/\eps^2)$  points $x_1, \ldots, x_k \sim p$ and $y_1, \ldots, y_k \sim q$ independently.

\item They run a protocol for $f$ on each pair $(x_i, y_i)$ to get $v_i$ such that $|v_i - f(x_i, y_i)| \leq \eps/2$.

\item They return the estimate
$\bar{v} = \frac{1}{k} \sum_{i=1}^k v_i.$

\end{enumerate}
\end{algorithm}

\begin{theorem}
\label{thm:var-red1}
For any function $f: \X \times \Y \to [-1, 1]$, we have $\Rhat_\e(f) \le \tilde{O}(\R(f)/\eps^2)$.
\end{theorem}

\begin{proof}
We show that \cref{alg:samp-ee} achieves the bound.
The output $\bar{v}$ satisfies
\[ \lt|\bar{v} - \frac{1}{k} \sum_{i=1}^k f(x_i, y_i) \rt| \leq \frac{\eps}{2}.\]
The sample average is an unbiased estimator for $\E_{p,q}[f]$ with variance bounded by $1/k$. So taking $k = O(1/\eps^2)$ and using Markov's inequality,  we have with probability 0.9,
\[ \lt|\E_{p,q}[f]  - \frac{1}{k} \sum_{i=1}^k f(x_i, y_i) \rt| \leq \frac{\eps}{2}.\]
This ensures that $\bar{v}$ is within $\eps$ of the true expectation.
The communication used is $k \R(f) = \tilde{O}(\R(f)/\eps^2)$, where we amplify the success probability of a randomized protocol so that we can union bound over $k$ invocations of the protocol.
\end{proof}

\subsection{Debiasing protocol for better variance reduction}\label{sec:debiasing}

In this section we show the following result:

\begin{theorem}\label{thm:mainupper} 
    For any function $f: \X \times \Y \to [-1, 1]$, we have $\Rhat_\e(f) \le \tilde{O}(\Row(f)/\eps)$.
\end{theorem}

Let us assume that Alice and Bob send each other their $k$
samples. They can now compute the estimate

\begin{equation}\label{eq:debias1}
     z_1 = \frac{1}{k^2} \sum_{i=1}^{k}\sum_{j=1}^k f(x_i, y_j). 
\end{equation}

For (almost) the same communication, in place of $k$ independent samples, we
now have $k^2$ {\em dependent} samples. Optimistically, one might hope that
the variance would behave as $1/k^2$ rather than $1/k$. It is
however not hard to come up with simple examples where the variance is in fact
$1/k$: let $p, q$ both be uniform on $\zo$, and let $f(x, y) = x$.

While the naive attempt fails, we will show that there indeed is a better estimator where the $k^2$ estimates behave as though they are independent.
For a point $(x, y)$, let
\begin{align} 
\label{eq:pr-est}
g(x, y) = \E_{a \sim p}[f(a, y)] + \E_{b \sim q}[f(x, b)] -
f(x,y).
\end{align}
Observe that
\[\E_{x \sim p, y \sim q}[g(x,y)] = \E_{a \sim p, y \sim q}[f(a, y)] +
\E_{x \sim p, b \sim q}[f(x, b)] - \E_{x \sim p, y \sim q}[f(x,y)] = \E_{p,q}[f]\]
so $g$ is an unbiased estimator for $f$. In fact $g$ is an unbiased estimator even
conditioned on either one of  $x$ or $y$. In matrix language, it is unbiased for every row and every column. 

\begin{lemma}
\label{lem:row-col}
  For all $x \in \X, y \in \Y$, it holds that
  \begin{align*}
    \E_{a \sim p}[g(a,y)] =  \E_{b \sim q}[g(x,b)] = \E_{p,q}[f].
   \end{align*}
\end{lemma}
\begin{proof}
  We have
  \begin{align*}
      \E_{a \sim p}[g(a,y)] = \E_{a \sim p}[f(a, y)] + \E_{a \sim p} \E_{b \sim q}[f(a, b)] - \E_{a \sim p}[f(a, y)]
      = \E_{p,q}[f].\\
      \E_{b \sim q}[g(x,b)] = \E_{b \sim q} \E_{a \sim p}[f(a,b)]
      + \E_{b \sim q}[f(x, b)]  - \E_{b \sim q}[f(x,b)] = \E_{p,q}[f].
  \end{align*}
\end{proof}
\begin{remark}
Given a function $f(x,y)$ on a product space $\mathcal{X} \times \mathcal{Y}$,  $g(x,y)$ as defined in \eqref{eq:pr-est} is the {\it highest degree} term   the Efron-Stein decomposition of $f$.  We refer the reader to Section 8.3 in O'Donnell \cite{o2014analysis} for an exposition. 
\end{remark}

We now turn to estimating $g(x, y)$, which will be a subroutine in our protocol for $f$.

\begin{lemma}
\label{lem:comm-g}
The function $g(x,y)$ can be estimated within error $\eps/2$ by a protocol that uses $\tilde{O}(\Row(f))$ bits, with failure probability $0.1$.
\end{lemma}
\begin{proof}
Alice sends a message about $x$ using a one-way protocol for computing $f$ with failure probability $\delta$. Bob uses this to compute $f(x,y)$. He also uses it to estimate $\E_{b \sim q}[f(x,b)]$ to within $\eps/4$ using $k = O(1/\eps^2)$ samples from his distribution. To union bound over these $k +1$ computations of $f$ we take $\delta = 1/(ck)$ for large enough $c$, hence Alice must send $O(\Row(f)\log(k))$ bits of communication. 

Similarly, Bob sends a message to Alice using $O(\Row(f)\log(k))$ bits that lets her estimate $\E_{a \sim p}[f(a, y)]$ to within $\eps/4$ with probability $0.99$. 

By Equation \eqref{eq:pr-est}, these suffice to estimate $g(x,y)$ within $\eps/2$. The total communication exchanged is 
\[ O(\Row(f)\log(1/\eps)) + O(\Row(f)\log(k)) = O(\Row(f)\log(1/\eps)) = \tilde{O}(\Row(f)) \] 
\end{proof}

A couple of notes about the proof: if the one-way protocols are deterministic, we do  not lose the added $\log(1/\eps)$ factor from amplification. Also, for real-valued $f$ we can allow protocols that have $O(\eps)$ additive error.

We now give our \varred protocol. The key idea is that while the simple approach from \cref{eq:debias1} fails when directly applied to $f$ it surprisingly works for $g$: the $k^2$ estimates of $g$ behave as though they are independent.

\begin{algorithm}
\caption{Debiasing Protocol}
\label{alg:var-red}
\noindent 

 \textbf{Parameters: } Function $f: \X \times \Y \to [-1,1]$.  Error parameter $\eps \in [0,1]$. 
 
 \textbf{Inputs: } Alice $\gets p \in \Delta_\X$, Bob  $\gets q \in \Delta_{\Y}$.

\textbf{Output:}  An $\eps$ additive estimate of $\E_{p,q}[f]$.

\begin{enumerate}
\item  Alice and Bob sample $k = O(1/\eps)$  points $x_1, \ldots, x_k \sim p$ and $y_1, \ldots, y_k \sim q$ independently.

\item They run the protocol for $g$  from Lemma \ref{lem:comm-g} on each pair $(x_i, y_j)$ to get $v_{i,j}$ such that $|v_{i,j} - g(x_i, y_j)| \leq \eps/10$.

\item They return the estimate
\[ \bar{v} = \frac{1}{k^2} \sum_{i=1}^k \sum_{j=1}^k v_{i, j}. \]
\end{enumerate}
\end{algorithm}

\begin{lemma}
\label{lem:var-red}
    Let
    \[ z = \frac{1}{k^2} \sum_{i=1}^k \sum_{j=1}^k g(x_i, y_j)\]
    Then 
    \[ \E[z] = \E_{p,q}[f], \ \ \Var[z] \leq \frac{16}{k^2}.\]
\end{lemma}
\begin{proof}
For brevity, let $\mu = \E_{p,q}[f]$. 
We define
$\hat{g}(x,y) = g(x,y) - \mu$ . By Lemma \ref{lem:row-col},
\begin{align*}
  \E_{a \sim p}[\hat{g}(a,y)] = 0, \   \E_{b \sim q}[\hat{g}(x,b)]
  = 0.
\end{align*}
Since $|f(x,y)| \leq 1$,
\[ |\hat{g}(x,y)| = | \E_{a \sim p}[f(a, y)] + \E_{b \sim q}[f(x, b)] -
f(x,y) - \mu| \leq 4.\]

We can write
\begin{align*}
  z - \mu &= \frac{1}{k^2}\sum_{i=1}^k\sum_{j=1}^k \left(g(x_i, y_j) - \mu\right) = \frac{1}{k^2} \sum_{i=1}^k\sum_{j=1}^k \hat{g}(x_i, y_j)
\end{align*}
Hence we have
\begin{align}
\label{eq:var-bd1}
  k^4\E[(z - \mu)^2] &= \sum_{i=1}^k\sum_{j =1}^k
  \E[\hat{g}(x_i, y_j)^2] + \sum_{i=1}^k\sum_{j \neq j' \in [k]}\E[\hat{g}(x_i,
  y_j)\hat{g}(x_i, y_{j'})]\notag\\
  & + \sum_{i \neq i' \in [k]}\sum_{j =1}^k\E[\hat{g}(x_i,
  y_j)\hat{g}(x_{i'}, y_j)] + \sum_{i \neq i' \in [k]}\sum_{j \neq
    j' \in [k]}\E[\hat{g}(x_i, y_j)\hat{g}(x_{i'},y_{j'})]
\end{align}
where the expectations are over the random choices of $\{x_i, y_j\}_{i \in [k], j \in [k]}$. 
We bound each of these terms as follows.
\begin{align*}
  \sum_{i=1}^k\sum_{j =1}^k
  \E[\hat{g}(x_i, y_j)^2] & \leq 16 k^2\\
  \sum_{i=1}^k\sum_{j \neq j' \in [k]}\E[\hat{g}(x_i,
    y_j)\hat{g}(x_i, y_{j'})] &=   \sum_{i=1}^k\sum_{j \neq j' \in [k]}\E_{x_i}\E_{y_j, y_{j'}}[\hat{g}(x_i,y_j)\hat{g}(x_i,y_{j'})]\\
    &= \sum_{i=1}^k\sum_{j \neq j' \in   [k]}\E_{x_i}\E_{y_j}[\hat{g}(x_i,y_j)]\E_{y_{j'}}[\hat{g}(x_i,y_{j'})]
    = 0. \\
\sum_{i \neq i' \in [k]}\sum_{j =1}^k\E[\hat{g}(x_i,
  y_j)\hat{g}(x_{i'}, y_j)] &= \sum_{j =1}^k\sum_{i \neq i' \in [k]}\E[\hat{g}(x_i, y_j)\hat{g}(x_{i'}, y_j)]\\
&= \sum_{j=1}^k\sum_{i \neq i' \in   [k]}\E_{y_j}\E_{x_i}[\hat{g}(x_i,y_j)]\E_{x_{i'}}[\hat{g}(x_{i'},y_{j})]
= 0.\\
\sum_{i \neq i' \in [k]}\sum_{j \neq
    j' \in [k]}\E[\hat{g}(x_i, y_j)\hat{g}(x_{i'},y_{j'})] &= \sum_{i \neq i' \in [k]}\sum_{j \neq
    j' \in [k]}\E_{x_i, y_j}[\hat{g}(x_i, y_j)] \E_{x_{i'}, y_{j'}}[\hat{g}(x_{i'},y_{j'})] 
    = 0.\\
\end{align*}
Plugging these bounds into Equation \eqref{eq:var-bd1} and rearranging gives the desired bound.
\end{proof}

We can now prove \cref{thm:mainupper}, our main upper bound for the \jee problem.
\eat{\begin{theorem}
\label{thm:var-red2}
For all $f:\X \times \Y \rightarrow [-1,1]$, $\Rhat_\e(f) \le \tilde{O}(\Row(f)/\eps)$.
\end{theorem}}
\begin{proof}[Proof of \cref{thm:mainupper}]
Indeed, we show that \cref{alg:var-red} achieves the bound.

Since  $|v_{i, j} - g(x_i, y_j)| \leq \eps/10$ for every $i, j$, and $\bar{v}$ is the average of the $v_{i,j}$s, we have
\[ \lt| \bar{v} - \fr{k^2}\sum_{i,j} g(x_i, y_j) \rt| \leq \frac{\eps}{10}. \]
By Lemma \ref{lem:var-red}, the latter is an unbiased estimator for $\E_{p,q}[f]$ with variance $16/k^2$. So by Markov's inequality, with probability $0.9$ over the choice of samples,
\[ \lt| \fr{k^2}\sum_{i,j} g(x_i, y_j) - \E_{p,q}[f]\rt| \leq \frac{40}{k}. \] 
By choosing $k =O(1/\eps)$, we can bound this by $\eps/2$, which implies that $\bar{v}$ is $\eps$ close to $\E_{p,q}[f]$.
\end{proof}

We note also that if $\X$ and $\Y$ have size $2^n$, we have $\Row(f) \leq n$, which gives the following corollary.

\begin{corollary}
    For every bounded function $f: \zo^n \times \zo^n \to [-1,1]$, we have $\Rhat_\e(f) \leq O(n/\eps)$. 
\end{corollary}

\paragraph{Debiasing versus random sampling: no best of both worlds.}

We have seen protocols for \jee that require communication $\tilde{O}(\R(f)/\eps^2)$ and $\Ot(\Row(f)/\eps)$. It raises the natural question whether the bound of $\max(\R(f)/\eps^2, \Row(f)/\eps)$ can be replaced by a bound like $\R(f)/\eps$. We will show that such a bound does not hold for arbitrary functions, and that the dependence on the $\eps$ parameter does actually need to be quadratic for $\eps$ sufficiently small.

We define the Double-indexing function $\DI:\pmo^n \times \pmo^n \to \pmo$ where $n = k + 2^k$ for some integer $k$ as follows. Alice's input is a pair $(i, x)$ for while Bob's input is a pair $(j, y)$ where $i, j \in [2^k], x, y  \in \pmo^k$.  We define $\DI((i,x), (j ,y) = x_jy_i$. In essence, both Alice and Bob have to solve an instance of indexing. It is easy to see that $\R(\DI) = O(k) = O(\log(n))$ since they can exchange $i, j$ in the first round and $x_j, y_i$ in the next. However, the lower bound for standard Indexing \cite{KremerNR99} shows that 
\[ \Row(\DI) =  \Omega(2^k) = \Omega(n).\]
Hence the upper bound for $\Rhat_{\eps}(\DI)$ is $\Ot(\min(\log(n)/\eps^2), n/\eps))$. The two bounds are essentially equal when $\eps \approx 1/n$, for larger $\eps$ we get $O(\log(n)/\eps^2)$, but for smaller $\eps$, we get $O(n/\eps)$. 

We will show the following lower bound:

\begin{theorem}
    \label{thm:di-lb}
    For $\eps \geq 2/n$, $\Rhat_{\eps}(\DI) = \Omega(1/\eps^2)$.
\end{theorem}
\begin{proof}
    We will reduce from Gap Hamming on Inputs of length $2^{2k}$. We think of the inputs to Alice and Bob as matrices $A, B$ each in $\pmo^{2^k \times 2^k}$. Let $\iprod{A}{B}_F$ denote the entry-wise or Frobenius inner product of $A$ and $B$. Computing this to within error $\eps2^{2k}$ requires $\Omega(\min(1/\eps^2, 2^{2k}))$ bits of communication by Theorem \ref{thm:cr}. 

    However, we can encode this as an instance of $\DI$. Alice gets as input $i \in[2^k]$ drawn uniformly at random and the $i^{th}$ row $x = A[i, :] \in \pmo^{2^k}$. Bob gets as input $j \in [2^k]$ drawn uniformly at random and the $j^{th}$ column $y = B[:, j] \in \pmo^{2^k}$. We then have
    \begin{align*}
        \E[\DI(i, x), (j ,y)] = \fr{2^{2k}} \sum_{i,j}x_jy_i = \fr{2^{2k}}\sum_{i,j}A[i,j]B_[i,j] = \frac{\iprod{A}{B}_F}{2^{2k}}.
    \end{align*}
    Hence estimating $\E[\DI(i, x), (j ,y)]$ requires $\Omega(1/\eps^2)$ bits of communication as long as $\eps \geq 2^{-k}$. For $k\geq 2$, $2^{-k} \geq 2/n$.
\end{proof}

\subsection{Protocols based on spectral properties}\label{sec:spectral}

For any $f: \zo^n \times \zo^n \to [0,1]$ there is always an $\eps$-error protocol of cost $O(n2^n\log(1/\eps))$ for any function where Alice sends Bob all her probabilities to within $\eps/2^n$ accuracy. In this section we discuss how to improve this bound for functions with certain spectral properties.

To each function $f$ we associate its $2^n$ by $2^n$ matrix $A_f$ whose $(i, j)$-th entry is $f(i, j)$. Then, define $\rank(f) = \rank(A_f)$. Additionally, we define the $\eps$-approximate rank $\rank_{\eps}(f)$ to be the minimum rank of any matrix $B$ such that $\|A_f - B\|_{\infty} \leq \eps$. We also define the spectral norm $\sigma(f)$ to be the largest singular value of $A_f$. Similarly, for any $t \in [2^n]$, let $\sigma_t(f)$ denote the $t$'th largest singular value of $f$.

We note also that we may extend these definitions to continuous $O(1)$-Lipschitz functions $f: [0,1]^n \times [0,1]^n \to [-1,1]$ by discretizing the domain to a grid of size $1/m$ for large enough $m = O(1/\eps)$, and then defining the matrix $A_f$ as above. Then Alice and Bob may discretize their distributions to the grid as well, incurring an additional additive error of at most $\e/3$ (for a suitable choice of $m$).

We now present a protocol for \jee which works for functions of low rank.

\begin{theorem}
\label{thm:svd-protocol}
    For any function $f$ with $\rank_{\e/2}(f) =r$, $\Dhat_\e(f) \leq O(rn + r\log(1/\eps))$.
\end{theorem}
\begin{proof}
    First, note that we may assume without loss of generality that $\rank(f) = r$, since if not we may replace $f$ by a rank-$r$ function that approximates it to within $\e/2$ in infinity norm, which only changes the expectation by at most $\e/2$.

    Let $A_f$ be the matrix corresponding to $f$, and let $A_f = U\Sigma V^{\sfT}$ be its singular value decomposition, where $\sigma_i$ is the $i$-th singular value. Then
    \[ p^{\sfT}Aq = \sum_{i \in [r]}\sigma_i p^{\sfT} u_iv_i^\sfT q\]
    Alice sends Bob $\tilde{p} \in \mathbb{R}^r$ which is $p^{\sfT}U $ rounded to precision $\eps/ 2^{n+1}$.  Bob uses this to compute the estimate
    \[ v =  \sum_{i \in [r]}\sigma_i \tilde{p}_i v_i^{\sfT} q.\]
    We can bound the error of the estimate by
    \begin{align*}
        |v - p^{\sfT}Aq| &\leq \sum_{i \in [r]}(\tilde{p}_i - p^{\sfT}u_i)\sigma_i v_i^{\sfT}q\\
        & \leq \frac{\eps}{2^{n+1}}\sum_{i \in [r]}|\sigma_i v_i^{\sfT}q| \  \ \ \text{(Holder's inequality)}\\
        & \leq \frac{\eps}{2^{n+1}}\lt(\sum_{i \in [r]}\sigma_i^2\rt)^{1/2} \lt(\sum_{i \in [r]}(v_i^tq)^2\rt)^{1/2} \ \text{(Cauchy-Schwarz)}\\
        & \leq \frac{\eps}{2^{n+1}}\|A\|_F\|q\|_2 \  \ \ \  \text{(orthonormality of $v$)}\\
        & \leq \frac{\eps}{2^{n+1}}2^n =\eps.
    \end{align*}
\end{proof}

Next we present a protocol that gives an upper bound in terms of the spectral norm. As a first step, we use a small-communication protocol for estimating the inner product between two real vectors of bounded norm. We define the communication problem $\rip(\eps)$ where Alice and Bob each have input vectors $a \in \mR^n$ and $b \in \mR^n$ respectively, and their goal is to estimate $a\cdot b$ to within additive error $\eps\|a\|_2\|b\|_2$. There is a well known protocol for this that forms the basis of the sketches in \cite{AlonMatiasSzegedy-JCSS99, CharikarChenFarachColton-ICALP02}. We provide the standard proof for completeness.

\begin{lemma}
\label{lem:ip-prot}
    There is a randomized protocol for $\rip(\eps)$ using $\tilde{O}(1/\eps^2)$ bits of communication.
\end{lemma}
\begin{proof}
    Alice and Bob use shared randomness to sample $u_1, \ldots, u_k \in \pmo^n$ for $k = O(1/\eps^2)$. They compute the estimate
    \[ v = \fr{k}\sum_{i \in [k]}(a\cdot u_i)(b\cdot u_i) \]
    to within accuracy $\eps/2$.
    This requires Alice to send Bob the inner products $\{a\cdot u_i\}_{i \in [k]}$ to $O(\log(1/\eps)$ bits of precision.

    For correctness, we observe that for random $u \in \pmo^k$, $(a\cdot u)(b\cdot u)$ is an unbiased estimator of $a\cdot b$ whose variance is bounded by $O(\|a\|^2_2\|b\|_2^2)$. Hence with $k = O(1/\eps^2)$ samples, the standard deviation is bounded  by $\eps\|a\|_2\|b\|_2$ with probability $0.9$. 
\end{proof}

We next present a protocol for \jee which does well for functions of small spectral norm.

\begin{theorem} \label{thm:spectral-upper}
    $\Rhat_\e(f) \leq \tilde{O}((\sigma(f)/\eps)^{2/3})$.
\end{theorem}
\begin{proof}
For a parameter $\beta$, defin $p^h$ to be the vector of \textit{heavy} components: $p^h_i = p_i$ if $p_i \geq \beta$ and $0$ otherwise. Let $p^\ell = p - p^h$ be the vector of \textit{light} components. Then, note that

\begin{equation*}
    p^{\sfT}Aq = (p^h)^{\sfT}Aq^h + (p^h)^{\sfT}Aq^\ell + (p^\ell)^{\sfT}Aq^h + (p^\ell)^{\sfT}Aq^\ell.
\end{equation*}

Then, we provide an algorithm in \cref{alg:spec-ee} which works essentially as follows. We estimate the first three terms by exchanging the heavy components, and then estimate the last term using the inner product protocol from Lemma \ref{lem:ip-prot}.

\begin{algorithm}
\caption{Spectral protocol for expectation estimation}
\label{alg:spec-ee}
\noindent 

 \textbf{Parameters: } Function $f: \zo^n \times \zo^n \to [-1,1]$.  Error parameter $\eps \in [0,1]$. 
 
 \textbf{Inputs: } Alice $\gets p \in \Delta_\X$, Bob  $\gets q \in \Delta_{\Y}$.

\textbf{Output:}  An $\eps$ additive estimate of $\E_{p,q}[f]$.

\begin{enumerate}
\item Set $\beta = (\eps/\sigma(A))^{2/3}$. 

\item Let $p^h$ be the vector of heavy components $p^h_i = p_i$ if $p_i \geq \beta$ and $0$ otherwise. Let $p^\ell = p - p^h$.

\item  
Alice and Bob exchange the vectors $p^h$ and $q^h$ (to $O(\log(1/\eps)$ bits of precision), and use this to compute $u  = (p^h)^{\sfT}Aq^h + (p^h)^{\sfT}Aq^{\ell} + (p^\ell)^{\sfT}Aq^h$ to within an additive error of $\eps/2$. \label{step:heavy-exch}

\item Let $A = U \Sigma V^t$. Alice and Bob respectively form the vectors 
\begin{align*} 
    a = \sum_{i} \sqrt{\sigma_i}p_i^{\ell}u_i, \  \ b = \sum_{j} \sqrt{\sigma_i}q_i^{\ell}v_j.
\end{align*}
They use Lemma \ref{lem:ip-prot} to compute $v$ such that 
\[ |v - a \cdot b| \leq \delta \|a\|_2\|b\|_2. \]  \label{step:ip-prot}

\item They return $u  + v$. 
\end{enumerate}
\end{algorithm}

To analyze the error of the protocol, we note that step \ref{step:heavy-exch} requires $\Ot(1/\beta)$ bits of communication, and results in additive error $\eps/2$.  

For step \ref{step:ip-prot},  we note that
\[ a\cdot b = \sum_{i \in \zo^n} \sigma_ip^\ell_iu_i q^\ell_i v_i = (p^\ell)^\sfT Aq^\ell.\]
We bound the norm of $a$ using the orthonormality of the $u_i$s as
\[ \|a\|_2^2 = \sum_i \sigma_i(p^\ell_i)^2 \leq \sigma \sum_{i}(p^\ell_i)^2 \leq \sigma\beta\]
where we use $p^\ell_i \leq \beta$ and $\sum_i p^\ell_i \leq 1$. A similar bound holds for $b$. Hence the error of step \ref{step:ip-prot} is bounded by $\delta \sigma \beta$; for this to be bounded by $\eps/2$ we set $\delta = \eps/(2\sigma\beta)$. This step requires $\tilde{O}(\sigma^2\beta^2/\eps^2)$ bits of communication. 

We balance the communication in the two steps by taking  $\beta = c(\eps/\sigma)^{2/3}$ for some constant $c$. We get a protocol that gives $\eps$ error in total using $\tilde{O}((\sigma/\eps)^{2/3})$ bits of communication.
\end{proof}

We can combine this with the protocol from \cref{thm:svd-protocol} to give a protocol that only requires a bound on $\sigma_t$ for some $t \geq 1$. This allows us to effectively ignore $t$ large singular values, at a cost of $\tilde{O}(tn)$. 

\begin{lemma}
    For any $f: \zo^n \times \zo^n \to [-1,1]$ and $t \in\{1, \ldots, 2^n\}$, we have $\Dhat_\e(f) \leq \tilde{O}(t\log(\sigma_1(f)) + (\sigma_t(f)/\eps)^{2/3})$.
\end{lemma}
\begin{proof}
For brevity, let $\sigma_i \equiv \sigma_i(f)$ for all $i$.
    Alice and Bob first exchange their $\beta$-heavy elements $p^h$ and $q^h$. Then, for all $j \leq t$, they exchange $(p^\ell)^\sfT u_j$ and $v_j^\sfT q^\ell$ to within precision $\eps/\sigma_1$, which costs $\tilde{O}(t\log(\sigma_1/\eps))$ bits. They then run \cref{alg:spec-ee} to estimate the contribution from the singular vectors for $j > t$, which costs $\tilde{O}((\sigma_t/\eps)^{2/3})$. 
\end{proof}

We can use this to derive the following useful bound.

\begin{corollary}
    Let $t_n(f) =|\{i:\sigma_i(f) > n\}|$. Then we have
    \[ \Rhat_\e(f) \leq \tdO(t_n(f)n + (n/\eps)^{2/3}). \]
\end{corollary}

%% file: upper.tex
\section{Protocols for specific functions}
\label{sec:specific}

\renewcommand{\f}{\frac}
\renewcommand{\fh}{\hat f}
\renewcommand{\al}{\alpha}
\renewcommand{\e}{\varepsilon}
\renewcommand{\del}{\partial}
\renewcommand{\p}[1]{\left(#1\right)}

\subsection{\texorpdfstring{$\EQ$}{EQ}}
We consider the case where $f$ is the equality function $\EQ(x, y) = \ind{x = y}$ for $x, y \in \zo^n$. The matrix $A_\EQ: \zo^n \times \zo^n \to \zo$ is the $2^n \times 2^n$ identity matrix $\Id$. Hence computing the matrix expectation corresponds to computing the inner product of the distributions $p$ and $q$:
\[ \E_{x \sim p, y\sim q}[\Id(x,y)] = \sum_{x \in \pmo^n}p(x)q(x).\]
We will write $\iprod{p}{q}$ as short hand for this inner product. 

In this section, we will prove the following upper bound.
\begin{theorem} \label{thm:eq-upper}
$\Rhat_\e(\EQ) \leq \tilde{O}(1/\eps^{2/3})$.
\end{theorem}

For ease of notation, we will show that we can estimate $\iprod{p}{q}$ within error $O(\e)$ instead of $\e$.

We say that $x \in \zo^n$ is $\e$-heavy for $p$ if $p(x) \geq \e$. We set $p^{\e}$ to be the vector where we set all coordinates in $p$ that are less than $\e$ to $0$. Dropping coordinates can only reduce the inner product, but the following lemma show that it does not reduce by too much.

\begin{lemma}
\label{lem:drop}
    We have $\iprod{p}{q} \leq \iprod{p^\e}{q^\e} + 2\e. $
\end{lemma}
\begin{proof}
    We can write 
    \begin{align*} 
        \iprod{p}{q} - \iprod{p^\e}{q^\e} &= \sum_{x:~ p(x) < \e \text{ or } q(x) < \e} p(x)q(x)\\
     & \leq \sum_{x:~ p(x) < \e} p(x)q(x) + \sum_{x:~ q(x) < \e} p(x)q(x)\\
     & \leq \e \sum_{x:~ p(x) < \e} q(x) + \e \sum_{x:~ q(x) < \e} p(x)\\
     & \leq 2\e.
    \end{align*}  
\end{proof}

Thus, to show \cref{thm:eq-upper}, it will suffice for us to provide a protocol to estimate $\iprod{p^\e}{q^\e}$ to within additive error $O(\e)$ for some parameter $\e$ to be chosen later. Note that both $p^\e$ and $q^\e$ have support size at most $1/\e$.

Using the public randomness, we sample a uniformly random hash function $h: \zo^n \to [m]$ for $m = 40/\e^2$. Note that since $p^\e$ and $q^\e$ have
support sizes at most $O(1/\e)$, the probability that $h$ maps any two elements in the support of $p^\e$ or $q^\e$ to the same value is at most $0.1$ by a union bound. Thus, letting $p'$ and $q'$ be the distributions on $[m]$ defined by mapping $p^\e$ and $q^\e$ through $h$, we have with probability at least $0.9$ that $\iprod{p'}{q'} = \iprod{p^\e}{q^\e}$. Thus, it remains only to show a protocol to estimate $\iprod{p'}{q'}$.

A simple $1$-round protocol is where Alice sends Bob the list of her non-zero probability elements, along with the probabilities of those elements to within additive error $\eps^2$. This uses $O(\log(1/\eps))$ bits of communication per element. This requires $\Ot(1/\eps)$ bits of communication. 

We can improve on the communication in 2 rounds as follows. Let $\beta > \e$ be a parameter to be chosen later. We can write
\begin{align*}
    \iprod{p'}{q'} = \sum_{x \in [m]} p'(x)q'(x) = \sum_{x: p'(x) \geq \beta}p'(x)q'(x) + \sum_{x: p'(x) < \beta}p'(x)q'(x).
\end{align*}
We can again estimate the first summation by Alice sending Bob all her $\beta$-heavy elements and their probabilities. To estimate the second term, let $q'' = q' / \norm{q'}_1$ be the probability distribution obtained by renormalizing $q'$. Then,
\begin{equation} \label{eq:eq-estimate}
    \sum_{x:p'(x) < \beta}p'(x)q'(x) = \norm{q'}_1 \E_{x \sim q''}[\ind{p'(x) < \beta}p'(x)].
\end{equation}
Bob now sends Alice $\norm{q'}_1$ (to precision $\e$) and $t$ samples $x_1, \ldots, x_t$ from $q''$. Alice then computes and returns
\[ \fr{t}\sum_{k \in [t]} \ind{p'(x_k) < \beta}p'(x_k) \]
which is an unbiased estimator for \eqref{eq:eq-estimate}. Note that the quantity in the expectation in \eqref{eq:eq-estimate} is bounded by $\beta$, so the variance of this estimator is at most $\beta^2/t$.
Hence with probability $0.9$, the error of the estimator is no more than $10\beta/\sqrt{t}$. We now choose $\beta = \eps^{2/3}$ and $t = 1/\eps^{2/3}$ so that this error is at most $O(\eps)$. The total communication of the protocol is then $\Ot(1/\beta)$ to communicate the $\beta$-heavy elements and $\Ot(t)$ to send the samples, so the total communication is $\tilde{O}(1/\eps^{2/3})$, as desired. This completes the proof of \cref{thm:eq-upper}.

\smallskip
We also note that this implies an upper bound for any $f$ whose matrix is sparse.

\begin{corollary} \label{cor:sparse}
Suppose that $f: \zo^n \times \zo^n \to [-1, 1]$ is such that the matrix $A_f$ is $s$-sparse (i.e., has at most $s$ nonzero entries per row). Then, $\Rhat_\e(f) \leq \tilde{O}(s^{5/3} / \eps^{2/3})$.
\end{corollary}
\begin{proof}
Note that we can write $A_f$ as the sum of $s$ 1-sparse matrices $A_1, \ldots, A_s$, where each $A_i$ has at most one nonzero entry per row.

Further, we can approximate each $A_i$ via a dyadic composition by a sum of multiples (with bounded coefficients) of $O(\log \inv \e)$ matrices which are $1$-sparse and have entries in $\{0, 1\}$.

Thus, it is enough to provide a protocol for the case where $A_f$ is 1-sparse and has entries in $\{0, 1\}$; we may then run this protocol on all the $\Ot(s)$ matrices (with error parameter $\Ot(\e / s)$) and take the appropriate linear combination the estimates.

In the case where $A_f$ is 1-sparse and has entries in $\{0, 1\}$, we can view $f$ as defining a partial function from $\zo^n$ to $\zo^n$. In this case, for each value of $x$, there is at most one value of $y$ such that $f(x, y) = 1$. Thus, Alice can instead take the distribution $p'$ to be the distribution on $\zo^n$ defined by mapping $p$ through this function (we can increase $n$ by 1 to have a new value to which to map the $x$ values which do not have a corresponding $y$ value). Then, we simply use the protocol of \cref{thm:eq-upper} to estimate this inner product to within additive error $O(\e)$ using $\tilde{O}(1/\eps^{2/3})$ bits of communication. This completes the proof.
\end{proof}

\begin{remark} \label{rem:improved-sparse}
The bound in \cref{cor:sparse} can actually be improved to $\tilde{O}(s^{4/3} / \eps^{2/3})$ because each instance of the protocol has independent error, so we may reduce their error parameters to $\Ot(\e / \sqrt s)$. We omit the details.
\end{remark}

\subsection{\texorpdfstring{$\GT$}{GT}} \label{sec:gt-upper}
We now consider the case of the greater-than function $\GT(x, y) = \ind{x \geq y}$. (Here $x, y \in \{0, 1\}^n$ are viewed as integers in $[2^n]$.)  Note that the matrix $A_\GT$ is the upper-triangular matrix defined by $A_\GT(x, y) = 1$ if $x \geq y$ and $0$ otherwise. We show the following upper bound.

\begin{theorem} \label{thm:gt-upper}
$\Rhat_\e(\GT) \leq \tilde{O}(n / \eps^{2/3})$, achieved by a protocol which uses 2 rounds of communication
\end{theorem}

We will use the following fact.
\begin{fact} \label{fact:partition-weak}
Let $p$ be a probability distribution on $[2^n]$. There exists a partition $J(p, \beta)$ of $[2^n]$ into disjoint intervals $I_1, \ldots, I_m$ where $m = O(1/\beta)$ such that $p(I_j) \le \beta$ for all $j$.
\end{fact}

Now, Alice and Bob each compute partitions $I^A_1, \dots, I^A_{m_A}$ and $I^B_1, \dots, I^B_{m_B}$ of $[2^n]$ using \cref{fact:partition-weak} with parameter $\beta$. Let $I_1, \dots, I_m$ be the common refinement of these two partitions, so that $m = m_A + m_B - 1$. Define $p_j = p(I_j)$  and $q_j = q(I_j)$ for each $j$, and let $p|_j$ and $q|_j$ denote the conditional distributions of $p$ and $q$ on $I_j$. Now, we can write

\begin{align}
    \E_{x \sim p, y \sim q}[\ind{x \geq y}] 
    &= \sum_{j, j'} p_j q_{j'} \E_{x \sim p|_j, y \sim q|_{j'}}[\ind{x \geq y}] \nonumber \\
    &= \sum_{j} p_j q_j \E_{x \sim p|_j, y \sim q|_j}[\ind{x \geq y}] + \sum_{j > j'} p_j q_{j'} \nonumber \\
    &= \E_{\substack{x \sim p \\ I_j \ni x}} \left[ q_j \Pr_{y \sim q|_j}[x \geq y] \right] + \sum_{j > j'} p_j q_{j'}. \label{eq:gt-decomp}
\end{align}

Note that given the interval partition, Bob can compute the sum (the second term in \eqref{eq:gt-decomp}) from just the values $p_j$. Additionally, for any specific value $x$, Bob can also compute the value inside the expectation in \eqref{eq:gt-decomp}, since he knows $q_j$ and the distribution $q|_j$. Thus, Bob can estimate the value of the expectation given random samples from $p$. This gives rise to the following protocol, illustrated in \cref{alg:gt}: Alice and Bob share the partition, and Alice sends the values $p_j$ for each $j$. Alice also sends $k = O((\beta/\e)^2)$ samples from $p$. Bob then computes the sum in \eqref{eq:gt-decomp} and estimates the expectation using the samples.

\begin{algorithm}
    
\caption{Randomized protocol for expectation estimation for $\GT$}
\label{alg:gt}
\noindent

\textbf{Inputs:} Alice $\gets p$, Bob $\gets q$. \\
\textbf{Output:} An estimate of $\E_{x \sim p, y \sim q}[\ind{x \geq y}]$.
\begin{enumerate}
\item Set $\beta = O(\e^{2/3})$.
\item Bob sends Alice the endpoints of the partition $J(q, \beta)$.
\item Alice sends Bob the endpoints of the partition $J(p, \beta)$. and the values $p_j$ (to a precision of $\e^3$) for each $j$.
\item Alice sends Bob $k = O((\beta/\e)^2)$ samples from $p$.
\item Bob computes the sum in \eqref{eq:gt-decomp} and estimates the expectation using the samples.
\end{enumerate}
\end{algorithm}

\begin{proof}[Proof of \cref{thm:gt-upper}]
We show that \cref{alg:gt} estimates $\E_{x \sim p, y \sim q}[\ind{x \geq y}]$ to within an additive error of $O(\e)$. The sum in \eqref{eq:gt-decomp} is estimated to within $O(\e)$ since there are $O(\beta^{-2})$ terms and each $p_j$ is estimated to within $\e^3$. For the expectation term, note that the value inside the expectation is bounded by $\beta$, so the variance of the estimate is $O(\beta^2/k) = O(\e^2)$. Thus, by Chebyshev's inequality, the estimate is within $O(\e)$ of the true value with probability at least $2/3$, as desired.

Finally, note that the communcation complexity of the protocol is as follows: it takes $O(n/\beta)$ bits for Alice and Bob to send their partitions, $\Ot(1/\beta)$ bits for Alice to send the values $p_j$, and $O(nk) = O(n(\beta/\e)^2)$ bits for Alice to send the samples. Thus, the overall communication complexity is $\tilde{O}(n/\beta + n(\beta/\e)^2)$. Setting $\beta = \e^{2/3}$ gives the desired bound of $\tilde{O}(n/\e^{2/3})$.
\end{proof}

In fact, we can improve the dependence in $n$ of the complexity of this protocol by using more rounds of communication.

\begin{theorem} \label{thm:gt-strong}
    $\Rhat_\eps(\GT) \le \Ot(\log(n) / \eps^{2/3})$.
\end{theorem}
\begin{proof}[Proof sketch]
    We sketch how to modify \cref{alg:gt} to achieve the desired complexity. We will use the classic result (see \cite{KN97}) that there is a protocol for $\GT$ using $O(\log n)$ bits of communication. The basic idea is that we can use the $\GT$ protocol to simulate \cref{alg:gt} without Alice or Bob sending the interval endpoints explicitly. 

    More specifically, Alice and Bob can use a merge sort to figure out the relative order of all their interval endpoints. In this way, we use $O(m) = O(1/\beta)$ instances of the $\GT$ protocol. The values $p_j$ (and similarly $q_j$ by Bob) can also be computed by Alice (up to $\poly(\e)$ precision) using a binary search to figure out where each interval endpoint lies relative to the distribution $p$, using $O(\log \inv \e) = \Ot(1)$ instances of $\GT$ per endpoint. Similarly, for each sampled point $x$, Bob can use a similar binary search to estimate $\Pr_{y \sim q|_j}[x \geq y]$.
\end{proof}

We note that as a corollary, we also get a protocol for Toeplitz matrices which are defined by a sequence with few changes.

\begin{corollary}
Suppose that $f$ corresponds to a Toeplitz matrix whose diagonal entries are given by a (bounded in $[-1, 1]$) sequence $a_{-2^n+1}, \ldots, a_{2^n-1}$ which is piecewise constant, changing at most $r$ times. Then, $\Rhat_\eps(f) \le \tilde{O}(r^{5/3}n / \eps^{2/3})$.
\end{corollary}
\begin{proof}
We can write any such function $f$ as the sum of shifted $\GT$ functions as follows:
\begin{equation*}
    f(x, y) = \sum_{i=1}^{r} c_i \cdot \GT(x + t_i, y),
\end{equation*}
where $|c_i| \le 2$ and $t_i \in [-2^n, 2^n]$ are the shift amounts corresponding to the change points of the sequence. Therefore, we can run the protocol for $\GT$ from \cref{thm:gt-upper} for each of these $r$ terms, with an error bound of $\e/r$ for each term. This gives the desired communication complexity of $\tilde{O}(r^{5/3}n / \eps^{2/3})$.
\end{proof}

\begin{remark}
Similarly to \cref{rem:improved-sparse}, the dependence on $r$ in the above corollary can be improved to $r^{4/3}$ by reducing the error parameter of each instance of the $\GT$ protocol to $\e / \sqrt r$.
\end{remark}

\begin{remark}
It is not possible to get a nontrivial bound for all Toeplitz matrices: consider the Paley matrix, which is a Toeplitz matrix whose entry $(i, j)$ is given by the Legendre symbol of $i-j$ modulo a prime $p \approx 1/\e$. This matrix has orthogonal rows of norm roughly $\sqrt{p}$, so one can use the spectral lower bound which we prove in \cref{thm:spectral} to show that the communication complexity of \jee with respect to the Paley matrix is $\Omega(1/\e)$.
\end{remark}

\subsection{Absolute difference}
Next we turn to the absolute difference function, $\ab(x, y) = |x - y|$, in the continuous setting $\X = \Y = [0, 1]$. In this case, we show the following upper bound.

\begin{theorem} \label{thm:abs-upper}
    $\Rhat_\eps(\ab) \leq \tilde{O}(1/\eps^{2/5})$.
\end{theorem}

The protocol will be very similar to the one for $\GT$ described in \cref{sec:gt-upper}. We will use the following stronger version of \cref{fact:partition-weak}.

\begin{lemma}
\label{lem:partition-strong}
Let $p$ be a probability distribution on $[0, 1]$. There exists a partition $J'(p, \beta)$ of $[0,1]$ into disjoint intervals $I_1, \ldots, I_m$ where $m = O(1/\beta)$ such that: 
\begin{itemize}
\item The width of every interval is at most $\beta$. 
\item $p(I_j) \le \beta$ for all $j$.
\end{itemize}
\end{lemma}
\begin{proof}
    We construct the partition greedily by starting a new interval each time either the probability or the width of the current interval reaches $\beta$ (assuming that the distribution is continuous). Since we can start at most $1/\beta$ new intervals for either reason, the size of the partition is bounded by $2/\beta$. 
\end{proof}

Again, we let Alice and Bob each compute a partition and let $I_1, \dots, I_m$ be their common refinement. Let $p_j, q_j, p|_j, q|_j$ be defined as before. Then, similarly to \eqref{eq:gt-decomp}, we can write
\begin{align}
\E_{x \sim p, y \sim q}[|x - y|]
&= \sum_{j, j'} p_j q_{j'} \E_{x \sim p|_j, y \sim q|_{j'}}[|x - y|] \nonumber \\
&= \sum_{j} p_j q_j \E_{x \sim p|_j, y \sim q|_j}[|x - y|] 
+ \sum_{j \neq j'} p_j q_{j'} \E_{x \sim p|_j, y \sim q|_{j'}}[x - y] \nonumber \\
&= \E_{\substack{x \sim p \\ I_j \ni x}} \left[ q_j \E_{y \sim q|_j}[|x - y|] \right]
+ \sum_{j \neq j'} p_j q_{j'} \left| \E_{x \sim p|_j}[x] - \E_{y \sim q|_{j'}}[y] \right|. \label{eq:abs-decomp}
\end{align}
Similarly, it is now enough for Alice to send $p_j$ and the conditional expectations $\E_{x \sim p|_j}[x]$ for all $j$ to compute the second term, in addition to samples from $p$ to estimate the first term. Note that this time, the quantity inside the expectation is bounded by $\beta^2$ instead of $\beta$ since the width of each interval is at most $\beta$. This gives rise to the following protocol, illustrated in \cref{alg:abs}.

\begin{algorithm}

\caption{Randomized protocol for expectation estimation for $\ab$}
\label{alg:abs}
\noindent
\textbf{Inputs:} Alice $\gets p$, Bob $\gets q$. \\
\textbf{Output:} An estimate of $\E_{x \sim p, y \sim q}[|x - y|]$.
\begin{enumerate}
\item Set $\beta = O(\e^{2/5})$.
\item Bob sends Alice the endpoints of the partition $J'(q, \beta)$.
\item Alice sends Bob the endpoints of the partition $J'(p, \beta)$, the values $p_j$ (to a precision of $\e^3$) for each $j$, and the values $\E_{x \sim p|_j}[x]$ (to a precision of $\e^3$) for each $j$.
\item Alice sends Bob $k = O((\beta^2/\e)^2)$ samples from $p$.
\item Bob computes the sum in \eqref{eq:abs-decomp} and estimates the expectation using the samples.
\end{enumerate}

\end{algorithm}

The analysis of \cref{alg:abs} is identical to that of \cref{alg:gt}, and we omit the details. The key difference is that the quantity inside the expectation in \eqref{eq:abs-decomp} is bounded by $\beta^2$ instead of $\beta$, which allows us to set $\beta = \e^{2/5}$ instead of $\e^{2/3}$, giving us the result of \cref{thm:abs-upper}.

\smallskip
Additionally, the result of \cref{thm:abs-upper} can be extended to convex Lipschitz functions:
\begin{corollary} \label{cor:convex-lipschitz}
    Let $f: [0, 1]^2 \to [-1, 1]$ be a $1$-Lipschitz convex function. Then, $\Rhat_\eps(f) \leq \tilde{O}(1/\eps^{2/5})$.
\end{corollary}

The proof of this fact follows essentially by approximating each $g_x(y) = f(x, y)$ by a linear combination of absolute value functions. Indeed, we have the following fact.

\begin{fact} \label{lem:convex-lipschitz-1d}
Let $f:[0,1]\to[-1,1]$ be convex and $1$-Lipschitz. Then there exist a probability measure $D$ on $[0,1]$ and a constant $c \in [-2, 2]$ such that
\[
f(x)=c+\mathbb{E}_{z\sim D}[|x-z|]\qquad\text{for all }x\in[0,1].
\]
\end{fact}

Since the proof is slightly technical and is not a statement about communication complexity, we defer its proof to \cref{sec:appendix-convex-lipschitz-1d}.

Now we can prove \cref{cor:convex-lipschitz}.
\begin{proof}[Proof of \cref{cor:convex-lipschitz}]
For each $x$, let $g_x(y) = f(x, y)$. By the previous lemma, for each $x$ there exist a constant $c_x \in [-2, 2]$ and a probability measure $D_x$ on $[0, 1]$ such that
\[g_x(y) = c_x + \E_{z \sim D_x}[|y - z|].\]
Therefore,
\begin{align*}
\E_{x \sim p, y \sim q}[f(x, y)]
&= \E_{x \sim p}[c_x] + \E_{\substack{x \sim p}} \left[\E_{z \sim D_x, y \sim q}[|y - z|] \right]\\
&= \E_{x \sim p}[c_x] + \E_{\substack{z \sim D, y \sim q}} [|y - z|],
\end{align*}
where $D$ is the distribution (known by Alice) obtained by first sampling $x \sim p$ and then $z \sim D_x$. Thus, Alice can send an estimate of $\E_{x \sim p}[c_x]$ to Bob using $O(\log \inv \e)$ bits, and then they can use the protocol from \cref{thm:abs-upper} to estimate $\E_{z \sim D, y \sim q} [|y - z|]$ to within an additive error of $O(\e)$. This gives the desired result.
\end{proof}

\begin{remark}
Consider the following communication problem: Alice has a convex function $f: [0, 1] \to [-1, 1]$ and Bob has a distribution $q$ on $[0, 1]$, and their goal is to estimate $\E_{y \sim q}[f(y)]$ to within additive error $\e$. Note that this is an instance of the expectation estimation problem for absolute difference, since by \cref{lem:convex-lipschitz-1d}, we can write $f$ as (a constant shift of) the expectation of $|x-y|$ as $x$ is drawn from some distribution $p$. The work of \cite{GORSS24} gives a one-way protocol for this problem with communication $\tilde{O}(1/\eps^{2/3})$, by finding an approximate linear basis for the space of bounded convex functions.
\end{remark}

\subsection{Smooth functions}
Next we consider smooth functions $f: [0, 1]^2 \to [-1, 1]$ whose $k$-th derivatives $\partial^i f / \partial^i x \partial^{k-i} y$ are bounded in $[-1, 1]$, for an absolute constant $k$. Let $\mD^{(k)}$ denote the set of such functions. We show a lower bound for such functions.

\begin{theorem} \label{thm:smooth-upper}
For $f \in \mD^{(k)}$, we have $\Rhat_\e(f) \leq \tilde{O}(\e^{-1/(k+1)})$.
\end{theorem}

First we show the following fact.
\begin{fact} \label{lem:bounded-lower-deriv}
If $f \in \mD^{(k)}$, then its $\ell$-th derivatives $\partial^\ell f / \partial x^i \partial y^{\ell-i}$ are bounded by $O(1)$ for all $\ell \le k$.
\end{fact}
We defer the proof of this fact to \cref{sec:appendix-bounded-lower-deriv}

\begin{proof}[Proof of \cref{thm:smooth-upper}]
Split the interval $\Y = [0, 1]$ into intervals $I_0, \dots, I_{\inv \al - 1}$ of size $\al = \e^{1/(k+1)}$, so that $I_j = [j\al, (j+1)\al]$. For each $x$, let $g_x(y) = f(x, y)$. Then, on each interval $I_j$ we can approximate $f(x, y) = g_x(y)$ with a Taylor series in $y$ centered at $j\al$:
\begin{equation*}
f(x, y) = \sum_{i=0}^{k-1} \f{g_x^{(i)}(j\al)}{i!}(y - j\al)^i + O(\al^k), \qquad y \in I_j,
\end{equation*}
where $g_x^{(i)}$ denotes the $i$-th derivative (in $y$) of $g_x$ (which by \cref{lem:bounded-lower-deriv} is bounded by $O(1)$). Therefore, we can get a uniform approximation to $f(x, y)$ everywhere as follows:
\begin{gather*}
f(x, y) = \fh(x, y) + O(\al^k), \qquad \text{for all $y$},\\
\fh(x, y) \coloneqq \sum_{j=0}^{\inv \alpha - 1} \sum_{i=0}^{k-1} \f{g_x^{(i)}(j\al)}{i!} \ind{y \in I_j} (y - j\al)^i.
\end{gather*}

Now, we will describe separate protocols to estimate $\E_{p,q}[\fh]$ and $\E_{p,q}[f-\fh]$. First, we observe that we can write
\begin{equation} \label{eq:double-exp-smooth}
\E_{p, q}[\fh] = \sum_{j=0}^{\inv \al - 1} \sum_{i=0}^{k-1} \f{1}{i!} \E_{x\sim p}[g_x^{(i)}(j\al)] \E_{y \sim q}[\ind{y \in I_j} (y - j\al)^i],
\end{equation}
so Alice may send the $k\inv \al = O(\inv \al)$ values $\E_{x\sim p}[g_x^{(i)}(j\al)]$ for each $i, j$ with precision $O(\e\al)$ (using $O(\log \inv \e)$ bits per value). With knowledge of $q$, Bob can then directly compute $\E_{p,q}[\fh]$ using \eqref{eq:double-exp-smooth}. Note that $g_x^{(i)}(j\al)$ and $(y - j\al)^i$ are both bounded by $O(1)$, so this indeed results in an $O(\e)$-approximation of $\E_{p, q}[\fh]$. 

It remains to approximate $\E_{p,q}[f-\fh]$; note that by the earlier discussion, $f - \fh$ is bounded by $O(\al^k)$. We split the interval $\X = \Y = [0, 1]$ into intervals $I'_0, \dots, I'_{\inv \e - 1}$ of length $\e$ which are a subdivision of the intervals $I_j$. Now, within each box $I'_r \times I'_s$ such that $I'_s \subset I_j$, we have
\begin{gather}
\fh(x, y) = \sum_{i=0}^{k-1} \f{1}{i!} \cdot g_x^{(i)}(j\al) \cdot (y - j\al)^i, \nonumber \\
\f{\del}{\del x} \fh(x, y) = \sum_{i=0}^{k-1} \f{1}{i!} \cdot \f{d}{dx}\p{g_x^{(i)}(j\al)} \cdot  (y - j\al)^i, \label{eq:smooth-partial-x} \\
\f{\del}{\del y} \fh(x, y) = \sum_{i=1}^{k-1} \f{1}{(i-1)!} \cdot g_x^{(i)}(j\al) \cdot  (y - j\al)^{i-1}. \label{eq:smooth-partial-y}
\end{gather}
Note that
\[\f{d}{dx}\p{g_x^{(i)}(j\al)} = \left.\f{\del^{i+1}}{\del x \del y^{i}} f(x, y)\right|_{y=j\al},\]
which is bounded by $O(1)$ by \cref{lem:bounded-lower-deriv}. Thus, the partial derivatives of $\fh$ given by \eqref{eq:smooth-partial-x} and \eqref{eq:smooth-partial-y} are all bounded by $O(1)$. Therefore, $\fh$ is $O(1)$-Lipschitz within each such box $I'_r \times I'_s$. Note that $f$ is also $O(1)$-Lipschitz everywhere by \cref{lem:bounded-lower-deriv}. Therefore, we have
\begin{equation*}
(f-\fh)(x, y) = (f-\fh)(r\e, s\e) + O(\e), \qquad \text{for all $(x, y) \in I'_r \times I'_s$}. 
\end{equation*}
Therefore, we may instead just consider the distributions $p$ and $q$ as their marginal distributions on the intervals $I'_r$ and $I'_s$ (in other words, we take the corresponding distributions over $r$ and $s$), and approximate $\E[(f-\fh)(r\e, s\e)]$. Recall that $f - \fh$ is bounded by $\al^k$. Then, scaling up by $\al^{-k}$, this is just an instance of expectation estimation for a general function on a discrete domain of size $2^n = 1/\e$ and with error $\e/\al^k$. By \cref{thm:mainupper}, it is possible to do this with communication complexity $\Ot(\al^k/\e)$. Thus, the overall communication complexity of this protocol is $\Ot(\inv \al + \al^k/\e) = \Ot(\e^{-1/(k+1)})$, as desired.

\end{proof}

%% file: spectral.tex
\section{A spectral technique for lower bounds}

So far, we have considered protocols for \jee for distributions $p \in \Delta_\X, q \in \Delta_\Y$. In the finite domain setting,  \jee is equivalent to computing $p^tMq$ where $M = [f(x,y)]$. A simple argument shows that we can relax the constraint on $p, q$ to allow vectors of bounded $\ell_1$ norm, which can have negative entries and get an additive guarantee that degrades with the norm. The stronger guarantee is helpful for proving lower bounds as we will soon see. We only need this in the discrete setting where $M$ is a matrix in $[-1,1]^{k \times k}$, but a similar statement holds for the continuous setting as well. 

\begin{lemma}
\label{lem:unbounded-pq}
    Let $M \in [-1,1]^{k \times k}$. Assume there is a communication protocol, where 
    Alice and Bob given $p \in \Delta_{[k]}$ and $q \in \Delta_{[k]}$ respectively,  use communication $c$ to compute an $\eps$ additive approximation to $p^tMq$ with probability $1 - \delta$. Then for any $\tilde{p} \in \RR^k$ and $\tilde{q} \in \RR^k$ where $\max(\|\tp\|_1, \|\tq\|_1) \leq B$ given to Alice and Bob respectively, there is a protocol that computes $\tilde{p}^tM\tilde{q}$ within additive error 
    \[ \eps' = \eps(\|\tp\|_1\|\tq\|_1 + \|\tp\|_1 + \|\tq\|_1)\]  
    using $O(c + \log(B/\eps))$ bits of communication.  
\end{lemma}
\begin{proof}
    We can write $\tp = \tp^+ - \tp^-, \tq = \tq^+ - \tq^-$ where each of $\tp^+, \tp^-, \tq^+, \tq^-$ only have non-negative entries. Alice defines distributions $p^+ = \tp^+/\|\tp^+\|_1$ and $p^- = \tp^-/\|\tp^-\|_1$, while Bob defines $q^+$ and $q^-$ similarly. We have
    \begin{align} 
    \label{eq:4-part}
    \tp M \tq &= (\tp^+)^t M \tq^+ + (\tp^-)^t M\tq^- - (\tp^+)^tM\tq^- - (\tp^-)^t M\tq^+
    \end{align}   
    Observe that 
    \begin{align} 
    \label{eq:part-1}
        (\tp^+)^t M \tq^+ = \|\tp^+\|_1(p^+)^tMq^+ \|\tq^+\|_1,
    \end{align}
    so it can be estimated by running the protocol for $(p^+, q^+)$ and rescaling. This suggests the following protocol:
    \begin{itemize}
    \item Alice and Bob run the protocol for distributions for all four pairs in $\{p^+, p^-\}\times \{q^+, q^-\}$ using $4c$ bits of communication. 
    \item Alice sends Bob the values $\|\tp^+\|_1, \|\tp^-\|_1$ to within an additive $\eps$ using $O(\log(B/\eps)$ bits of communication.
    \item Bob computes the output by suitably rescaling and then summing the outcomes of each pair. 
    \end{itemize}
    
    By Equation \eqref{eq:part-1}, $(\tp^+)^t M \tq^+ = a b c$, 
    where $a =\|\tp^+\|_1 $, $b = (p^+)^tMq^+$ and $c = \|\tq^+\|_1$. 
    Our protocol will approximate this by $a'b'c$, where $a'$ is the approximation of $\|\tp^+\|_1$ and $b'$ is the outcome of the protocol for $M$ on $p^+, q^+$, so $|b'| \leq 1$. Hence 
    \[ |a' - a| \leq \eps, \ |b ' - b| \leq \eps. \]
    We can bound
    \[ |a'b' - ab| \leq |a(b - b')| + |(a - a')b'| \leq \eps(\|\tp^+\|_1 + 1) \]
    hence the error in computing $abc$ is bounded by $\eps(\|\tp^+\|_1 + 1)\|\tq^+\|_1$. 
    Summing the error over all four terms gives a bound of
    \[ \eps(\|\tp^+\|_1 + \|\tp^-\|_1 + 2)\|(\tq^+\|_1 + \|\tq^-\|_1) \leq \eps(\|\tp\|_1 +2)\|\tq\|_1.\]
    
    We can also reverse the roles of Alice and Bob, by having Bob send $\|\tq^+\|_1, \|\tq^-\|_1$ to within an additive $\eps$, and Alice output the answer. This will give error 
    $\eps(\|\tq\|_1 + 2)\|\tp\|_1$. By randomizing between the two,\footnote{This can be easily derandomized by comparing $\|p\|_1$ and $\|q\|_1$} we can obtain the average error which is
    \[ \eps(\|\tp\|_1\|\tq\|_1 + \|\tp_1\|_1 + \|\tq_1\|_1).\]
\end{proof}

\paragraph{Gap Hamming} 
In the Gap-Hamming problem, the inputs are $a ,b \in \{1,-1\}^t$ and we are interested in approximating 
 \[\iprod{a}{b} = \sum_{i \in [t]}a_ib_i. \]
 Concretely, the goal of $\gh_{t,\eps}(a,b)$ is to decide whether $\iprod{a}{b} \geq \eps t$ or $\iprod{a}{b} \leq - \eps t$. 
The following lower bounds for Gap-Hamming were first shown by Chakrabarti and Regev.
\begin{theorem}
\label{thm:cr}
\cite{ChakrabartiR12}
    We have 
    \begin{align*} 
    \R(\gh_{t, \eps}) &= \Omega(\min(1/\eps^2,t)),\\
    \D(\gh_{t, 0.1}) &= \Omega(t).
    \end{align*}
\end{theorem}

\begin{definition}\label{def:lambdat}
    Given a matrix $A \in [-1,1]^{k \times k}$, and $t \leq \rank(A)$, let 
    \[ \lambda_{t}(A) = \sum_{i =1}^t\fr{\sigma_i}\]
    where the $\sigma_i$s are the singular values of $A$.  
\end{definition}

If $A$ is invertible, then $1/\sigma_i$ are the singular values of $A^{-1}$ and $\lambda_k(A)$ is the trace norm of $A^{-1}$. By allowing $t < k$, we allow the flexibility of finding a subspace of the domain over which the inverse is well behaved. 

Let us estimate $\lambda_t(k)$ for some simple matrices.
\begin{itemize}
\item For the $k \times k$ identity matrix $I_k$, $\lambda_t(I_k) =t$ for all $t \leq k$.
\item For $k \times k$ Hadamard matrix $H_k$ in $\{-1,1\}^k$, all the singular values are $\sqrt{k}$, so $\lambda_t(H_k) = t/\sqrt{k}$.
\end{itemize}

For the Hadamard matrix, we get $k^{1/2}$ for $t = k$. We show a general lower bound which implies that this is tight.

\begin{lemma}
\label{lem:lambda-lb}
    We have
    \begin{align}
    \label{eq:lambda-lb}
         \lambda_t(A) \geq \frac{t^{3/2}}{k}.
    \end{align}
\end{lemma}
\begin{proof}
    By the AM-HM inequality,
    \begin{align}
    \label{eq:am-hm}
        \lt(\sum_{i =1}^t\sigma_i \rt) \lt(\sum_{i=1}^t \fr{\sigma_i}\rt) \geq t^2. 
    \end{align}
    Since $A \in [-1,1]^{k \times k}$, 
    \[ \lt(\sum_i \sigma_i^2\rt)^{1/2} \leq \|A\|_F \leq k.\]
    By Cauchy-Schwartz,
    \[  \lt(\sum_{i=1}^t \sigma_i\rt) \leq \sqrt{t}(\sigma_i^2)^{1/2} \leq \sqrt{t}k\]
    Plugging this into Equation \eqref{eq:am-hm} gives the desired bound.
\end{proof}

Given a matrix $A \in [-1,1]^{k \times k}$, and a function $f:\X \times \Y \to [-1,1]$, we say $f$ realizes the matrix $A$ if there exist $\{x_i \in \X\}_{i \in [k]}$ and $\{y_j \in \Y\}_{j \in [k]}$ so that 
\[ A = [f(x_i, y_j)]_{i \in [k], j \in [k]}.\] 
In the case where the domain is discrete, this is the same as saying that the communication matrix $M_f$ contains $A$ as a submatrix. 

We show a generic lower bound on \jee for functions whose communication matrices have submatrices with nice spectral properties as quantified by $\lambda_t(\;)$ from \cref{def:lambdat}: 

\begin{theorem}
\label{thm:spectral}
If $f:\X \times \Y \rightarrow [-1,1]$ realizes a matrix $A \in [-1,1]^{k \times k}$, then for all $t \leq \rank(A)$,  
\[  \Rhat_{\epsilon}(f) \geq \Omega\left(\min\left(t, \frac{t^2}{\eps^2 k^2\lambda_t(A)^2}\right) \right) \]
\end{theorem}
\begin{proof}
Let $\{x_i \in \X\}_{i \in [k]}$ and $\{y_j \in \Y\}_{j \in [k]}$ be such that $A_{ij} = f(x_i, y_j)$. Let $r = rank(A)$, and let the singular-value-decomposition of $A$ be $A = U\Sigma V^t$ where we denote the columns of $U$ and $V$ by $\{u_i \in \RR^k\}_{i \in [r]}$ and $\{v_j \in \RR^k\}_{j \in [k]}$, and the singular values by $\{\sigma_i\}_{i \in [r]}$. 
    
    We will prove the bound using a reduction from Gap-Hamming on $\pmo^t$. 
    Given an instance $(a, b)$ of Gap-Hamming in $\pmo^t$, we define vectors $\tilde{p} \in\RR^k$ and $\tilde{q} \in \RR^k$ supported on the $x_i$'s and $y_j$'s respectively by
    \begin{align*} 
        \tilde{p} = \sum_{i \in [t]} \frac{a_iu_i}{\sqrt{\sigma_i}},\\
        \tilde{q} = \sum_{j \in [t]} \frac{b_jv_j}{\sqrt{\sigma_j}}.
    \end{align*}

\eat{
First observe that 
\begin{align*}
    \tilde{p}^t A \tilde{q} &= \sum_{i, j \in [k]} \frac{a_i b_j}{\sqrt{\sigma_i \sigma_j}} u_i^T A v_j.
\end{align*}
Now, by the properties of singular vectors, $u_i^T A v_j = \sigma_i$ if $i = j$ and $0$ otherwise. Thus, 

$$ \tilde{p}^t A \tilde{q} = \sum_{i \in [k]} a_i b_i. $$}

    If we define the vectors $c, d \in \RR^k$ by $c_i = a_i/\sqrt{\sigma_i}$ and $d_i = b_i/\sqrt{\sigma_i}$, then by orthogonality of $u_i$s and $v_j$s, we have $\tilde{p}^tU = c^t, V^t\tilde{q} = d$. Hence
    \begin{align*}
        \tilde{p}^tA\tilde{q} = \tilde{p}^tU\Sigma V^t\tilde{q} = c^t\Sigma d = \sum_{i \in [t]} a_ib_i.
    \end{align*}

We next bound the $\ell_1$ norms of $\tilde{p}, \tilde{q}$ so as to use \cref{lem:unbounded-pq}. We can bound $\|\tilde{p}\|_1$ by observing that
    \begin{align*}
        \|\tilde{p}_2\|_2 = \left(\sum_{i \in [t]} \frac{a_i^2}{\sigma_i}\right)^{1/2} = \sqrt{\lambda_{t}(A)}\\
        \|\tilde{p}\|_1 \leq \sqrt{k}\|\tilde{p}\|_2 \leq \sqrt{k \lambda_{t}(A)} 
    \end{align*}
    Similarly we get $
        \|\tilde{q}\|_1  \leq \sqrt{k \lambda_{t}(A)}$.
     By Lemma \ref{lem:unbounded-pq}, a protocol for $\EE_{f, \eps}$ can estimate $\iprod{a}{b}$ to within error 
     \[ \eps(k \lambda_{t}(A) + 2\sqrt{k \lambda_t(A)}) \leq 3\eps k \lambda_t(A)\] 
     where we use $k\lambda_t(A) \geq 1$ by Equation \eqref{eq:lambda-lb}.
     By theorem \ref{thm:cr}, this requires $\Omega(\min(t, t^2/(\eps^2 k^2\lambda_t(A)^2))$ bits of communication.
\end{proof}

An easy to use corollary is the following, which is obtained by taking $\eps$ sufficiently small in the statement above. 

\begin{corollary}
    \label{cor:spectral}
    If $f$ realizes the matrix $A \in [-1,1]^{k \times k}$, then for all $t \leq \rank(A)$,  $\EE(f, \eps)$ requires $\Omega(t)$ bits of communication for
\[ \eps \leq \frac{\sqrt{t}}{k}\fr{\lambda_t(A)}.\]
\end{corollary}
To get a strong lower bound, we need to find matrices $A$ realized by $f$ where $\lambda_t(A)$ is as small as possible.  As a sanity check, note that Theorem \ref{thm:mainupper} guarantees that $t$ bits of communication suffice to achieve error $O(\Row(f)/t)$. So the bound on $\eps$ given by Theorem \ref{thm:spectral} cannot be asympototically stronger  than $1/t$ (say $1/t^{0.9}$). Indeed, using the lower bound on $\lambda_t(A)$ from Equation \ref{eq:lambda-lb},
\[ \frac{\sqrt{t}}{k}\fr{\lambda_t(A)} \leq \frac{\sqrt{t}}{k}\frac{k}{t^{3/2}} = \fr{t}. \]
So the strongest lower bound this method can prove is $\Omega(1/\eps)$ for error $\eps$.

\subsection{Lower bounds for specific functions} \label{sec:specific-lower}
We will use the spectral technique to show tight lower bounds as a function of $1/\eps$ for the \jee problems considered so far. 

\begin{lemma} \label{lem:ip-lb}
    We have $\Rhat_{\eps}(\ip) \geq \Omega(\min(1/\eps, 2^n))$.
\end{lemma}
\begin{proof}
    The inner product function on $n$ bits realizes a Hadamard submatrix of size $2^\ell \times 2^\ell$ for every $\ell \leq n$.  We have shown that $\lambda_k(H_k) = \sqrt{k}$. This gives a lower bound of $\Omega(k)$ for $\eps \leq 1/k$ for every $k$ that is a power of $2$ bounded by $2^n$. This gives a $\Omega(1/\eps)$ lower bound for every $\eps \geq 2^{-n}$. 
\end{proof}

We show a tight lower bound for the Equality function on $n$ bits.
\begin{lemma}
\label{lem:eq-lb}
    We have $\Rhat_\eps(\EQ_n) \ge \Omega(\min(1/\eps^{2/3}, 2^n))$.
\end{lemma}
\begin{proof}
    Since the communication matrix for identity is $I_{2^n}$, it contains a $k \times k$ submatrix for $k \leq 2^n$. Setting $t = k$, since $\lambda_k(I_k) = k$, Theorem \ref{thm:spectral} implies that error $\eps \leq 1/k^{3/2}$ requires $\Omega(k)$ bits of communication. Parametrizing by $\eps$, this give a $1/\eps^{2/3}$ lower bound for all $\eps \geq 2^{-3n/2}$. For $\eps < 2^{-3n/2}$, we get a $\Omega(2^n)$ lower bound by considering $\eps = 2^{-3n/2}$. 
\end{proof}

A similar lower bound applies to any function that realizes a $k \times k$ matrix $A$ where $\lambda_k(A) = O(k)$. In particular, we get that the protocol for $\GT_n$ from \cref{thm:gt-strong} is tight upto a logarithmic factor:

\begin{lemma}
\label{lem:gt-lb}
    We have $\Rhat_{\eps}(\GT_n) \ge \Omega(\min(1/\eps^{2/3}, 2^n))$.
\end{lemma}
\begin{proof}
    $\GT_n$ realizes the upper triangular $k \times k$ matrix $T_k$ for all $k \leq 2^n$. 
    We claim that $\lambda_k(T_k) \leq 2k$. Theorem \ref{thm:spectral} then gives the lower bound by a similar argument to Lemma \ref{lem:eq-lb}. 

    The inverse of $T_k$ is the matrix $k \times k$ matrix $U_k$ where $U_{ii} = 1$ and $U_{i, i+1} = -1$. It follows that $\|U_k\| \leq 2$, hence $\sigma_i(U_k) \geq 1/\|U_k\|_2 \geq 1/2$ for all $i$. So $\lambda_k(\GT_k) \leq 2k$.
\end{proof}

The spectral method also gives tight lower bounds for some real-valued functions. We show a lower bound matching \cref{thm:smooth-upper}:
\begin{theorem} \label{thm:smooth-lower}
There is $f \in \mD^{(k)}$ such that $\Rhat_{\eps}(f) \ge \Omega(\e^{-1/(k+1)})$.
\end{theorem}

We begin by defining the function
\begin{equation*}
g(x, y) = 
\begin{cases}
\f 1C \p{\f12 + x}^{k+1} \p{\f12 - x}^{k+1} \p{\f12 + y}^{k+1} \p{\f12 - y}^{k+1}, & \text{$x, y \in \left[-\f12, \f12\right]$}, \\
0, & \text{else},
\end{cases}
\end{equation*}
where the constant $C$ is chosen such that the $k$-th derivatives of $g$ are bounded in $[-1, 1]$. Note that $g(x, y)$ has $k$-th derivatives (and all lower derivatives) equal to $0$ at the boundary of its support $[-\f12, \f12]^2$. Now, let $\hat g$ be the scaled version of $g$ given by
\begin{equation*}
    \hat g(x, y) = \al^k g(x/\al, y/\al).
\end{equation*}
Note that $\hat g$ is supported on $[-\al/2, \al/2]^2$, and has $k$-th derivatives bounded in $[-1, 1]$.

Now, we will prove \cref{thm:smooth-lower} by reducing it to the lower bound for $\ip$, given by \cref{lem:ip-lb}. Let $\al = \e^{1/(k+1)}$. Let $M_\ip$ be the inner product matrix with size $2^{n'} = 1/\al$. Then, write the following function $f$:
\begin{equation*}
    f(x, y) = \sum_{i, j \in [1/\al]} M_\ip(i, j) \hat g(x - i\al, y - j\al).
\end{equation*}
This is a combination of non-overlapping translations of $\hat g$, so it also has $k$-th derivatives bounded in $[-1, 1]$. Also, for all $i, j$, we have $f(i\al, j\al) = c \al^k M_\ip(i, j)$, for some constant $c > 0$. Now, for any instance of $\EE_{M_\ip}$, we can construct an instance of $\EE_{f, \e}$ by letting Alice and Bob take distributions $p'$ and $q'$ supported on the points $\{i\al\}_{i \in [1/\al]}$ and setting their probabilities to be the same as in the original instance. Then, we have
\begin{equation*}
    \E_{x \sim p', y \sim q'}[f(x, y)] = c\al^k \E_{x \sim p, y \sim q}[\ip(x, y)],
\end{equation*}
so if we can solve $\EE_{f, \e}$, we can solve $\EE_{M_\ip,\, \e'}$ for $e' = \e / (c \al^k)$. Therefore, by \cref{lem:ip-lb}, we have $\EE(f, \e) \ge \Omega(\min(1/\e',2^{n'})) = \Omega(\min(c\al^k/\e, 1/\al)) = \Omega(1/\e^{1/(k+1)})$, as desired.

Recall that $\ab(x,y) = |x - y|$ for $x,y, \in [0,1]$, and that $\EE(\ab, \eps) = \tilde{O}(1/\eps^{2/5})$ by \cref{thm:abs-upper}. We will show a tight lower bound (up to logarithmic terms).
For this we need to following technical claim.

\begin{claim}
\label{claim:gl}
Let $D_k \in [0,1]^{k \times k}$ where $D_k = |i -j|/k$ for $i, j \in [k]$. Then $\lambda_k(D_k) \leq\sqrt{3/2}k^2$.
\end{claim}
\begin{proof}
    The matrix $E_k = kD_k$ is the distance matrix for the path $P_k$ on $k$ vertices, whose $i,j$ entry is the shortest path distance between vertex $i$ and $j$. An explicit formula for its inverse is given by \cite[Lemma 1]{GL78}. We only state it in the special case of the path $P_n$. Let $v = e_1 + e_k \in \{0, 1\}^k$, and let $L(P_k)$ denote the Laplacian of the path $P$. Then
    \[ E_k^{-1} = -\fr{2}L(P_k) + \frac{vv^t}{2(k-1)}.\]
    $E_k^{-1}$ has $k$ entries bounded by $1$ and $2k$ entries bounded by $1/2$, so
    \[ \|E_k^{-1}\|_F \leq \sqrt{3k/2}.\]
    Note that $E_k^{-1} = D_k^{-1}/k$, so its singular values are $\mu_i = 1/(k\sigma_i)$ where $\sigma_i$ are the singular values of $D_k$. This implies that
    \begin{align*}
        \sum_{i \in [k]} \fr{k\sigma_i} = \sum_{i \in [k]}\mu_i \leq \sqrt{k} \|E_k^{-1}\|_F \leq \sqrt{3/2}{k}. 
    \end{align*}
    hence we get 
    \[ \lambda_k(D_k) =\sum_{i \in [k]}\fr{\sigma_i} \leq \sqrt{3/2}k^2. \]  
\end{proof}

\begin{lemma}
\label{lem:abs-lb}
    We have $\Rhat_\e(\ab) \ge \Omega(1/\eps^{2/5})$.
\end{lemma}
\begin{proof}
    For all $k \geq 1$, the function $\ab$ realizes the matrix $D_k$.  Applying Claim \ref{claim:gl} in Theorem \ref{thm:spectral} says that error $\eps \leq 1/(2k^{5/2})$ requires $\Omega(k)$ bits of communication, which translates to a lower bound of $\Omega(1/\eps^{5/2})$ as a function of $\eps$. 
\end{proof}

\subsection{Lower bounds for Boolean functions with high rank}

We now come to our main lower bound in this section, which is that no Boolean function of large rank is much easier than $\EQ$ for expectation estimation. This is in contrast to the bounded setting, where natural functions like $\ab$ have much lower complexity (as a function of the error $\varepsilon$. Our lower bound says that a lower bound of $\tilde{\Omega}(1/\eps^{2/3})$ holds for sufficiently large $\eps$ ($\eps \geq \Omega(1/\rank(f)^{3/2})$). 

\begin{theorem} \label{thm:boolean-lb}
    Let $f:\zo^n\times \zo^n \to \zo$ be a (non-constant) Boolean function. Then 
    \[ \Rhat_\e(f) \ge \Omega(\min((1/\eps)/\log(1/\eps))^{2/3}, \rank(f)/\log(\rank(f))).\] 
\end{theorem}

Our key technical lemma is the following:

\begin{lemma}
\label{lem:boolean-lb}
    For $k \geq 4$, let $B \in [-1,1]^k$ be a bounded matrix with $|\mathrm{det}(B)| \geq 1$. For $t= \lfloor k/(\log k)\rfloor$, $\lambda_t(B) \leq 4t$.
\end{lemma}
\begin{proof}
    Let $\sigma_1 \geq \cdots \geq \sigma_k$ be the singular values of $B$. We claim that for $t$ as above, $\sigma_t \geq 1/4$. Assuming this holds true,
    \[\lambda_t(B) =  \sum_{i \in [t]} \fr{\sigma}_i \leq \frac{t}{\sigma_t} \leq 4t.\]

    To prove the claim, assume for contradiction that $\sigma_t < 1/4$. We can bound the product of the first $t$ singular values using the AM-GM inequality applied to the $\sigma_i^2$ as
    \begin{align*}
        \prod_{i \in [t]}\sigma_i \leq \lt(\fr{t}\sum_{i \in [t]}\sigma_i^2\rt)^{t/2} = \lt(\frac{\|B\|^2_F}{t}\rt)^{t/2} \leq \lt(\frac{k}{\sqrt{t}}\rt)^t. 
    \end{align*}
    Since all singular values for $i \geq t +1$ are bounded by $1/4$, their product is bounded as
    \begin{align*}
        \prod_{i=t +1}^k\sigma_i < \fr{4^{k - t}}
    \end{align*}
    But then we have
    \begin{align}
    \label{eq:det-bound}
        \prod_{i \in [k]}\sigma_i < \lt(\frac{k}{\sqrt{t}}\rt)^t\fr{4^{k -t}} = \lt(\frac{k/\sqrt{t}}{4^{k/t -1}}\rt)^t.
    \end{align}
    For $k \geq 4$, we claim that 
    \[ \frac{k}{\sqrt{t}} \leq \sqrt{2k \log(k)} \leq \frac{k^2}{4} \leq  4^{k/t -1}.\]
    The first and second inequality hold for $k \geq 4$, and the last holds since $t \leq k/\log(k)$, so $k/t \geq \log(k)$.
    By Equation \eqref{eq:det-bound} we get $\prod_{i \in [k]}\sigma_i < 1$, which contradicts our assumption about the determinant of $B$.
\end{proof}

\begin{proof}[Proof of Theorem \ref{thm:boolean-lb}]
    We may assume $\rank(f) \geq 4$, since otherwise we only need to show a lower bound of some constant, which holds trivially since $f$ is not constant. 
    For $4 \leq k \leq \rank(f)$, let $A \in \zo^{k \times k}$ be a full rank Boolean matrix realized by $f$. Since $A$ is Boolean, $\det(A) = 1$. By Lemma \ref{lem:boolean-lb}, 
    \[ \lambda_t(A) \leq 4k/\log(k) \ \text{for} \ t = \lfloor k/(\log(k)) \rfloor \geq k/(2\log(k)). \]
    We observe that
    \[\frac{\sqrt{t}}{k \lambda_t(A)} \geq \frac{\sqrt{k/\log(k)}}{k}\frac{\log(k)}{4k}= \frac{\sqrt{\log(k)}}{4k^{3/2}}. \]
    Theorem \ref{thm:spectral} implies a lower bound of $\Omega(t) = \Omega(k/\log(k))$ for error $\eps \leq \sqrt{\log(k)}/(4k^{3/2})$. By choosing 
    \[ k = \min(\rank(f), \Omega(\log(1/\eps)^{1/3}/\eps^{2/3}) \]
    we derive a lower bound of
    \[ \Omega(k/\log(k)) = \Omega(\min((1/\eps)/\log(1/\eps))^{2/3}, \rank(f)/\log(\rank(f)). \]
\end{proof}

%% file: lower_bound_general_discrepancy.tex
\newcommand{\cS}{\mathcal{S}}
\newcommand{\cT}{\mathcal{T}}
\newcommand{\nb}{t}

\section{Lower bound from general discrepancy over product distributions}

\newcommand{\discf}{\mathrm{Disc}_f}
\newcommand{\disct}{\mathrm{Disc}_{\Theta}}
\newcommand{\Disc}{\mathrm{Disc}}
\newcommand{\dinf}{\mathrm{D}_{\infty}}
\newcommand{\Hinf}{\mathrm{H}_{\infty}}
\newcommand{\IP}{\mathsf{IP}}


In this section we prove \cref{th:discintro}. We begin by recalling the definition of discrepancy of a function, a parameter that has been extensively used in  communication complexity. In the following, to be concrete, we will work with functions $f:[q] \times [q] \rightarrow \{1,-1\}$ (i.e., we identify the domains as $[q]$). 

A general product distribution over $[q] \times [q]$ is written as $\Theta_X \times \Theta_Y$ for some marginal distributions $\Theta_X, \Theta_Y$.  However, for notational simplicity, we will use $\Theta \times \Theta$ to denote a product distribution over $[q] \times [q]$.  In fact, \cref{thm:main} and its proof  apply as is to arbitrary product distributions $\Theta_X \times \Theta_Y$, where $X$ and $Y$ don't have the same marginal distribution.

\begin{definition}
\label{def:rect-disc}
    For $f:[q]\times [q] \rightarrow \{1,-1\}$, the discrepancy $\disct(f)$ under a distribution $\Theta$ is the following supremum over all rectangles $R \subset [q] \times [q]$,

    \[ \disct (f) = \sup_{R \subset [q] \times [q]} \Pr_{\Theta \times \Theta}[(x,y) \in R] \cdot \big |\E[f(x,y) | x,y \in R] \big| \]
    
    The discrepancy of the function $\Disc^{\times}(f)$ under product distributions is,

    \[ \Disc^{\times}(f) = \inf_{\text{ product distributions } \Theta} \left(\disct(f) \right) .\]
  
\end{definition}

Fix a function \(f:[q]\times[q]\to\{-1,1\}\) (the proof also holds for any \(f\) with \(|f|\le 1\)).
Recall that in \jee, Alice and Bob receive distributions \(\mu_A,\mu_B\) on
\([q]\), respectively, and their goal is to estimate $\E_{\mu_A, \mu_B}[f]$. 
\eat{
\[
  f(\mu_A,\mu_B)\;=\;\E_{a\sim\mu_A,\;b\sim\mu_B}\big[f(a,b)\big]
\]
to additive error at most \(\epsilon\).}
%

%
%

This section is devoted to proving the following lower bound based on discrepancy.
\begin{theorem}
  \label{thm:main}
    Fix a $f: [q] \times [q] \to \{-1,1\}$ and a product distribution $\Theta \times \Theta$ \footnote{We denote a product distribution by $\Theta \times \Theta$ for notational simplicity, but in general a product distribution $\Theta_X \times \Theta_Y$ may have different marginals} over $[q] \times [q]$.
    Let 
\[ k = \frac{1}{240}\log \left(\frac{1}{\Disc_{\Theta}(f)} \right)\]
and fix any $\epsilon > \left(\frac{1}{\Disc_{\Theta }(f)}\right)^{1/200}$.
  Let $\Pi$ be any randomized protocol that on every pair of input distributions $(\mu_A,\mu_B)$,
  outputs \(\Pi(\mu_A,\mu_B)\) with
  \[
    \Pr\big[\,|f(\mu_A,\mu_B)-\Pi(\mu_A,\mu_B)|>\epsilon\,\big]\;\le\; \epsilon^2,
  \]
  where Alice communicates at most \(m_A \cdot k\) bits and Bob communicates at most \(m_B \cdot k\) bits.  
  Then we have
  \[
    \lceil m_A \rceil \cdot \lceil m_B \rceil \;\ge\; \Omega\!\left(\frac{1}{\epsilon^2}\right).
  \]
\end{theorem}
We first state a couple of immediate consequences.
\begin{corollary} \label{cor:disclower}
Fix a function $f: [q] \times [q] \to \{-1,1\}$ 
Then for all $\epsilon > (\Disc^{\times}(f))^{1/200}$,
\begin{align*}
    \Rhat_{\epsilon}(f) & \geq \Omega\left(\frac{1}{\epsilon} \cdot \log\left(\frac{1}{\Disc^{\times}(f)}\right) \right),\\
\Rhow_{\epsilon}(f) & \geq \Omega\left(\frac{1}{\epsilon^2} \cdot \log\left(\frac{1}{\Disc^{\times}(f)}\right)   \right).
\end{align*}

\end{corollary}
\begin{proof}
The total number of bits communicated by the protocol is at least $(m_A + m_B) \cdot k \geq (\lceil m_A \rceil + \lceil m_B \rceil -2) \cdot k$ which is at least $ \Omega(\frac{1}{\epsilon}) \cdot k$ since $\lceil m_A \rceil \lceil m_B \rceil \geq \Omega(1/\epsilon^2)$.  This implies the first inquality.

For the second inequality, observe that in a one-way protocol, Alice sends a message and Bob outputs the answer (or vice versa).  Bob only communicates $\log (1/\epsilon)$ bits to announce the answer implying that $\lceil m_B \rceil = 1$.  The inequality follows from \cref{thm:main}
\end{proof}
\begin{corollary}\label{cor:ipcor}
For all $\epsilon > 2^{-\Theta(\sqrt{n})}$,
\begin{align*}
    \Rhat_{\epsilon}(\IP_n) & \geq \Omega\left(\frac{n}{\epsilon} \right), \\
\Rhow_{\epsilon}(\IP_n) & \geq \Omega\left(\frac{n}{\epsilon^2} \right).\\
\end{align*}
    
\end{corollary}

\begin{proof}
Both inequalities follow by using $\disct(\IP_n) \leq 2^{-n/2}$ under the uniform distribution over $\{0,1\}^n \times \{0,1\}^n$ along with \cref{cor:disclower}
\end{proof}

\input{lowerboundproofoverview}

\subsection{Preliminaries}
We make a few definitions and prove a few simple facts needed to carry out the proof of \cref{thm:main}.
\paragraph{Max-Divergence}
\begin{definition}
    For two distributions $\Phi$ and $\Theta$ over a domain $[q]$, let
    \[ \dinf(\Phi \| \Theta) = \sup_{z} \log\left( \frac{\Pr_{Z \sim \Phi} [Z = z]}{\Pr_{Z \sim \Theta}[Z = z]}\right) \]    
\end{definition}
We will often write $\dinf(Z \| \Theta)$ for a random variable $Z$ over the domain, to denote the $\dinf(\Phi \| \Theta)$ of its probability density function $\Phi$.
Any base of logarithm is fine as long as it is consistent; we use base 2 so \(\log q=n\).
\begin{lemma}[Typical $D_{\infty}$-bound]
\label{lem:typical-Dinf}
Let $(X,Z)$ be random variables with $Z$ supported on a finite set $\mathcal D$.
For $C>0$, define the \emph{typical} set
\[
\mathcal T_C \;:=\; \Bigl\{\, z\in\mathcal D \ \Big|\ 
D_{\infty}\!\bigl(P_{X\mid Z=z}\,\Vert\,P_X\bigr)\ \le\ \log|\mathcal D|\,+\,C \Bigr\}.
\]
Then
\[
\Pr[\,Z\notin \mathcal T_C\,]\ \le\ 2^{-C}.
\]
\end{lemma}

\begin{proof}
Write $g(z):=\exp\!\bigl(D_{\infty}\!(P_{X\mid Z=z}\Vert P_X)\bigr)
=\sup_x \frac{P_{X\mid Z=z}(x)}{P_X(x)}$. Note that for every $z$,
\[
P_Z(z)\,g(z)
\;=\; P_Z(z)\,\sup_x \frac{P_{X\mid Z=z}(x)}{P_X(x)}
\;=\; \sup_x \frac{P_{X,Z}(x,z)}{P_X(x)}
\;\le\; 1,
\]
since $P_{Z\mid X}(z\mid x)\le 1$ for all $x,z$. Summing over $z\in\mathcal D$ gives
\[
\mathbb E\big[g(Z)\big]
\;=\;\sum_{z} P_Z(z)\,g(z)
\;=\;\sum_{z} \sup_x \frac{P_{X,Z}(x,z)}{P_X(x)}
\;\le\; |\mathcal D|.
\]
By Markov's inequality,
\[
\Pr\!\left[g(Z)>\,|\mathcal D|\,2^{C}\right]
\;\le\; \frac{\mathbb E[g(Z)]}{|\mathcal D|\,2^{C}}
\;\le\; 2^{-C}.
\]
Equivalently, $\Pr\!\left[D_{\infty}(P_{X\mid Z}\Vert P_X) > \log|\mathcal D|+C\right]\le 2^{-C}$,
which is the claim.
\end{proof}

\subsection{Discrepancy and max-divergence}

\begin{definition}[Discrepancy w.r.t.\ $\Theta$]
For $f:[q]\times[q]\to \mathbb{R}$ and a baseline distribution $\Theta$ on $[q]$, define
\[
\mathrm{Disc}_\Theta(f)\ :=\ \sup_{\text{rectangles }R\subseteq[q]\times[q]}
\ \Big|\ \E_{(x,y)\sim\Theta\times\Theta}\big[\,f(x,y)\,\mathbf 1_{R}(x,y)\,\big]\ \Big|.
\]
Equivalently, $\mathrm{Disc}_\Theta(f)=\sup_{R} \Pr_{\Theta\times\Theta}[R]\cdot \big|\E_{\Theta\times\Theta}[f\mid (x,y)\in R]\big|$.
\end{definition}

Here is a simple lemma saying that discrepancy gives reasonable bounds for correlation of $f$ with product distributions as long as their ratios with $\Theta$ are bounded. Similar results have been proved in the literature in various contexts, we include the proof here for completeness.

\begin{lemma}[Discrepancy bound under $D_\infty$-closeness to $\Theta$]
\label{lem:disc-Dinfty}
Let $\Theta$ be a distribution on $[q]$ and let $f:[q]\times[q]\to \mathbb{R}$. For distributions $P_X,Q_Y$ on
$[q]$ with finite max-divergences
\[
D_\infty(P_X\Vert\Theta)\;=\;\log\max_{x}\frac{P_X(x)}{\Theta(x)}\;=\;C_X\;<\infty,
\qquad
D_\infty(Q_Y\Vert\Theta)\;=\;\log\max_{y}\frac{Q_Y(y)}{\Theta(y)}\;=\;C_Y\;<\infty,
\]
we have
\begin{equation}\label{eq:disc-Dinfty}
\Big|\,\E_{x\sim P_X,\,y\sim Q_Y}\big[f(x,y)\big]\,\Big|
\;\le\;
\mathrm{Disc}_\Theta(f)\cdot 2^{\,C_X+C_Y}.
\end{equation}
\end{lemma}

\begin{proof}
Define the likelihood ratios $r_X(x):=P_X(x)/\Theta(x)$ and $r_Y(y):=Q_Y(y)/\Theta(y)$. Then
$0\le r_X(x)\le 2^{C_X}$ and $0\le r_Y(y)\le 2^{C_Y}$ for all $x,y$, and
$\sum_x \Theta(x)r_X(x)=\sum_y\Theta(y)r_Y(y)=1$. We can write
\[
\E_{x\sim P_X,\,y\sim Q_Y}\big[f(x,y)\big]
=\E_{(x,y)\sim \Theta\otimes\Theta}\!\big[f(x,y)\,r_X(x)\,r_Y(y)\big].
\]
Decompose $r_X,r_Y$ as follows,
\[
r_X(x)=\int_{0}^{2^{C_X}}\!{\bf 1}\{r_X(x)>s\}\,ds,\qquad
r_Y(y)=\int_{0}^{2^{C_Y}}\!{\bf 1}\{r_Y(y)>t\}\,dt.
\]
Substituting and interchanging the integrals and expectations,
\[
\E_{\Theta\otimes\Theta}\!\big[f\,r_X r_Y\big]
=\int_{0}^{2^{C_X}}\!\!\int_{0}^{2^{C_Y}}
\E_{\Theta\otimes\Theta}\!\Big[f(x,y)\,{\bf 1}\{r_X(x)>s\}\,{\bf 1}\{r_Y(y)>t\}\Big]\,dt\,ds.
\]
For each $(s,t)$, the indicator ${\bf 1}\{r_X(x)>s\}{\bf 1}\{r_Y(y)>t\}$ is the indicator of a rectangle
$R_{s,t}=A_s\times B_t$ with $A_s:=\{x:r_X(x)>s\}$ and $B_t:=\{y:r_Y(y)>t\}$. By the definition of
$\mathrm{Disc}_\Theta(f)$ (unnormalized rectangle discrepancy under $\Theta$),
\[
\Big|\E_{\Theta\otimes\Theta}\big[f(x,y)\,{\bf 1}_{R_{s,t}}(x,y)\big]\Big|
\le \mathrm{Disc}_\Theta(f)\quad\text{for every rectangle }R_{s,t}.
\]
Therefore,
\[
\Big|\,\E_{x\sim P_X,\,y\sim Q_Y} f(x,y)\,\Big|
\le \int_{0}^{2^{C_X}}\!\!\int_{0}^{2^{C_Y}} \mathrm{Disc}_\Theta(f)\,dt\,ds
= \mathrm{Disc}_\Theta(f)\cdot 2^{C_X+C_Y},
\]
which is \eqref{eq:disc-Dinfty}. 
\end{proof}

\paragraph{Conjunctive block-wise dense distributions}


We restate the definition of conjunctive block-wise dense distributions (cbd) of \cite{GLMWZ16} in terms of max-divergences, in place of min-entropy.  For our theorems, it becomes necessary to develop decompositions parametrized by max-divergences.

\begin{definition}[cbd distributions]
A distribution $\mathcal D$ over $[q]^t$ is $k$-\emph{conjunctive block-wise dense} ($k$-cbd)
if there exist a set $J\subseteq[t]$ of \emph{fixed} coordinates with $|J|\le t/10$ and a value
$\alpha\in [q]^J$ such that $X_J=\alpha$ under $\mathcal D$, and for every
$S\subseteq [t]\setminus J$ we have
\[
\dinf\!\big((\mathcal D)_S\ \Vert\ \Theta^S\big)\ \le\ k\,|S|.
\]
We call the indices in $J$ fixed and those in $[t]\setminus J$ unfixed.
\end{definition}

We adapt the  decomposition of arbitrary distributions into cbd distributions, parametrized by max-divergence $\dinf$.  Our adaptation below follows along the lines of Lemma 6.1 in \cite{KMR21}.

\begin{lemma}[cbd decomposition]
\label{lem:cbd-decomp}
Let $\Theta$ be any distribution on $[q]$. Let $\cD$ be a distribution on $[q]^t$
with $\dinf(\cD\Vert\Theta^t)\le kt$. Then there exists a partition
\[
[q]^t\ =\ \Bigl(\ \bigcup_{i=1}^{N}\cS_i\Bigr)\ \cup\ \cS_{\mathrm{err}}
\]
such that:
\begin{itemize}
\item for each $i$, the conditioned distribution $\cD\mid\cS_i$ is $10k$-cbd (with respect to $\Theta$);
\item $\Pr_{x\sim\cD}[x\in \cS_{\mathrm{err}}]\ \le\ 1/100$.
\end{itemize}
\end{lemma}

\begin{proof}
We will use three basic facts about $\dinf$.

\medskip
\noindent\textbf{(F1) Marginal monotonicity.}
For any $S\subseteq[t]$ and any distribution $\mu$ on $[q]^t$,
\(
\dinf(\mu_S\Vert\Theta^S)\ \le\ \dinf(\mu\Vert\Theta^t).
\)
This is because
$$\mu_S(x_S)/\Theta^S(x_S) \!=\! \sum_{y}\!\big(\mu(x_S,y) \Theta^{S^c}(y)/\sum_{y} \Theta^t(x_S,y) \Theta^{S^c}(y)\big)
\le \sup_{x}\mu(x)/\Theta^t(x).$$

\smallskip
\noindent\textbf{(F2) Conditioning inequality.}
Let $\mu$ be a distribution on coordinates $S$, fix $I\subseteq S$ and $\alpha\in[q]^I$
with $\mu[X_I=\alpha]>0$, and set $R=S\setminus I$. Then
\begin{equation}
\label{eq:cond}
\dinf\!\big(\mu\mid X_I=\alpha\ \Vert\ \Theta^R\big)
\ \le\
\dinf(\mu\Vert\Theta^S)\ -\ \log\frac{\mu_I(\alpha)}{\Theta^I(\alpha)}.
\end{equation}
Indeed,
\(
\frac{(\mu\mid X_I=\alpha)(y)}{\Theta^R(y)}
= \frac{\mu(y,\alpha)}{\Theta^S(y,\alpha)}\cdot \frac{\Theta^I(\alpha)}{\mu_I(\alpha)}.
\)

\smallskip
\noindent\textbf{(F3) Inclusion-maximal heavy implies immediate CBD.}
Fix $\rho>0$ and say that $(I,\alpha)$ is \emph{$\rho$-heavy} for $\mu$ if
$\log\big(\mu_I(\alpha)/\Theta^I(\alpha)\big)\ \ge\ \rho\,|I|$.
If $(I,\alpha)$ is $\rho$-heavy and \emph{inclusion-maximal} (no strict superset of $I$
is $\rho$-heavy), then the conditional
$\mu' := \mu\mid (X_I=\alpha)$ is $\rho$-CBD on the remaining coordinates:
for every $U\subseteq S\setminus I$,
\[
\dinf(\mu'_U\Vert\Theta^U)\ \le\ \rho\,|U|.
\]
\emph{Proof of (F3):} Otherwise, for some $U$ we would have, by \eqref{eq:cond},
\(
\dinf(\mu_{I\cup U}\Vert\Theta^{I\cup U}) >
\dinf(\mu'_U\Vert\Theta^U)+\log\frac{\mu_I(\alpha)}{\Theta^I(\alpha)}
> \rho|U|+\rho|I|=\rho|I\cup U|,
\)
contradicting inclusion-maximality of $I$. \qed

\medskip
\noindent\textbf{The algorithm (single error set).}
Set
\[
\rho\ :=\  10k\ +\ \frac{10\log 100}{t}
\qquad(\text{note that }\rho\le 20k \text{ for large } k).
\]
Initialize the remainder $R_0:=[q]^t$ and $i\gets 0$.
While $\cD(R_i)>1/100$ and the induced $\mu_i:=\cD\mid R_i$ admits a nonempty
$\rho$-heavy pair $(I,\alpha)$ on its unfixed coordinates, do:
\begin{enumerate}
\item Pick an \emph{inclusion-maximal} $\rho$-heavy $(I,\alpha)$ for $\mu_i$;
let $C_{i+1}:=\{x\in R_i:\ x_I=\alpha\}$.
\item By (F3), the piece $\cD\mid C_{i+1}$ is $\rho$-CBD; moreover, since $\rho\le 20k$,
it is also $20k$-CBD.
\item Remove it: set $R_{i+1}:=R_i\setminus C_{i+1}$ and increment $i$.
\end{enumerate}
When the loop stops:
\begin{itemize}
\item If it stopped because \emph{no} $\rho$-heavy pair remains in $R_i$, then
$\cD\mid R_i$ is $\rho$-CBD by definition, hence $20k$-CBD; output it as a final good piece and set $\cS_{\mathrm{err}}:=\varnothing$.
\item Otherwise, it stopped because $\cD(R_i)\le 1/100$; in that case declare
$\cS_{\mathrm{err}}:=R_i$. This is the \emph{only} error set we produce.
\end{itemize}
By construction the output pieces $\{C_j\}$ together with $R_i$ (if nonempty) form a disjoint partition of $[q]^t$, and
$\cD(\cS_{\mathrm{err}})\le 1/100$.

\medskip
\noindent\textbf{Bound on the number of fixed coordinates per piece.}
Fix any set $C$ produced in the loop, arising from a $\rho$-heavy pair $(I,\alpha)$ for $\mu:=\cD\mid R$ with $\cD(R)>1/100$.
We claim that $|I|\le t/10$, so the good piece $\cD\mid C$ is $(20k)$-cbd with at most $t/10$ fixed coordinates.

Indeed, by (F1) and the definition of heaviness,
\[
\rho\,|I|\ \le\ \dinf(\mu_I\Vert\Theta^I)\ \le\ \dinf(\mu\Vert\Theta^t).
\]
Using the standard conditioning bound (apply \eqref{eq:cond} with $I=\emptyset$ and event $R$),
\[
\dinf(\mu\Vert\Theta^t)\ =\ \dinf(\cD\mid R\ \Vert\ \Theta^t)
\ \le\ \dinf(\cD\Vert\Theta^t)\ +\ \log\frac{1}{\cD(R)}
\ \le\ kt\ +\ \log 100.
\]
Combining and recalling our choice of $\rho$,
\[
\rho\,|I|\ \le\ kt+\log 100
\qquad\Longrightarrow\qquad
|I|\ \le\ \frac{kt+\log 100}{\,10k+\tfrac{10\log 100}{t}\,}\ =\ \frac{t}{10}.
\]

The same argument applies verbatim to the final remainder $R$ if the loop terminated because no $\rho$-heavy pair remains: in that case $I=\varnothing$ and there are no fixed coordinates.

\medskip
\noindent\textbf{Conclusion.}
Each output set $C$ yields a good piece $\cD\mid C$ that is $\rho$-CBD and hence $20k$-CBD, with at most $t/10$ fixed coordinates. If the loop ends by exhaustion of heavy pairs, the final remainder is also good; otherwise the single error set has $\cD$-mass at most $1/100$. This proves the lemma.
\end{proof}

\subsection{Proof of Theorem \ref{thm:main}}
\begin{proof} 
For a vector \(x=(x_1,\dots,x_{\nb})\in [q]^{\nb}\), define the empirical distribution \(\mu_x\) on \([q]\) by
sampling a uniformly random coordinate: \(\beta\sim \mu_x\) means that \(\beta=x_I\) for a uniform
\(I\in[\nb]\).


Set \(\nbx=  10 \lceil m_A \rceil\)  and \(\nby= 10 \lceil  m_B \rceil \). 
Without loss of generality, assume that $\lceil m_A \rceil \cdot \lceil m_B \rceil \le \frac{1}{100 \epsilon^2}$.
Consider the ``hard'' product distribution on inputs
\((\mu_x,\mu_y)\) defined by
\begin{itemize}
  \item \(\cD_A:\) sample \(x\leftarrow [q]^{\nbx}\) from distribution $\Theta^{\nbx}$  and give Alice \(\mu_x\);
  \item \(\cD_B:\) sample \(y\leftarrow [q]^{\nby}\) from distribution $\Theta^{\nby}$ and give Bob \(\mu_y\).
\end{itemize}
Let \(\Pi_A,\Pi_B\) denote the message transcripts sent by Alice and Bob. For
\(\pi_A\in\{0,1\}^{m_A \cdot k}\) define
\[
  \mathrm{vol}(\pi_A)\;=\;\Pr_{x\sim\cD_A,\,y\sim\cD_B}\big[\Pi_A(\mu_x,\mu_y)=\pi_A\big],
\]
and analogously for \(\mathrm{vol}(\pi_B)\).
Applying Lemma \ref{lem:typical-Dinf} to the random variables $\pi_A$ and $\pi_B$, we conclude that with probability at least $\frac{3}{4}$, one has typical transcript $\pi = (\pi_A,\pi_B)$ such that conditioning on \(\pi\), the product input distribution becomes \(\cD'_A\times \cD'_B\) where
\[
  D_\infty(\cD'_A \| \Theta^{\nbx})\;\le\;  m_A \cdot k + 2 \quad\text{and}\quad
    D_\infty(\cD'_B \| \Theta^{\nby})\;\le\;  m_B \cdot k + 2 
\]



For \(x\in[q]^{\nbx}\), \(y\in[q]^{\nby}\) define
\[
  F(x,y)\;=\;\sum_{i=1}^{\nbx}\sum_{j=1}^{\nby} f(x_i,y_j).
\]
Note that \(f(\mu_x,\mu_y)=\frac{1}{\nbx\nby}F(x,y)\). We will prove the following  lower bound on variance of $F$ under distributions that have high min-entropy with respect to $\Theta$.

\begin{lemma}
\label{lem:variance}
Fix integers $k, \nbx,\nby$ and a probability distribution $\Theta$ over $[q]$.
Let \(\cD'_A\) be a distribution over \([q]^{\nbx}\) with \(D_\infty(\cD'_A \| \Theta^{\nbx})\le (\nbx/10) \cdot k \),
and let \(\cD'_B\) be a distribution over \([q]^{\nby}\) with \(D_\infty(\cD'_B \| \Theta^{\nby})\le (\nby/10) \cdot k\).

For every function $f: [q] \times [q] \to \{-1,1\}$, we have
\[
  \Var_{x\sim\cD'_A,\,y\sim\cD'_B}\left[ \sum_{i \in [\nbx], j \in [\nby]} f(x_i,y_j)\right]
  \;\ge\;
  \frac{\nbx \nby}{4}  \ -\ 2^{6k} \nbx^{34} \nby^{34}  \,
 \Disc_\Theta(f)\ .
\]
\end{lemma}
We will defer the proof of the Lemma to the next subsection.  Assuming the lemma, we can finish the proof of \cref{thm:main}.
Substituting the values for $k, \nbx$ and $\nby$, and using $\nbx \cdot \nby \leq 100 \lceil m_A \rceil \cdot \lceil m_B \rceil \leq \frac{1}{\epsilon^2}$, we get that 
\(\Var_{x\sim\cD'_A,y\sim\cD'_B}\big[F(x,y)\big]\ge \frac{\nbx\nby}{5}\) using bounds on  discrepancy of $f$; this immediately implies
\[
  \Var_{(\mu_x,\mu_y)\sim \cD'_A\times\cD'_B}\big[f(\mu_x,\mu_y)\big]
  \;=\;\Var\!\left[\frac{F}{\nbx\nby}\right]\;\ge\;\frac{1}{5\,\nbx\nby}.
\]
Since the protocol achieves squared-loss (and hence variance) at most \(\epsilon^2\) conditioned on
\(\pi\), we obtain
\[
  \epsilon^2\;\ge\;\Var\big[f(\mu_x,\mu_y)\big]\;\ge\;\frac{1}{5\,\nbx\nby},
\]
which yields \(\nbx\nby\ge 1/(5\epsilon^2)\), and hence
\(\lceil m_A \rceil \lceil m_B \rceil = (\nbx/10)(\nby/10) \ge 1/(500 \epsilon^2)\) as required.
\end{proof}


\subsection{Lower bounding variance of large rectangles}

We will show \cref{lem:variance} in two stages: first for the case of conjunctive block-wise dense distributions, which is then generalized to arbitrary large rectangles. 

\paragraph{Conjunctive block-wise dense distributions}


\begin{lemma}[cbd case]
\label{lem:cbd-variance}
Fix integers $k,\nbx,\nby,n>0$.
Let $\cD_A$ be $k$-cbd on $[q]^{\nbx}$ and $\cD_B$ be $k$-cbd on $[q]^{\nby}$.
\[
  \Var_{x\sim\cD_A,\,y\sim\cD_B}\left[ \sum_{i \in [\nbx], j \in [\nby]} f(x_i,y_j)\right]
  \;\ge\;
  \tfrac{\nbx \nby}{3} \ -\ 2^{6k} \nbx^{34} \nby^{34}  \,
 \Disc_\Theta(f)\ .
\]
\end{lemma}

\begin{proof}

Write $J_X\subseteq[\nbx]$ and $J_Y\subseteq[\nby]$ for the (at most) $t_A/10$ and $t_B/10$
\emph{fixed} blocks of $\cD_A$ and $\cD_B$, respectively, and set
\[
\overline{J_X}:=[\nbx]\setminus J_X,\quad \overline{J_Y}:=[\nby]\setminus J_Y,\quad \text { so that }
|\overline{J_X}|\ (\ge 0.9\nbx),\ \ |\overline{J_Y}|\ (\ge 0.9\nby).
\]
For $x\sim \cD_A$ and $y\sim \cD_B$ define
\[
B[i,j]\ :=\ f(x_i,y_j)\in \{-1,1\},\qquad 
F(x,y)\ :=\ \sum_{i=1}^{\nbx}\sum_{j=1}^{\nby} B[i,j].
\]

Our goal is to lower bound the variance of $F(x,y)$.  Note that $F$ can be written as,
\begin{equation}\label{eq:expandf}
F(x,y) = \underbrace{\sum_{i \in J_X, j \in J_Y} F(x_i, y_j)}_{\text{$x_i,y_j$ fixed}} + \underbrace{\sum_{i \in J_X, j \in \overline{J}_Y} F(x_i, y_j)}_{\text{$x_i$ fixed}} +
\underbrace{\sum_{i \in \overline{J}_X, j \in J_Y} F(x_i, y_j)}_{\text{$y_j$ fixed}} +
\underbrace{\sum_{i \in \overline{J}_X, j \in \overline{J}_Y} F(x_i, y_j)}_{\text{$x_i,y_j$ high min-entropy}}
\end{equation}
Much of the variance of $F$ is derived from the final term in the above equation.  The first term is constant and does not affect the variance.  To control the effect of the second and third terms, we will condition the distributions $\cD_A$, $\cD_B$ on the values of these terms..

\paragraph{Step 1: Two one-sided ``typical'' events that preserve independence.}

Define the events
\[
\mathcal{E}_Y(\beta_Y):=\left\{\,  \sum_{(i,j) \in J_X \times \overline{J_Y}} f(x_i,y_j)  =\beta_Y\,\right\},\qquad
\mathcal{E}_X(\beta_X):=\left\{\,  \sum_{(i,j) \in \overline{J_X} \times J_Y} f(x_i,y_j)   =\beta_X\,\right\}
\]
Since the coordinates in $J_X$ are constant, the event $\mathcal{E}_Y$ depends only on $y$.  Similarly, the event $\mathcal{E}_X$ only depends on $x$.  Hence conditioning on the event
\[
\mathcal{E}\ :=\ \mathcal{E}_X(\beta_X)\ \wedge\ \mathcal{E}_Y(\beta_Y)
\]
preserves independence of $(x_i)_{i\in\overline{J_X}}$ and $(y_j)_{j\in\overline{J_Y}}$.

The $k$-cbd property implies that 
\begin{align}
\forall i\in\overline{J_X}:\quad
&D_\infty\big(\mathrm{Law}(x_i)\ \Vert\ \Theta\big)\ \le\ k,
&&
\forall\,i\neq i'\in\overline{J_X}:\ \
D_\infty\big(\mathrm{Law}(x_i,x_{i'})\ \Vert\ \Theta^{\otimes 2}\big)\ \le\ 2k,\\
\forall j\in\overline{J_Y}:\quad
&D_\infty\big(\mathrm{Law}(y_j)\ \Vert\ \Theta\big)\ \le\ k,
&&
\forall\,j\neq j'\in\overline{J_Y}:\ \
D_\infty\big(\mathrm{Law}(y_j,y_{j'})\ \Vert\ \Theta^{\otimes 2}\big)\ \le\ 2k.
\end{align}
We will use Lemma \ref{lem:typical-Dinf} to bound the $D_{\infty}$ quantities after conditioning on $\mathcal{E}$.  
Let $C$ be a parameter that we fix as\  $C = 4 \log t_X \cdot t_Y$.
The event $\mathcal{E}$ corresponds to conditioning on at most $b := 2\log(t_X \cdot t_Y) $ bits.
By applying Lemma \ref{lem:typical-Dinf}
along with union bound on all singletons $i \in \overline{J_X} \cup \overline{J_Y}$ and pairs $i,i'$ in  $\overline{J_X}$ and $\overline{J_Y}$, we conclude the following.

For \emph{all but a $(t_X^2 + t_Y^2) \cdot 2^{-C}$-fraction} of $\beta_X$ (resp.\ $\beta_Y$), the following
$D_\infty$-bounds hold under $\mathcal{E}$:
\begin{align}
\forall i\in\overline{J_X}:\quad
&D_\infty\big(\mathrm{Law}(x_i\mid \mathcal{E})\ \Vert\ \Theta\big)\ & \le\ k+ b + C & < k + 2C \\
\forall\,i\neq i'\in\overline{J_X}:\ \
&D_\infty\big(\mathrm{Law}((x_i,x_{i'})\mid \mathcal{E})\ \Vert\ \Theta^{\otimes 2}\big)\ & \le\ 2k+ b +C &< 2k + 2C ,
\label{eq:X-Dinf}\\
\forall j\in\overline{J_Y}:\quad
&D_\infty\big(\mathrm{Law}(y_j\mid \mathcal{E})\ \Vert\ \Theta\big)\ &\le\ k+ b +C & < k+ 2C,\\
\forall\,j\neq j'\in\overline{J_Y}:\ \
&D_\infty\big(\mathrm{Law}((y_j,y_{j'})\mid \mathcal{E})\ \Vert\ \Theta^{\otimes 2}\big)\ & \le\ 2k+ b +C &< 2k + 2C.
\label{eq:Y-Dinf}
\end{align}
We henceforth fix $\beta_X,\beta_Y$ with
these properties; all expectations below are conditioned on $\mathcal{E}$, and $(x_i)$ and $(y_j)$
for unfixed indices remain independent.

\paragraph{Step 2: Variance expansion under $\mathcal{E}$.}
Rewrite \eqref{eq:expandf} as 
\[
\widehat{F}\ :=\ \sum_{i\in\overline{J_X}}\sum_{j\in\overline{J_Y}} f(x_i,y_j),
\qquad\text{so that}\qquad
F\ =\ \widehat{F}\ +\ \text{(a constant under $\mathcal{E}$)}.
\]
Thus $\Var(F\mid\mathcal{E})=\Var(\widehat{F}\mid\mathcal{E})$, and
\begin{align}
\Var(\widehat{F}\mid\mathcal{E})
&=\sum_{i\in\overline{J_X},\,j\in\overline{J_Y}} \E\big[f(x_i,y_j)^2\mid\mathcal{E}\big]
+\sum_{i\in\overline{J_X},\,j\neq j'} \E\big[f(x_i,y_j)f(x_i,y_{j'})\mid\mathcal{E}\big]\notag\\
&\quad+\sum_{i\neq i',\,j\in\overline{J_Y}} \E\big[f(x_i,y_j)f(x_{i'},y_j)\mid\mathcal{E}\big]
+\sum_{i\neq i',\,j\neq j'} \E\big[f(x_i,y_j)f(x_{i'},y_{j'})\mid\mathcal{E}\big]\notag\\
&\quad-\Big(\sum_{i\in\overline{J_X},\,j\in\overline{J_Y}} \E\big[f(x_i,y_j)\mid\mathcal{E}\big]\Big)^2.
\label{eq:var-expand}
\end{align}
Since $|f| =  1$, the first term contributes exactly $|\overline{J_X}|\cdot| \overline{J_Y}|$.  We will argue that the remaining terms are negligible.

\paragraph{Step 3: Bounding the correlation terms.}
We now bound \emph{each} expectation in the remaining three lines of
\eqref{eq:var-expand} by an application of Lemma~\ref{lem:pair-robust}, Lemma~\ref{lem:disc-Dinfty} and Lemma~\ref{lem:pairpair-from-disc}.

\smallskip
\emph{(i) Same $i$, different $j\neq j'$.}
Fix $i\in\overline{J_X}$ and $j\neq j'\in\overline{J_Y}$. Under $\mathcal{E}$, the random variables
$X:=x_i$ and $(Y_1,Y_2):=(y_j,y_{j'})$ are independent. By using \eqref{eq:X-Dinf}--\eqref{eq:Y-Dinf} along with Lemma~\ref{lem:pair-robust} (applied conditionally on $\mathcal{E}$) yields
\begin{equation}\label{eq:same-i}
\big|\E[f(x_i,y_j)f(x_i,y_{j'})\mid\mathcal{E}]\big|
\ \le\ \Disc_\Theta(f)\cdot 2^{\,3k+8C + 1}\;+\;3\cdot 2^{-C}.
\end{equation}

\smallskip
\emph{(ii) Same $j$, different $i\neq i'$.}
By symmetry, the same bound as \eqref{eq:same-i} holds:
\begin{equation}\label{eq:same-j}
\big|\E[f(x_i,y_j)f(x_{i'},y_j)\mid\mathcal{E}]\big|
\ \le\ \Disc_\Theta(f)\cdot 2^{\,3k+8C+1}\;+\;3\cdot 2^{-C}.
\end{equation}

\smallskip
\emph{(iii) Different $i\neq i'$, different $j\neq j'$.}
Here $(x_i,x_{i'})$ is independent of $(y_j,y_{j'})$ under $\mathcal{E}$; by
\eqref{eq:X-Dinf}--\eqref{eq:Y-Dinf},
\[
D_\infty\big((x_i,x_{i'})\Vert\Theta^{\otimes 2}\big)\le 2k+2C,\qquad
D_\infty\big((y_j,y_{j'})\Vert\Theta^{\otimes 2}\big)\le 2k+2C.
\]
Apply Lemma~\ref{lem:pairpair-from-disc} we get
\begin{equation}\label{eq:diff-diff}
\big|\E[f(x_i,y_j)f(x_{i'},y_{j'})\mid\mathcal{E}]\big|
\ \le\ \Disc_\Theta(f)\cdot 2^{\,4k+6C}+ 2 \cdot 2^{-C}.
\end{equation}

\smallskip
\emph{(iv) Mean term.}
By Lemma~\ref{lem:disc-Dinfty} with $D_\infty(x_i\Vert\Theta)\le k + 2C$ and
$D_\infty(y_j\Vert\Theta)\le k+2C$ we have
\begin{equation}\label{eq:mean-term}
\big|\E[f(x_i,y_j)\mid\mathcal{E}]\big|
\ \le\ \Disc_\Theta(f)\cdot 2^{\,2k +4C}.
\end{equation}

\paragraph{Step 3: Assemble \eqref{eq:var-expand}.}
Using $f^2\le 1$ for the diagonal and \eqref{eq:same-i}–\eqref{eq:mean-term} for the other sums,
\begin{align*}
\Var(\widehat{F}\mid\mathcal{E})
&\ge |\overline{J_X}| | \overline{J_Y}|
 - t_X\,t_Y(t_Y-1)\,\Big(\Disc_\Theta(f) 2^{6k+8C}+3\cdot 2^{-C}\Big)\\
&\quad - t_Y\,t_X(t_X-1)\,\Big(\Disc_\Theta(f) 2^{6k+8C}+3\cdot 2^{-C}\Big)\\
&\quad - t_X(t_X-1)\,t_Y(t_Y-1)\,\Big(\Disc_\Theta(f) 2^{4k+6C} + 2\cdot 2^{-C}\Big)\\
&\quad - \big(t_X t_Y\cdot \Disc_\Theta(f) 2^{2k+4C}\big)^2.
\end{align*}
We obtain
\[
\Var(\widehat{F}\mid\mathcal{E})
\ \ge\ |\overline{J_X}| \cdot |\overline{J_Y}|\ -\ t_X^2 t_Y^2\,
\Big(\ \Disc_\Theta(f)\cdot 2^{\,6k+8C}\ +\ 3\cdot 2^{-C}\ \Big).
\]

By construction of the typical $\beta_X,\beta_Y$, $\Pr[\mathcal{E}]\ge 1-2^{-C}-2^{-C}\ge \frac{1}{2}$
(for $C\ge 1$). Therefore,
\[
\Var(F)\ \ge\ \Pr[\mathcal{E}]\,\Var(\widehat{F}\mid\mathcal{E})
\ \ge\ \tfrac{1}{2}\left(|\overline{J_X}| \cdot |\overline{J_Y}| \ -\ t_X^2 t_Y^2\,
\Big(\ \Disc_\Theta(f)\cdot 2^{\,6k+8C}\ +\ 3\cdot 2^{-C}\ \Big)\right).
\]
Recall that we set $C = 4 \log (t_X \cdot t_Y)$, and $|\overline{J_X}| > 0.9 t_X$, $|\overline{J_Y}| \geq 0.9 t_Y$. The above bound can be simplified to,

\[
\Var(F)\ \ge\  \tfrac{t_X t_Y}{3} \ -\ 2^{6k} t_X^{34} t_Y^{34}  \,
 \Disc_\Theta(f)\ .
\]

\end{proof}


\paragraph{General Distributions}
We restate Lemma \ref{lem:variance} here and finish the proof.
\begin{lemma} (restatement of Lemma \ref{lem:variance})
\label{lem:variance1}
Fix integers $k, \nbx,\nby$ and a probability distribution $\Theta$ over $[q]$.
Let \(\cD'_A\) be a distribution over \([q]^{\nbx}\) with \(D_\infty(\cD'_A \| \Theta^{\nbx})\le (\nbx/10) \cdot k \),
and let \(\cD'_B\) be a distribution over \([q]^{\nby}\) with \(D_\infty(\cD'_B \| \Theta^{\nby})\le (\nby/10) \cdot k\).

For every function $f: [q] \times [q] \to \{-1,1\}$, we have
\[
  \Var_{x\sim\cD_A,\,y\sim\cD_B}\left[ \sum_{i \in [\nbx], j \in [\nby]} f(x_i,y_j)\right]
  \;\ge\;
  \tfrac{\nbx \nby}{4} \ -\ 2^{6k} \nbx^{34} \nby^{34}  \,
 \Disc_\Theta(f)\ .
\]
\end{lemma}

\begin{proof}
Apply Lemma~\ref{lem:cbd-decomp} to each of \(\cD'_A\) and \(\cD'_B\), producing partitions
\(\{\cS^{A}_i\}\) and \(\{\cS^{B}_j\}\). Then \([q]^{\nbx}\times [q]^{\nby}\) decomposes as
\[
  \bigcup_{i,j} \big(\cS^{A}_i\times \cS^{B}_j\big)\;\cup\; T_{\mathrm{err}},
  \qquad
  \Pr[T_{\mathrm{err}}]\;\le\; 2/100.
\]
On each piece \(\cS^{A}_i\times \cS^{B}_j\), both conditionals are  \(20k\)-cbd. so
Lemma~\ref{lem:cbd-variance} gives
\[ \Var\big[F\mid x\in \cS^{A}_i,\,y\in \cS^{B}_j\big]  \ge\  \tfrac{t_X t_Y}{3} \ -\ 2^{120k} t_X^{34} t_Y^{34}  \,
 \Disc_\Theta(f)\ .\]
Variance decreases under conditioning, hence
\begin{align*}
  \Var_{x\sim\cD'_A,y\sim\cD'_B}[F(x,y)]
  &\ge \sum_{i,j} \Pr[x\in\cS^{A}_i \wedge y\in\cS^{B}_j]\cdot
        \Var\big[F\mid x\in \cS^{A}_i,\,y\in \cS^{B}_j\big]\\
  &\ge \Big(1-\Pr[T_{\mathrm{err}}]\Big)\cdot \left(\tfrac{\nbx \nby}{3} \ -\ 2^{120k} t_X^{38} t_Y^{38}  \,
 \Disc_\Theta(f)\ \right) \\ 
 & \;\ge\; \tfrac{\nbx \nby}{4} \ -\ 2^{120k} t_X^{34} t_Y^{34}  \,
 \Disc_\Theta(f)\  ,
\end{align*}
 This completes the proof of
Lemma~\ref{lem:variance}.
\end{proof}

\subsection{Expectations of pairwise products}\label{sec:disctensor}
In this section we show that discrepancy has some \emph{tensoring properties}. Specifically, we already saw that if $f:[q] \times [q] \rightarrow \{1,-1\}$ has small discrepancy, then $\E[f(X,Y)]$ is small when $X, Y$ have controlled max-divergences. We show that the same extends to coupled products of the form $f(X_1,Y_1) f(X_2, Y_2)$ or $f(X,Y_1) f(X,Y_2)$ also have low bias when $(X_1, X_2), (Y_1, Y_2)$ are independent and have controlled max-divergences. 
\begin{lemma}[Decoupled pairwise product is small from single-pair discrepancy]
\label{lem:pairpair-from-disc}
Let $\Theta$ be a baseline distribution on a finite $[q]$, and let
$f:[q]\times[q]\to\{-1,1\}$. Suppose $(X_1,X_2)$ and $(Y_1,Y_2)$ are
\emph{independent} $[q]^2$-valued random variables satisfying
\begin{equation}\label{eq:pair-entr}
  D_\infty\!\big(\mathrm{Law}(X_1,X_2)\ \Vert\ \Theta\times\Theta\big)\ \le\ C_X,
  \qquad
  D_\infty\!\big(\mathrm{Law}(Y_1,Y_2)\ \Vert\ \Theta\times\Theta\big)\ \le\ C_Y.
\end{equation}
Then for every $C>0$,
\begin{equation}\label{eq:pairpair-target-correct}
  \Big|\, \E\big[ f(X_1,Y_1)\, f(X_2,Y_2) \big] \,\Big|
  \;\le\;
  2^{\,C_X + C_Y + 2C}\,\disct(f)\;+\;2\cdot 2^{-C}.
\end{equation}
\end{lemma}

\begin{proof}
By independence of $(X_1,X_2)$ and $(Y_1,Y_2)$,
\begin{equation}\label{eq:outer-decouple}
  \E\big[f(X_1,Y_1)\,f(X_2,Y_2)\big]
  \;=\;
  \E_{X_2,Y_2}\Big[\, f(X_2,Y_2)\cdot \E\big[f(X_1,Y_1)\mid X_2,Y_2\big] \,\Big].
\end{equation}

\paragraph{Step 1: Typical conditionings for $X_2$ and $Y_2$.}
Apply Lemma~\ref{lem:typical-conditional-Dinf} (the ``typical conditional $D_\infty$'' bound) to the
pairs $(X_1,X_2)$ and $(Y_1,Y_2)$ with the parameter $C>0$. We obtain sets
\[
\mathcal{G}_X\subseteq [q],\qquad \mathcal{G}_Y\subseteq [q],
\]
such that
\begin{align}
&\Pr[X_2\in\mathcal{G}_X]\ \ge\ 1-2^{-C},\qquad
\forall\,x_2\in\mathcal{G}_X:\ 
D_\infty\!\big(\mathrm{Law}(X_1\mid X_2=x_2)\ \Vert\ \Theta\big)\ \le\ C_X + C,
\label{eq:typX}\\
&\Pr[Y_2\in\mathcal{G}_Y]\ \ge\ 1-2^{-C},\qquad
\forall\,y_2\in\mathcal{G}_Y:\ 
D_\infty\!\big(\mathrm{Law}(Y_1\mid Y_2=y_2)\ \Vert\ \Theta\big)\ \le\ C_Y + C.
\label{eq:typY}
\end{align}
Let $\mathcal{G} := \mathcal{G}_X\times \mathcal{G}_Y$ and $\mathcal{B}:=[q]^2\setminus\mathcal{G}$.
By a union bound,
\begin{equation}\label{eq:bad-mass}
\Pr\!\left[(X_2,Y_2)\in\mathcal{B}\right]\ \le\ 2\cdot 2^{-C}.
\end{equation}

\paragraph{Step 2: Bound the inner conditional mean on the good event.}
Fix any $(x_2,y_2)\in\mathcal{G}$. Under this conditioning, $X_1$ and $Y_1$ are independent, with
marginals
\[
P_X := \mathrm{Law}(X_1\mid X_2=x_2),\qquad Q_Y := \mathrm{Law}(Y_1\mid Y_2=y_2),
\]
and by \eqref{eq:typX}–\eqref{eq:typY} they satisfy
\[
D_\infty(P_X\Vert\Theta)\ \le\ C_X + C,\qquad
D_\infty(Q_Y\Vert\Theta)\ \le\ C_Y + C.
\]
Hence Lemma~\ref{lem:disc-Dinfty} (discrepancy bound under $D_\infty$-closeness) yields
\begin{equation}\label{eq:inner-good}
\Big|\,\E\big[f(X_1,Y_1)\ \big|\ X_2=x_2,\ Y_2=y_2\big]\,\Big|
\ \le\ \disct(f)\cdot 2^{\,C_X + C_Y + 2C}.
\end{equation}

\paragraph{Step 3: Finish by averaging and accounting for the bad event.}
From \eqref{eq:outer-decouple}, using $|f|\le 1$, we get
\begin{align*}
\Big|\,\E\big[f(X_1,Y_1)\,f(X_2,Y_2)\big]\,\Big|
&\le \E\Big[\,\big|\E[f(X_1,Y_1)\mid X_2,Y_2]\big|\ \mathbf{1}_{\{(X_2,Y_2)\in\mathcal{G}\}}\Big]
    \;+\; \Pr\big[(X_2,Y_2)\in\mathcal{B}\big]\\
&\le \disct(f)\cdot 2^{\,C_X + C_Y + 2C}\ \Pr\!\left[(X_2,Y_2)\in\mathcal{G}\right]
    \;+\; 2\cdot 2^{-C}\\
&\le \disct(f)\cdot 2^{\,C_X + C_Y + 2C}\ +\ 2\cdot 2^{-C},
\end{align*}
using \eqref{eq:inner-good} on $\mathcal{G}$ and \eqref{eq:bad-mass}. This is exactly
\eqref{eq:pairpair-target-correct}.
\end{proof}

\begin{lemma}[Robust mixed-pair bound via typical conditionings and rectangle discrepancy]
\label{lem:pair-robust}
Let $\Theta$ be a baseline distribution on a finite domain $[q]$, and let
$f:[q]\times[q]\to\{-1,1\}$. Let $X$ and $(Y_1,Y_2)$ be independent random variables with
$X\in[q]$, $(Y_1,Y_2)\in[q]^2$ such that
\[
  D_\infty(X \,\Vert\, \Theta)\ \le\ k
  \qquad\text{and}\qquad
  D_\infty\!\big((Y_1,Y_2) \,\Vert\, \Theta\times\Theta \big)\ \le\ 2k.
\]
Then for every $C>0$,
\begin{equation}\label{eq:target-robust}
  \Big|\,\E\big[f(X,Y_1)\,f(X,Y_2)\big]\,\Big|
  \ \le\ 2^{\,3k+2C}\,\Disc_{\Theta}(f)\ +\ 3\cdot 2^{-C}.
\end{equation}
\end{lemma}

\begin{proof}
Write the target as
\begin{equation}\label{eq:outer}
  \E\big[f(X,Y_1)\,f(X,Y_2)\big]
  \;=\; \E_{Y_2}\Big[\, \E\big[\, f(X,Y_1)\,f(X,Y_2)\ \big|\ Y_2\,\big] \,\Big].
\end{equation}

\paragraph{Step 1: Typical values of $Y_2$.}
Apply Lemma~\ref{lem:typical-conditional-Dinf} with the pair $(Y_1,Y_2)$, baseline $\Theta$, and
parameter $C>0$. We obtain a \emph{good} set $\mathcal G_Y\subseteq[q]$ with
\begin{equation}\label{eq:good-y2}
\Pr\big[Y_2\in\mathcal G_Y\big]\ \ge\ 1-2^{-C},\qquad
\forall\,y_2\in\mathcal G_Y:\quad
D_\infty\!\big(\mathrm{Law}(Y_1\mid Y_2=y_2)\ \Vert\ \Theta\big)\ \le\ 2k + C.
\end{equation}

\paragraph{Step 2: Binning $f(X,Y_2)$ and heavy-bin event.}
Call a pair $y_2 \in [q]$ and $z \in \{-1,1\}$ to be \emph{heavy} if, 
$\Pr[f(X,Y_2)=z\mid Y_2=y_2]\ge 2^{-C}$, and \emph{light} otherwise. Define the
\emph{light-bin event}
\[
  \mathsf{Light}\ :=\ \big\{ (X,Y_2):\ (Y_2,f(X,Y_2))\ \text{is light} \big\}.
\]
By definition,
\begin{equation}\label{eq:light-mass}
  \Pr[\mathsf{Light}\mid Y_2=y_2]\ \le\ 2^{-C}
  \quad\text{for every }y_2,
  \qquad\Rightarrow\qquad \Pr[\mathsf{Light}]\ \le\ 2^{-C}.
\end{equation}

\paragraph{Step 3: Conditioning cost for $X$ on heavy bins.}
Since $X$ is independent of $Y_1,Y_2$, conditioning on a
\emph{heavy} $(Y_2=y_2,f(X,y_2) = z)$ reveals at most $C$ bits of information about $X$. Thus
\begin{equation}\label{eq:X-Dinf-heavy}
  \forall\ (y_2,z)\ \text{heavy}:\quad
  D_\infty\!\big(\mathrm{Law}(X\mid Y_2=y_2,\ Z=z)\ \Vert\ \Theta\big)\ \le\ k + C.
\end{equation}
Moreover, $X$ remains independent of $Y_1$ under the conditioning $(Y_2=y_2,f(X,y_2)=z)$ 


\paragraph{Step 4: Bounding the main term on the good/heavy event via discrepancy.}
Split according to the heavy-bin event and the typicality of $y_2$:

\begin{align*}
 \E \big[f(X,Y_1)f(X,y_2)\ \big|\ Y_2=y_2\big] &= \E\big[f(X,Y_1)\, f(X,y_2)\ \mathbf{1}_{\neg\mathsf{Light}}\ \big|\ Y_2=y_2\big] \\
\ & \ \ \ +\ \E\big[f(X,Y_1)\, f(X,y_2)\ \mathbf{1}_{\mathsf{Light}}\ \big|\ Y_2=y_2\big].
\end{align*}

If $y_2\in\mathcal G_Y$ and we are on the heavy event, then by \eqref{eq:good-y2} and
\eqref{eq:X-Dinf-heavy} we have two independent marginals with
\[
D_\infty\big(\mathrm{Law}(X\mid y_2,f(X,y_2))\Vert\Theta\big)\ \le\ k+C,\qquad
D_\infty\big(\mathrm{Law}(Y_1\mid y_2)\Vert\Theta\big)\ \le\ 2k+C,
\]
so Lemma~\ref{lem:disc-Dinfty} (discrepancy bound under $D_\infty$) gives
\[
\Big|\,\E\big[f(X,Y_1)\ \big|\ Y_2=y_2,\ f(X,y_2)\big]\,\Big|
\ \le\ \Disc_\Theta(f)\cdot 2^{\, (k+C) + (2k+C)}\ \le\ \Disc_\Theta(f)\cdot 2^{\,3k+2C}.
\]
Multiplying by $|f(X,y_2)|= 1$ preserves the bound, and averaging over heavy bins (still with
$y_2\in\mathcal G_Y$) yields
\begin{equation}\label{eq:good-heavy-bound}
\Big|\,\E\big[f(X,Y_1)\, f(X,y_2)\big]\ \mathbf{1}_{\neg\mathsf{Light}}\ \big|\ Y_2=y_2\big]\,\Big|
\ \le\ \Disc_\Theta(f)\cdot 2^{\,3k+2C}.
\end{equation}
On the light event, use $|f|\le 1$ and \eqref{eq:light-mass} to get, for any $y_2$,
\begin{equation}\label{eq:light-term}
\Big|\,\E\big[f(X,Y_1)\, f(X,y_2)\big]\ \mathbf{1}_{\mathsf{Light}}\ \big|\ Y_2=y_2\big]\,\Big|
\ \le\ \Pr[\mathsf{Light}\mid Y_2=y_2]\ \le\ 2^{-C}.
\end{equation}

\paragraph{Step 5: Combine and average over $Y_2$.}
Putting  \eqref{eq:good-heavy-bound}, and \eqref{eq:light-term} together, for
$y_2\in\mathcal G_Y$ we have
\[
\Big|\,\E\big[f(X,Y_1)\,f(X,Y_2)\ \big|\ Y_2=y_2\big]\,\Big|
\ \le\ \Disc_\Theta(f)\cdot 2^{\,3k+2C}\ +\ 2^{-C}\ +\ 2^{-C}.
\]
If $y_2\notin\mathcal G_Y$, the inner absolute value is trivially at most $1$. Hence, using
\eqref{eq:outer} and \eqref{eq:good-y2},
\begin{align*}
\Big|\,\E\big[f(X,Y_1)\,f(X,Y_2)\big]\,\Big|
&\le \Pr[Y_2\in\mathcal G_Y]\cdot \Big(\Disc_\Theta(f)\cdot 2^{\,3k+3C+1}+2\cdot 2^{-C}\Big)
\ +\ \Pr[Y_2\notin\mathcal G_Y]\cdot 1\\
&\le \Disc_\Theta(f)\cdot 2^{\,3k+2C}\ +\ 2\cdot 2^{-C}\ +\ 2^{-C} = \Disc_\Theta(f)\cdot 2^{\,3k+2C}\ +\ 3\cdot 2^{-C},
\end{align*}
This is exactly
\eqref{eq:target-robust}.
\end{proof}

\begin{lemma}[Typical conditional $D_\infty$ from joint $D_\infty$]
\label{lem:typical-conditional-Dinf}

Let $\Theta$ be a baseline distribution on a finite domain $[q]$. Let $(X_1,X_2)$ be a pair of
$[q]$-valued random variables with joint law $P_{X_1,X_2}$ such that
\[
D_\infty\!\big(P_{X_1,X_2}\ \Vert\ \Theta\times\Theta\big)\ \le\ 2k,
\quad\text{i.e.}\quad
P_{X_1,X_2}(x_1,x_2)\ \le\ 2^{2k}\,\Theta(x_1)\Theta(x_2)\ \ \text{for all }x_1,x_2\in[q].
\]
Then for every $C>0$, the set of \emph{typical}$^\ast$ values
\[
\mathcal T_C\ :=\ \Bigl\{\,x_2\in[q]\ :\
D_\infty\!\big(P_{X_1\mid X_2=x_2}\ \Vert\ \Theta\big)\ \le\ 2k + C\ \Bigr\}
\]
has $P_{X_2}$-mass at least $1-2^{-C}$. Equivalently,
\[
\Pr_{X_2}\big[\,D_\infty\!\big(P_{X_1\mid X_2}\ \Vert\ \Theta\big)\ >\ 2k + C\,\big]\ \le\ 2^{-C}.
\]
\end{lemma}

\smallskip
\noindent{\bf Proof.}
Fix $x_2\in[q]$. For any $x_1\in[q]$,
\[
\frac{P_{X_1\mid X_2=x_2}(x_1)}{\Theta(x_1)}
\;=\;\frac{P_{X_1,X_2}(x_1,x_2)}{\Theta(x_1)\,P_{X_2}(x_2)}
\;\le\; \frac{2^{2k}\,\Theta(x_1)\Theta(x_2)}{\Theta(x_1)\,P_{X_2}(x_2)}
\;=\; 2^{2k}\cdot \frac{\Theta(x_2)}{P_{X_2}(x_2)}.
\]
Taking the maximum over $x_1$ gives
\begin{equation}\label{eq:cond-ratio}
\max_{x_1}\frac{P_{X_1\mid X_2=x_2}(x_1)}{\Theta(x_1)}
\;\le\; 2^{2k}\cdot \frac{\Theta(x_2)}{P_{X_2}(x_2)}.
\end{equation}
Therefore, whenever $x_2$ satisfies
\begin{equation}\label{eq:typical-x2}
P_{X_2}(x_2)\ \ge\ 2^{-C}\,\Theta(x_2),
\end{equation}
we get from \eqref{eq:cond-ratio} that
\[
\max_{x_1}\frac{P_{X_1\mid X_2=x_2}(x_1)}{\Theta(x_1)}\ \le\ 2^{2k+C},
\quad\text{i.e.}\quad
D_\infty\!\big(P_{X_1\mid X_2=x_2}\ \Vert\ \Theta\big)\ \le\ 2k + C.
\]
Thus $\{x_2:\text{$x_2$ is typical}\}\supseteq \{x_2:\ P_{X_2}(x_2)\ge 2^{-C}\Theta(x_2)\}$.

It remains to show that \eqref{eq:typical-x2} fails with probability at most $2^{-C}$ under
$P_{X_2}$. Define the ``light'' set
$L:=\{x_2:\ P_{X_2}(x_2)< 2^{-C}\Theta(x_2)\}$. Then
\[
\Pr[X_2\in L]\ =\ \sum_{x_2\in L} P_{X_2}(x_2)
\ <\ 2^{-C}\sum_{x_2\in L} \Theta(x_2)
\ \le\ 2^{-C}\sum_{x_2\in[q]} \Theta(x_2)
\ =\ 2^{-C}.
\]
Hence $\Pr[X_2\notin L]\ge 1-2^{-C}$, and for all $x_2\notin L$ we have
$D_\infty(P_{X_1\mid X_2=x_2}\Vert \Theta)\le 2k+C$ as shown above. This proves the claim.
\hfill$\Box$

%% file: lowerboundproofoverview.tex
\subsection{Proof overview}
\newcommand{\FL}{f^{t\times t}}
We next sketch the proof of \cref{thm:main} for the special case where $m_A = \Theta(m_B)$; this is enough to get the lower bound on $\Rhat_\eps(f)$. For $f:[q] \times [q] \to \{-1,1\}$, define a new \emph{lifted function} $\FL:[q]^t \times [q]^t \rightarrow [-1,1]$ as follows:
$$\FL(x,y) = \frac{1}{t^2} \sum_{i,j \in [t]} f(x_i, y_j).$$

We show a \emph{lifting theorem}: computing $\FL$ within error $\Omega(1/t)$, needs $\Omega(t \log(1/\Disc^\times(f)))$ bits of communication:

\begin{claim}[Informal; special case of \cref{thm:main}]\label{clm:discinformal}
There exists a constant $c > 0$ such that for any $f : [q] \times [q] \rightarrow \{1,-1\}$, and $t < q^{c}$, computing $\FL$ to error $c /t$ requires $c t \cdot \log(1/\Disc^\times (f))$ bits of communication.
\end{claim}

Solving \jee to within error $\eps$ for $f$ is clearly stronger than computing $\FL$ to error $\eps$: Given a protocol for \jee for $f$, we can solve $\FL$ on inputs $x \in [q]^t, y \in [q]^t$ by taking Alice's distribution to be uniform on the coordiantes of $x$, and Bob's distribution to be uniform on the coordiantes of $y$.

There is a rich body of work on proving lifting theorems in communication complexity; here we will use techniques from \cite{GLMWZ16, KMR21}.

\subsubsection{Lower bound for computing \texorpdfstring{$\IP_n^{t \times t}$}{IP\_n\^\{t x t\}}}
Let us specialize further to the important case where $[q] \equiv \{0,1\}^n$, and $f$ is $\IP_n$. Here, the discrepancy under $\Theta = \Unif(\{0,1\}^n \times \{0,1\}^n)$ is $2^{-n/2}$. The proof for general functions and arbitrary $\Theta$ is similar in spirit but involves several technical challenges. Working with $\IP_n$ lets us illustrate the main ideas more clearly. 

For brevity, let $F = \IP_n^{t \times t}$. 
We first pick a \emph{hard distribution} for $F$. The natural one works: Let $X, Y \in (\{0,1\}^n)^t$ be independent and uniformly at random. By Yao's min-max principle, we can focus on deterministic protocols for computing $F$ under this distribution. 

Now, suppose there is a protocol $\Pi$ that succeeds in computing $F$ to within accuracy $\varepsilon$ with probability at least $1 - \eps^{O(1)}$ and communication at most $C$. Then, by a standard averaging argument, there exists a large rectangle $R \subseteq (\{0,1\}^n)^t \times ((\{0,1\}^n)^t$ on which the variance of $F$ is small. In particular, it follows that there exist distributions $D_A', D_B'$ on $((\{0,1\}^n)^t$ with \emph{min-entropy} \footnote{For any distribution $p$ over a domain $[N]$, min-entropy is $\min_{a \in Supp(p)} - \log_2 p(a)$.} at least $n t - C$ each such that for $X \sim D_A', Y \sim D_B'$,
\begin{equation}\label{eq:varprotocol}
\Var(F(X,Y)) = O(\varepsilon^2).    
\end{equation}

We next give a sketch as to why this is impossible by the following structural statement:
\begin{claim}[Informal]\label{clm:varip}
    Let $X \sim D_A', Y \sim D_B'$ be independently distributed on $(\{0,1\}^n)^t$ with min-entropy at least $nt - C$ for $C < c_0 n t$. Then, $\Var(\IP^{t \times t}(X,Y)) = \Omega(1/t^2)$. 
\end{claim}

\paragraph{The case of ``juntas''} Let $X \sim D_A', Y \sim D_B'$. For intuition, let us first consider the ideal ``junta'' case where $X,Y$ are fixed on $O(C/n)$ coordiantes and the rest are uniform. This is commensurate with the min-entropy deficit: fixing each coordinate (an element of $\{0,1\}^n$) costs $n$ bits of entropy. 

Let $I, J \subseteq [t]$ with $|I|, |J| = O(C/n)$ be such that $X_i$ are fixed for $i \in I$, and the same holds for $Y_j, j \in J$. The remaining coordinates are uniformly random. Now, 
\begin{align}\label{eq:Fdecomp}
    t^2 F(X,Y) &= \sum_{i,j \in [t]} \IP(X_i, Y_j) \nonumber\\
    &= \underbrace{\sum_{i \in I, j \in I} \IP(X_i, Y_j)}_{:= Z_0} + \underbrace{\underbrace{\sum_{i \in I, j \notin J} \IP(X_i, Y_j)}_{:=Z_{1,B}} + \underbrace{\sum_{i \notin I, j \in J} \IP(X_i, Y_j)}_{:=Z_{1,A}}}_{:=Z_1} + \underbrace{\sum_{i \notin I, j \notin J} \IP(X_i, Y_j)}_{:= Z_2}.
\end{align}

Naturally, under our assumptions, $Z_0$ is constant. Further, as $X_i, Y_j$ are independent uniformly random over $\{0,1\}^n$, a simple calculation shows that 
\begin{equation}\label{eq:ipsketch1}
\E[Z_1] = \E[Z_2] = \E[Z_1 Z_2] = 0;\;\; \E[Z_2^2] = (t - |I|) (t - |J|). 
\end{equation}

This in particular implies that 
$$\Var(F(X,Y)) \geq \frac{1}{t^4} (t-|I|) (t - |J|) = \Omega(1/t^2),$$
proving the required conclusion from \cref{clm:varip}.


\paragraph{The case of \emph{blockwise-dense} distributions} The above argument had a severe restriction that it only applied to $D_A'$, $D_B'$ of a very special structure: all but $O(C/n)$ coordinates are fixed and rest are uniformly random. 

Fortunately, we can vaguely reduce to the above case by using the notion of \emph{blockwise-dense} distributions studied in the context of lifting theorems (cf.~\cite{GLMWZ16}). We will state a more general and formal result later; for now let us describe what we need for $\IP$ informally: \cite{GLMWZ16} show that any distribution $D_A'$ on $(\{0,1\}^n)^t$ of min-entropy at least $nt - C$ can be written as a convex-combination of distributions with the following properties:
$$D_A' = \sum_{\ell} \lambda_\ell D_A^{(\ell)} + \lambda_{err} D_A^{err},$$

\begin{enumerate}
    \item Each $D_A^{(\ell)}$ is $0.9$-\emph{conjunctive blockwise-dense} (CBD): there exists a set $I_\ell \subseteq [n]$, $|I_\ell| = O(C/n)$ such that $D_A^{(\ell)}$ is fixed in the coordiantes in $I_\ell$; for any $S \subseteq [t]\setminus I_\ell$, the min-entropy of $D_A^{(\ell)}$ restricted to coordinates in $S$ is at least $0.9 |S| n$. 
    \item $\lambda_{err} < 2^{-\Omega(nt - C)}$.
\end{enumerate}

We can do a similar decomposition for $D_B'$. By the law of total variance, to lower bound the variance of $F(X,Y)$ for $X \sim D_A', Y \sim D_B'$, it suffices to lower bound the variance of $F(X,Y)$ for any $X \sim D_A^{(\ell)}$, and $Y \sim D_B^{(\ell')}$. The error terms corresponding to $\lambda_{err}$ are negligible. 

Let us focus on one such term: We have $X \sim D_A^{(\ell)}$, $Y \sim D_B^{(\ell')}$ where $X,Y$ are CBD-distributions as defined above. For brevity, let $I \equiv I_\ell, J \equiv J_\ell$ be the fixed coordinates of size $O(C/n)$. We can once again decompoe $t^2 F(X,Y)$ as a sum of random variables $Z_0 + Z_1 + Z_2$ as in \cref{eq:Fdecomp}. 

As before, $Z_0$ is a constant. However, we no longer even have $\E[Z_1] = 0$ or $\E[Z_2] = 0$. Nevertheless, we can now use properties of $\IP_n$ to argue that $Var(Z_2) = \Omega(t^2)$, and $Z_1$ cannot destroy this variance. 

\paragraph{Variance of $Z_2$} We exploit the well-known fact that $\IP_n$ has the following  discrepancy property (\cite{ChorGoldreich})\footnote{In another parlance, $\IP_n$ is a \emph{two-source extractor.}}: For $U,V \sim \{0,1\}^d$ of min-entropy at least $d - k_1, d - k_2$ respectively, we have 
\begin{equation}\label{eq:ipext}
\E[\IP(U,V)] \leq 2^{k_1 + k_2 - d}    
\end{equation}

In our case, the individual terms in $Z_2$, $\IP(X_i, Y_j)$ for $i \notin I, j \notin J$, all have min-entropy at least $0.9 n$. Thus, 
\begin{equation} \label{eq:z2exp}
\E[Z_2] = t^2 2^{-\Omega(n)}    
\end{equation}

Further, to calculate the variance, we have

\begin{equation}\label{eq:z2var}
\E[Z_2^2] = (t - |I|) (t- |J|) + \sum_{(i,j) \neq (i',j') \in \overline{I} \times \overline{J}} \E\left[\IP(X_i, Y_j) \cdot \IP(X_{i'}, Y_{j'})\right].
\end{equation}

Now note that for $(i,j) \neq (i', j')$ we must have either $i \neq i'$ or $j \neq j'$. Therefore, as $X, Y$ are CBD-distributions, the total entropy in $((X_i, X_{i'})), (Y_j, Y_{j'})$ is at least $3 \cdot (0.9 n)$. Thus, applying \cref{eq:ipext}, we immediately get
the following bound for all cross-terms: if $(i,j) \neq (i',j') \in \overline{I} \times \overline{J}$,  

\begin{equation}\label{eq:z2cross}
\E\left[\IP(X_i, Y_j) \cdot \IP(X_{i'}, Y_{j'})\right] = \E[\IP_{2n}((X_i, X_{i'}), (Y_j, Y_{j'}))] \leq 2^{.3 n - 2n} = 2^{-\Omega(n)}.
\end{equation}

Thus, we can conclude that 
$$\Var(Z_2) \geq \Omega(t^2) - t^2 2^{-\Omega(n)}.$$

\paragraph{Interference from partially fixed blocks} We are now almost done but for the fact that we don't have $\E[(Z_1 - \E[Z_1]) Z_2] = 0$ as in the easy junta case. Fortunately, we can handle this in one of two ways. A more detailed calculation can be used to directly show that the cross-terms arising in $Z_1 Z_2$ have small expectations $2^{-\Omega(n)}$ by \cref{eq:ipext} once again. 

However, with a view towards extending to the case of general functions, we present an alternate, cleaner way. Note that in \cref{eq:Fdecomp}, we can further decompose $Z_1 = Z_{1,B} + Z_{1,A}$. As the coordinates of $X$ in $I$ are fixed, $Z_{1,B}$ only depends on $Y$; similarly $Z_{1,A}$ is a function of $X$ alone. Thus, we can condition on a typical value of $Z_{1,A}, Z_{1,B}$ while mainting independence of $X, Y$. 

Concretely, say, we condition on values of $Z_{1,A}, Z_{1,B}$ that occur with probability at least $2^{-0.1n}$. Under this conditioning, the distributions of $X, Y$ will $0.8$-CBD-distributions and are still independent. Thus, we can use the same exact calculations for lower bounding the variance of $Z_2$ (and hence of $Z$ as $Z_0, Z_1$ are now fixed) by following Equations \ref{eq:z2exp},\ref{eq:z2var},\ref{eq:z2cross} (with a slightly different constant in $2^{-\Omega(n)}$). 

\subsubsection{Lower bound for computing \texorpdfstring{$f^{t \times t}$}{f\^\{t x t\}} from discrepancy} Let us now return to the proof of \cref{thm:main}, specifically the special case as alluded to in \cref{clm:discinformal}. 

Let $\Theta$ be a product distribution under which $\Disc_{\Theta}(f)$ is small. As before, the idea is to use the following hard distribution: $X \sim [q]^t$, where each $X_i$ is independently drawn from $\Theta$ and the same for $Y$. By Yao's min-max principle, suppose there is a deterministic protocol $\Pi$ to compute $\FL$ within accuracy $\varepsilon = c_0/t$ with communication at most $C < c_0 t \log(1/\Disc^\times(f))$. 

The first difference from $\IP$ is that $\Theta$ need not be an uniform distribution \footnote{It is possible to reduce the lower-bound proof to the uniform distribution by modifying the domain $[q]$ to a bigger domain $[q']$, so that the non-uniform distirbution on $[q]$ corresponds to a uniform distribution over $[q']$.  While this simplifies the distribution $\Theta$, all the remaining technicalities including the need for max-divergence measure persist.  We encourage the reader to think of $\Theta$ as a uniform distribution over $[q]$, since all of the technical details are exactly the same for an arbitrary distribution $\Theta$.}. 
Thus, when moving to a large \emph{rectangle} over $[q] \times [q]$, we should use measure of \emph{largeness} relative to $\Theta$. We do so by looking at \emph{max-divergence} as the measure: For distributions $P, Q$ on the same space define
 $$ D_\infty(P\Vert Q)\ :=\ \log\max_{z}\frac{P(z)}{Q(z)}.$$

Now, a standard averaging argument shows that under the guarantees of the protocol, there exist independent distributions $D_A', D_B'$  on $[q]$ with $D_\infty(D_A' \| \Theta) \leq C, D_\infty(D_B' \| \Theta) \leq C$ such that 
$$\Var_{X \sim D_A', Y \sim D_B'}[\FL(X,Y)] = O(\eps^2).$$

Just as before, we will show that this is impossible by showing an analog of \cref{clm:varip}:

\begin{claim}[Informal]\label{clm:varipf}
    Let $X \sim D_A', Y \sim D_B'$ be independently distributed on $[q]$ with $D_\infty(D_A' \| \Theta)$, $D_\infty(D_B' \| \Theta) < C$ for $C < c_0 t \cdot \log (1/\Disc^\times(f))$. Then, $\Var(\FL(X,Y)) = \Omega(1/t^2)$. 
\end{claim}

The proof of the above follows the same high-level approach as for $\IP$.
Let $\Delta = \log(1/Disc^\times(f))$. We first show a simple extension of the CBD decomposition result of \cite{GLMWZ16} that applies to arbitrary distributions $\Theta$ on $[q]$ instead of the uniform distribution. This reduces us to study the case where $X, Y$ have the following property: $X,Y$ are fixed in some sets $I , J \subset [t]$ of size $O(C/\Delta)$, and for any $S \subseteq \overline{I}$, we have 
$$D_\infty(X_{|S} \| \Theta^S) \leq 0.1 \Delta|S|,$$

where $X_{|S}$ denotes $X$ restricted to the coordinates in $S$, and $\Theta^S$ denotes the distribution on $[q]^S$ where each coordinate is drawn independently from $\Theta$. A similar condition holds for $Y$. 

Given the above, we can once again use a decomposition of $\FL(X,Y)$ as in \cref{eq:Fdecomp} by writing $t^2 F(X,Y) = Z_0 + Z_1 + Z_2$. 

The key step is to lower bound the variance of $Z_2$. We do so by first showing an analog of \cref{eq:ipext} for any $f$ with small discrepancy (see \cref{lem:disc-Dinfty}): For any random variables $U,V$ over $[q]$ with $D_\infty(U \| \Theta) \leq k_1$,  $D_\infty(V \| \Theta) \leq k_2$, we have 
$$\E[f(U,V)] \leq 2^{k_1 + k_2 - \Delta}.$$

Next, to bound $\E[Z_2^2]$ we need an analog of \cref{eq:z2cross}. Unfortunately, this is not as clean as for $\IP$ and we have to consider different cases. Nevertheless, we can show similar \emph{tensoring} inequalities for discrepancy (see \cref{lem:pairpair-from-disc},  \cref{lem:pair-robust}): if $(i,i') \neq (j,j') \in \overline{I} \times \overline{J}$, then
$$\E[f(X_i, X_{i'}) f(Y_j, Y_{j'})] \leq 2^{-\Omega(\Delta)}. $$
We can then finish the argument by substituting the above equations in place of \cref{eq:ipext,eq:z2exp,,eq:z2var,eq:z2cross} and conditioning on $Z_{1,A}, Z_{1,B}$ as for $\IP$. 



%% file: onewayparities.tex
\section{One-way communication lower bound for parities}

\begin{theorem}\label{th:iponeway}
There is an absolute constant $c_0>0$ such that the following holds. Let $2^{-c_0 n} <\varepsilon<\tfrac13$ and $2^{- n}<\delta<\tfrac12$. Consider $f(x,y)=(-1)^{\langle x,y\rangle}$ on $\{0,1\}^n\times\{0,1\}^n$. Any one-way public-coin protocol that, on every pair of inputs $(p,q)$ (Alice has $p$, Bob has $q$), outputs an estimate of $\E_{x\sim p,\,y\sim q}[f(x,y)]$ within additive error~$\varepsilon$ with success probability at least $1-\delta$, must send a message of length
\[
\Delta \;=\; \Omega\!\Big(\frac{n\,\log_2(1/4\delta(1-\delta))}{\varepsilon^2}\Big).
\]
\end{theorem}

Note that there is a trivial upper-bound of $O(n \log(1/\delta)/\eps^2)$: Alice can sample $\log(1/\delta)/\eps^2$ points according to $p$, and send them across to Bob. Bob can then compute the average value in batches of size $O(1/\eps^2)$ and take the median to get an estimate within error $\eps$ and $1-\delta$ probability of success. The above bound says that this protocol is essentially optimal. Further, we can assume $\delta > 2^{-c_0 n}$ (by allowing a worse constant hidden in the $\Omega$).
\begin{proof}
We work with the \emph{parity spectrum} $\widetilde p(\beta):=\E_{x\sim p}[(-1)^{\langle \beta,x\rangle}]$. Note that for any $q$,
\[
\E_{p,q}[f] \;=\; \E_{y\sim q}\,\widetilde p(y),
\]
so in particular if $q=\delta_\beta$ then the protocol returns an estimate of $\widetilde p(\beta)$.

\paragraph{Hard distribution of Alice's input.}
Let $N:=\alpha n/\varepsilon^2$ with a constant $\alpha\in(0,\tfrac12)$ to be set shortly. Draw $S=(a_1,\ldots,a_N)$ with the $a_i\in\F_2^n$ i.i.d.\ uniform (collisions occur with prob.\ $o(1)$ and can be ignored). We will show how to use the protocol for estimating $\tilde{p}$ to encode $S$ in a compressed manner. This will lead to a lower bound on the cost of the protocol. 

Let $y\in\{\pm1\}^N$ be i.i.d.\ Rademachers, independent of~$S$. Define
\[
\phi_{S,y}(x)\ :=\ \prod_{i=1}^N\big(1+\varepsilon\,y_i\,(-1)^{\langle a_i,x\rangle}\big)\,,
\qquad
p_{S,y}(x)\ :=\ \frac{2^{-n}\,\phi_{S,y}(x)}{\E_{x\sim\mu}[\phi_{S,y}(x)]}\,,
\]
where $\mu$ is the uniform measure on $\{0,1\}^n$. Since $\varepsilon<1$, $\phi_{S,y}(x)\ge0$ and $p_{S,y}$ is a valid distribution.

Write the uniform–Fourier transform $\widehat\phi(\beta):=\E_{x\sim\mu}[\phi_{S,y}(x)(-1)^{\langle\beta,x\rangle}]$. Then
\[
\widetilde p_{S,y}(\beta)\ =\ \E_{x\sim p_{S,y}}[(-1)^{\langle\beta,x\rangle}]
\ =\ \frac{\widehat\phi_{S,y}(\beta)}{\E_{x\sim\mu}[\phi_{S,y}(x)]}.
\]

\paragraph{Second-moment bounds.}

Expanding $\phi$,
\[
\phi_{S,y}(x)=\sum_{I\subseteq[N]}\varepsilon^{|I|}\prod_{i\in I} y_i\cdot (-1)^{\langle \sum_{i\in I} a_i,\,x\rangle}.
\]
Let
\[
\mathrm{err}_0(S,y):=\E_{x\sim\mu}[\phi_{S,y}(x)]-1
=\sum_{\emptyset\neq I\subseteq[N]} \varepsilon^{|I|}\prod_{i\in I} y_i\cdot \mathbf 1\!\left[\sum_{i\in I} a_i=0\right].
\]
A standard computation using $\E_y[\prod_{i\in I} y_i\prod_{i\in J} y_i]=\mathbf 1[I=J]$ and
$\Pr_S[\sum_{i\in I} a_i=0]=2^{-n}$ for $I\neq\emptyset$ gives
\[
\E_{S,y}\big[\mathrm{err}_0(S,y)^2\big]
= \big((1+\varepsilon^2)^N-1\big)\,2^{-n}\ \le\ e^{\varepsilon^2 N}\,2^{-n}.
\]
Similarly, for any fixed $j\in[N]$ let
\[
\mathrm{err}_j(S,y):=\widehat\phi_{S,y}(a_j)-\varepsilon y_j.
\]
By the same second-moment expansion (omitted for brevity), there is a universal $C>0$ such that
\[
\E_{S,y}\big[\mathrm{err}_j(S,y)^2\big]\ \le\ C\,e^{\varepsilon^2 N}\,2^{-n}.
\]
Therefore, choosing any $\gamma\in(0,\tfrac12-\alpha/\ln 2)$ and applying Markov/Chebyshev plus a union bound over $j\in[N]$, with probability at least $1-2^{-\gamma n}$ over $(S,y)$ we have simultaneously
\[
\big|\mathrm{err}_0(S,y)\big|\le 2^{-\gamma n},\qquad 
\big|\mathrm{err}_j(S,y)\big|\le \varepsilon/20\quad(\forall j\in[N]).
\]
On this event,
\[
\widetilde p_{S,y}(a_j)
=\frac{\varepsilon y_j+\mathrm{err}_j}{1+\mathrm{err}_0}
=\varepsilon y_j \pm \varepsilon/10\qquad(\forall j\in[N]).
\]

Next, for any $\beta$, 
$$ \widehat{\phi_{S,y}}(\beta) = \sum_{I \subseteq [N]} \epsilon^{|I|} \cdot \prod_{i \in I} y_i \cdot 1(\sum_{i \in I} a_i = \beta).$$
Thus, 
\begin{align*}
    \E_{S,y}[\widehat{\phi_{S,y}}(\beta)^2] &= \sum_{\ell=1}^N \epsilon^{2\ell} \binom{N}{\ell} 2^{-n} \leq e^{\eps^2 N} 2^{-n}.
\end{align*}

Hence, by Markov,
\[
\E_{S,y}\big[\,|\{\beta:\ |\widehat\phi_{S,y}(\beta)|>\varepsilon/20\}|\,\big]
\ \le\ M\ :=\ C_2\,\frac{e^{\varepsilon^2 N}}{\varepsilon^2}.
\]
With probability at least $9/10$ over $(S,y)$,
\begin{equation}\label{eq:bad}
|\mathrm{Big}(S,y)|\ :=\ |\{\beta:\ |\widehat\phi_{S,y}(\beta)|>\varepsilon/20\}|
\ \le\ 10M.
\end{equation}
Since $|\mathrm{err}_0|\le 2^{-\gamma n}$, the same bound (with constants adjusted) holds if we define $\mathrm{Big}$ via $|\widetilde p_{S,y}(\beta)|>\varepsilon/10$.

\paragraph{Decoding $S$ from one message.}
Fix $(S,y)$ in the good event above. Consider the one-way public-coin protocol run with Alice’s input $p_{S,y}$. Let $m$ be Alice’s message (a function of $p_{S,y}$ and the public coins). For each $\beta\in\{0,1\}^n$, Bob can, using fresh private coins and the \emph{same} message $m$, run the protocol with $q=\delta_\beta$ to obtain an estimate $\widetilde p_{S,y}(\beta)\pm \varepsilon/20$; define
\[
\widetilde B(S,y)\ :=\ \big\{\beta:\ |\text{Bob's estimate of }\widetilde p_{S,y}(\beta)|>\varepsilon/10\big\}.
\]
By the per-input success guarantee and linearity of expectation,
\[
\E[\,|S\cap \widetilde B(S,y)|\,]\ \ge\ N(1-\delta),\qquad
\E[\,|\widetilde B(S,y)\setminus \mathrm{Big}(S,y)|\,]\ \le\ \delta\,2^n.
\]
Markov’s inequality implies that, with probability at least $4/5$ over the protocol coins,
\[
|S\cap \widetilde B(S,y)|\ \ge\ N- O(N\delta),\qquad
|\widetilde B(S,y)|\ \le\ |\mathrm{Big}(S,y)|+O(\delta\,2^n).
\]

First, assume that $\delta < \eps^2/n$. For any $1/2 > \delta > \eps^2/n$ (example $\delta = \Omega(1)$), we can reduce to the other case by amplifying the success probability to at least $1 - n/\eps^2$ by repeating the protocol $O(1) (\log(n/\eps^2)/\log(1/4\delta(1-\delta)))$ times.

Choose the constant in the upper bound on $N$ so that $N\delta\le 1/10$, hence $S\subseteq \widetilde B(S,y)$ w.h.p., i.e., no “misses.” Combining with \eqref{eq:bad} we get, with probability at least $4/5$ over the protocol coins,
\[
|\widetilde B(S,y)|\ \le\ 10M + 10\delta\,2^n.
\]
Given $m$ and the public coins, Alice can finish identifying $S$ to Bob by sending the index of $S$ as an $N$-subset of $\widetilde B(S,y)$, using an additional
\[
\log_2\binom{|\widetilde B(S,y)|}{N}\ \le\ \log_2\binom{10M+10\delta\,2^n}{N}
\]
bits. Thus, with probability at least $4/5$, $S$ is recovered from a total of
\[
\Delta\ +\ \log_2\binom{10M+10\delta\,2^n}{N}
\]
bits.

 We can now apply \cref{clm:encdec} with $A = 2^n, K = N$ and $\tau = 4/5$. Thus, we must have 

\[
\Delta\ \ge\ \log_2\binom{2^n}{N}\ -\ \log_2\binom{10M+10\delta\,2^n}{N}\ -\ O(1).
\]

\paragraph{Parameter choice and conclusion.}
Take $N=\alpha n/\varepsilon^2$ with
\[
0<\alpha<\min\!\big\{\tfrac12,\ (1-c_0)\ln 2\big\},
\]
so that
\[
M \;=\; \Theta\!\Big(\frac{e^{\varepsilon^2 N}}{\varepsilon^2}\Big)
\;=\; \Theta\!\Big(\frac{e^{\alpha n}}{\varepsilon^2}\Big)
\;\ll\; \delta\,2^n,
\]
as from our hypothesis we have $\delta > 2 ^{-c_0 n}$.

Hence $10M+10\delta 2^n= (10+o(1))\,\delta\,2^n$. Using Stirling’s approximation in the “small-$N$” regime,
\[
\log_2\binom{M_1}{N}-\log_2\binom{M_2}{N}
= N\log_2\!\frac{M_1}{M_2}\ \pm o(N)\qquad(M_1,M_2\gg N),
\]
we get
\[
\Delta \ \ge\ (1-o(1))\,N\log_2\frac{2^n}{10 \delta\,2^n}\ =\ (1-o(1))\,N\log_2\frac{1}{10 \delta}
\ =\ \Omega\!\Big(\frac{n\,\log_2(1/\delta)}{\varepsilon^2}\Big).
\]

Finally, if $\delta\in(\varepsilon^2/n,1/2)$, amplify to $\delta'=\Theta(\varepsilon^2/n)$ via the median of $t= \Theta\!\big(\log(n/\varepsilon^2)\big/\log(1/(4\delta(1-\delta)))\big)$ independent runs, apply the bound at $\delta'$, and divide by $t$ to obtain the stated dependence on $\log(1/\delta)$.
\end{proof}

The following standard fact follows for example from Fano's inequality:
\begin{claim}[Decoding Lemma]\label{clm:encdec}
Suppose we have encoder, decoder functions $E:2^{[A]} \rightarrow \{0,1\}^M$, and $D:\{0,1\}^M \rightarrow 2^{[A]}$ such that for a uniformly random set $S$ of $[A]$ of size $K$, we have $\Pr[D(E(S)) = S] \geq \tau$, then $M \geq \log_2 \binom{A}{K} + \log_2 \tau$. 
\end{claim}

%% file: related.tex
\section{Further related work}
\label{sec:related}

\paragraph{Relation to direct sum}

The canonical sampling strategy for \jee involves Alice and Bob computing $\frac{1}{t} \sum_{i \in [t]} f(x_i,y_i)$ for independent samples $x_1,\ldots,x_t$ and $y_1,\ldots,y_t$ drawn from their distributions.
This problem is related to the {\it direct sum} and {\it direct product} theorems in communication complexity.
The goal of direct sum theorems is to show that if computing a function $f$ takes $C$ bits of communication, then computing $t$ independent instances simultaneously should require $t \cdot C$ bits of communication.
Direct product theorems on the other hand ask what is the maximum success probability of a protocol using less than $t \cdot C$ bits of communication, computes all of the values $f(x_i,y_i)$ accurately.
We refer the reader to \cite{BRWY13,BBCR13,BR11} and references therein.

\paragraph{Query to communication lifting}
Our main lower bound result relating \jee of a function $f$ to its discrepancy (Theorem \ref{thm:main}) borrows ideas from query-to-communication lifting results.  Specifically, the notion of conjunctive block-wise dense (cbd) decomposition of distributions introduced in \cite{GLMWZ16} play a major role in our proof.

Hadar, Liu, Polyanskiy and Shayevitz \cite{HLPS19} study the communication complexity of estimating correlations.  Here Alice and Bob receive copies of $\rho$-correlated bits (or Gaussians), and their goal is to produce an estimate $\hat{\rho}$ of the correlation $\rho$.  This work shows that computing the correlation within an error $\varepsilon$ requires $\Omega(1/\varepsilon^2)$ bits of communication, and also recover a tight lower bound for Gap-Hamming problem as a consequence.  In our setting of \jee, Alice and Bob receive independent random variables and thus seems markedly different in flavor.

\paragraph{Join estimation.}

Given two relations $\mathcal{A} \in [m] \times [n]$ and $\mathcal{B} \in [n] \times [p]$, 
        the {\it composition} and {natural join} of the two relations is defined as,
        \begin{align*}
        \text{(Composition)\ \ \ }  \mathcal{A} \circ \mathcal{B} & = \{ (i,j) | \exists \ell : (i,\ell) \in \mathcal{A} \text{ and } (\ell,j) \in \mathcal{B}\} \\
        \text{(Natural join)\ \ \ } \mathcal{A} \bowtie \mathcal{B} & = \{ (i,\ell,j) | \exists \ell : (i,\ell) \in \mathcal{A} \text{ and } (\ell,j) \in \mathcal{B}\} 
        \end{align*}
        
        The communication complexity of computing these joins and their the sizes was studied in \cite{van2015communication}.  The problem of estimating $| \mathcal{A} \circ \mathcal{B}|$ to within additive error $\eps |\mathcal{A}||\mathcal{B}|$ is an instance of \jee for $\EQ$. Indeed this is one of the applications of CountMinSketch \cite[Corollary 1]{CormodeMuthukrishnan}. Similarly, estimating the size of the composition join is an instance of \jee for the Set-disjointness function. We refer the reader to related work for derivations of these relations. 

        Given $\mathcal{A} \subseteq [m] \times [n]$, define a probability distribution $\mu_A$ over $[n]$ as, the marginal distribution of $\ell$ for a random tuple $(i,\ell) \in \mathcal{A}$.  Define $\mu_B$ over $[n]$ as the marginal distribution of $\ell'$ for a random tuple $(\ell',k) \in \mathcal{B}$. The size of the {\it natural join} is given by,
        \[ \frac{|\mathcal{A} \bowtie \mathcal{B}|}{\mathcal{A} \cdot \mathcal{B}} = \E_{\ell \sim \mu_A, \ell' \sim \mu_B} \left[ \mathsf{EQ}(\ell,\ell') \right]\]

        Towards formulating composition, define the sets $A_i = \{ \ell | (i,\ell) \in \mathcal{A} \}$ and $B_j = \{ \ell | (\ell,j) \in \mathcal{B} \}$ for each $i \in [m], j \in [p]$.  The size of the {\it composition} can be written as,
        
        \[ \frac{|\mathcal{A} \circ \mathcal{B}|}{m \cdot p} = \E_{A_i, B_j} \left[ \overline{\mathsf{DISJ}}(A_i,B_j) \right]\]
        where $\overline{\mathsf{DISJ}}(A,B) = 1[ A \cap B \neq \emptyset]$ is the negation of the familiar disjointness function, and the expectation is over $A_i$ (resp. $B_j$) drawn uniformly from $\{A_1,\ldots,A_m\}$ (from $\{B_1,\ldots,B_p\}$).

        \paragraph{Acknowledgments} This work benefited from conversations with numerous colleagues, whom we would like to thank: Abhishek Shetty, Charlotte Peale, Vinod Raman, Kunal Talwar and Pravesh Kothari. Mihir Singhal was partially supported by grant NSF CCF-2311648 and by Venkatesan Guruswami's Simons Investigator award. 

%% file: appendix.tex
\section{Omitted proofs}

\subsection{Proof of \texorpdfstring{\cref{lem:convex-lipschitz-1d}}{Fact \ref{lem:convex-lipschitz-1d}}} \label{sec:appendix-convex-lipschitz-1d}

Since $f$ is convex and $1$-Lipschitz on $[0,1]$, its right derivative $f'_+(x)$ exists for every $x\in[0,1)$, is nondecreasing, right-continuous, and satisfies $f'_+(x)\in[-1,1]$. Define $F:\mathbb{R}\to[0,1]$ by
\[
F(x)=
\begin{cases}
0,& x<0,\\[2pt]
\dfrac{1+f'_+(x)}{2},& 0\le x<1,\\[6pt]
1,& x\ge 1.
\end{cases}
\]
Then $F$ is nondecreasing, right-continuous, with $\lim_{x\to-\infty}F(x)=0$ and $\lim_{x\to+\infty}F(x)=1$, hence is the CDF of a probability measure $D$ supported on $[0,1]$. Then, let
\[
g(x) = \E_{z \sim D}[|x-z|].
\]

\begin{claim} \label{claim:right-deriv}
For every $x\in[0, 1)$, the right derivative of $g$ exists and
\[
g'_+(x)=2F(x)-1.
\]
\end{claim}
\begin{proof} For $h>0$,
\[
\frac{g(x+h)-g(x)}{h}
=\E_{z \sim D}\left[\frac{|x+h-z|-|x-z|}{h}\right].
\]
Let $\phi_h(z)$ be the term inside the expectation. Then,
we have that $-1\le \phi_h(z)\le 1$ for all $z$ and, as $h \downarrow 0$,
\[
\phi_h(z)\to \ind{z\le x}-\ind{z>x}=2\ind{z\le x}-1
\]
By dominated convergence, 
\[
g'_+(x)=\lim_{h\downarrow 0}\E_{z \sim D}[\phi_h(z)]
=\E_{z \sim D}\bigl[2\ind{z\le x}-1\bigr]
=2F(x)-1,
\]
proving the claim.
\end{proof}

With our choice of $F$, \cref{claim:right-deriv} gives for every $x\in[0,1)$
\[
g'_+(x)=2F(x)-1 = f'_+(x).
\]
Both $f$ and $g$ are $1$-Lipschitz on $[0,1]$, so $f'$ and $g'$ exist almost everywhere and equal $f'_+$ and $g'_+$ almost everywhere. It follows that $(f-g)'(x)=0$ almost everywhere,
so $f-g$ is constant on $[0,1]$. Writing $c:=f(0)-g(0)=f(0)-\E[z]$ we obtain
\[
f(x)=c+g(x)=c+\E[|x-z|]\qquad\text{for all }x\in[0,1].
\]
Note that both $f$ and $g$ are bounded in $[-1,1]$, so $c\in[-2,2]$.

\subsection{Proof of \texorpdfstring{\cref{lem:bounded-lower-deriv}}{Fact \ref{lem:bounded-lower-deriv}}} \label{sec:appendix-bounded-lower-deriv}

We argue by induction on $\ell$ going backwards from $k$; assume that the $(\ell-1)$-th mixed partials are bounded by some constant $C_{\ell-1}$.
Fix $0\le i\le \ell$ and set
\[
g(x,y)\coloneqq \del_x^{\,i}\del_y^{\,\ell-i} f(x,y).
\]
Suppose, for contradiction, that $g$ is very large somewhere; without loss of generality there exists $(x_0,y_0)$ with $g(x_0,y_0)=M>0$. Since
\[
\del_x g=\del_x^{\,i+1}\del_y^{\,\ell-i} f,\qquad
\del_y g=\del_x^{\,i}\del_y^{\,\ell-i+1} f
\]
are $(\ell+1)$-st partials, the induction hypothesis gives $|\del_x g|,|\del_y g|\le C_{\ell+1}$. Hence $g$ is $C_{\ell+1}$–Lipschitz in each coordinate, so for every $(x,y)\in[0,1]^2$,
\[
g(x,y)\ge M - C_{\ell+1}|x-x_0| - C_{\ell+1}|y-y_0|\ge M-2C_{\ell+1}.
\]
Take $M = 2C_{\ell+1}+K_\ell$ (where $K_\ell$ is defined below); then $g \ge K_\ell$ on the whole square.

Now define the nonnegative weight function $W$ by
\[
\phi_i(t)=t^i(1-t)^i,\qquad \psi_{\ell-i}(t)=t^{\ell-i}(1-t)^{\ell-i},\qquad
W(x,y)=\phi_i(x)\psi_{\ell-i}(y).
\]
All derivatives of $\phi_i$ up to order $i-1$ (respectively of $\psi_{\ell-i}$ up to order $\ell-i-1$) vanish at $0$ and $1$. Integrating by parts $i$ times in $x$ and $\ell-i$ times in $y$ (there are no boundary terms since the derivatives vanish on the edge of the square) gives
\[
\iint_{[0,1]^2} g\,W
= (-1)^\ell \iint_{[0,1]^2} f \,\del_x^{\,i}\del_y^{\,\ell-i} W.
\]
Using $|f|\le 1$ we obtain
\[
K_\ell \underbrace{\iint W}_{\displaystyle A_{i,\ell}}
\;\le\; \iint gW
\;=\; \Big|\iint f\,\del_x^{\,i}\del_y^{\,\ell-i}W\Big|
\;\le\; \underbrace{\iint \big|\del_x^{\,i}\del_y^{\,\ell-i}W\big|}_{\displaystyle B_{i,\ell}}.
\]
Here $A_{i,\ell}>0$ and 
$B_{i,\ell}<\infty$, both depending only on $\ell$ (and $i$). Thus, picking $K_\ell$ to be larger than all ratios $B_{i,\ell}/A_{i,\ell}$ for $0\le i\le \ell$ yields a contradiction.
Let $K_\ell\coloneqq \max_{0\le i\le \ell} \frac{B_{i,\ell}}{A_{i,\ell}}$. The previous display contradicts $m>K_\ell$. Therefore, in summary, we have
\[
g(x, y) \le 2C_{\ell+1}+K_\ell,
\]
for all $i$, completing the induction.